\documentclass[12pt]{article}
\usepackage{amsmath, amsthm, amssymb, amsthm, outlines,parskip,mathtools}
\usepackage{graphicx,subcaption,multirow,booktabs}
\usepackage[capitalise]{cleveref}
\usepackage{enumerate}
\usepackage{natbib}
\usepackage{url} 
\usepackage{shortcuts}
\usepackage{xcolor}
\usepackage{colortbl}
\usepackage{tikz}
\usepackage{diagbox}
\usetikzlibrary{positioning}
\newcommand{\blind}{1}

\newcommand{\norm}[1]{\left\|{#1}\right\|}
\newcommand{\En}[1]{\hat\E_N\bracks{{#1}}}
\newcommand{\pre}{\mathtt{pre}}
\newcommand{\post}{\mathtt{post}}
\newcommand{\C}{\mathcal C}
\newcommand{\T}{\mathcal T}
\newcommand{\mb}[1]{\mathbf{{#1}}}
\newcommand{\HR}{\mathrm{HR}}
\newcommand{\VR}{\mathrm{VR}}
\newcommand{\fac}{\mathrm{F}}

\newcommand{\ATT}{\mathrm{ACT}}
\newcommand{\W}{\mathcal{W}}
\newcommand{\did}{\mathrm{DID}}
\newcommand\smallO{
  \mathchoice
    {{\scriptstyle\mathcal{O}}}
    {{\scriptstyle\mathcal{O}}}
    {{\scriptscriptstyle\mathcal{O}}}
    {\scalebox{.6}{$\scriptscriptstyle\mathcal{O}$}}
  }

\newtheorem{theorem}{Theorem}
\newtheorem{corollary}{Corollary}
\newtheorem{lemma}{Lemma}
\newtheorem{proposition}{Proposition}
\newtheorem{assumption}{Assumption}
\newtheorem{definition}{Definition}
\theoremstyle{remark}

\newcommand*{\vertbar}{\rule[-1ex]{0.5pt}{2.5ex}}
\newcommand*{\horzbar}{\rule[.5ex]{5ex}{0.5pt}}

\def\spacingset#1{\renewcommand{\baselinestretch}%
{#1}\small\normalsize} \spacingset{1}

\definecolor{Gray}{gray}{0.85}

\begin{document}


\if1\blind
{
\newcommand*\samethanks[1][\value{footnote}]{\footnotemark[#1]}
  \title{\bf Controlling for Unmeasured Confounding in Panel Data Using Minimal Bridge Functions: \\ From Two-Way Fixed Effects to Factor Models}
  \author{Guido Imbens\thanks{Alphabetical order.}$^{~a}$\qquad Nathan Kallus\samethanks$^{~b}$\qquad Xiaojie Mao\samethanks$^{~c}$ \hspace{.2cm}\\
    $^{a}$Graduate School of Business, 
    Stanford University,\\
    $^{b}$School of Operations Research and Information Engineering, 
    Cornell University \\
    $^{c}$Department of Management Science and Engineering, 
    Tsinghua University}
    \date{}
  \maketitle
} \fi

\if0\blind
{
  \bigskip
  \bigskip
  \bigskip
  \begin{center}
    {\LARGE\bf Title}
\end{center}
  \medskip
} \fi

\bigskip
\begin{abstract}
We develop a new approach for identifying and estimating average causal effects in panel data under a linear factor model with unmeasured confounders. 
Compared to other methods tackling factor models such as synthetic controls and matrix completion, our method does not require the number of time periods to grow infinitely.
Instead, we draw inspiration from the two-way fixed effect model as a special case of the linear factor model, where a simple difference-in-differences transformation identifies the effect.
We show that analogous, albeit more complex, transformations exist in the more general linear factor model, providing a new means to identify the effect in that model.
In fact many such transformations exist, called bridge functions, all identifying the same causal effect estimand. This poses a unique challenge for estimation and inference, which we solve by targeting the \emph{minimal} bridge function using a regularized estimation approach.
We prove that our resulting average causal effect estimator is $\sqrt{N}$-consistent and asymptotically normal, and we provide asymptotically valid confidence intervals. 
Finally, we provide extensions for the case of a linear factor model with time-varying unmeasured confounders.
\end{abstract} 


\newpage
\spacingset{1.9} 
\section{Introduction}
\label{sec:intro}
Panel data, where the researcher has multiple observations on the same units over time, are ubiquitous in applications and offer a unique opportunity to control for unmeasured confounding and draw credible causal inferences. For example, under a two-way fixed effect (TWFE) model, we may remove confounding effects by simply subtracting pre-treatment outcomes. The resulting estimator is known as the difference-in-differences (DID) estimator. Because the assumptions underlying a TWFE model can be controversial in practice, more recent estimators such as synthetic controls and matrix completion allow for generalizations such as linear factor models, which may be more realistic. However, these and related methods require both the number of units \emph{and} the number of time periods to grow large for consistent estimation of the effects of interest. 

In this paper, we tackle causal inference from panel data under a linear factor model with a \emph{fixed} number of time periods. 
We focus on the setting where units either start being treated at the same time period and then remain treated afterward, or are never treated at all.
We show how causal effect identification is possible when we have observations in both pre-treatment and post-treatment periods. 
To do this, we interpret the DID  differencing transformation in TWFE as a \emph{bridge function}, a function that transforms pre-treatment variables so that the effect of unmeasured confounders on the transformed variables is the same as on the unobserved counterfactual outcome. In the DID setting the bridge function is very simple: any linear combination of the pretreatment outcomes with the weights summing to one will work.
We show that, analogously, bridge functions also exist in the more general linear factor model, if pre-treatment variables are sufficiently informative about unobserved confounding factors.
These bridge functions can effectively control for unmeasured confounding, facilitating the identification of average causal effects. 
Although bridge functions are defined in terms of unmeasured confounders, we can learn them by using moment equations based on post-treatment observations under an additional serial independence assumption and use them to learn average causal effects. 
Importantly, the number of pre-treatment or post-treatment periods does not need to grow to infinity, but only to be sufficiently large in order to account for all unmeasured confounders.

However, we show that there often exist many different bridge functions even when the average causal effect is uniquely identified. 
Although each of these bridge functions is valid in identifying the true average causal effect, the multiplicity of bridge functions creates challenges for estimation and inference.
We solve this using a new regularized generalized method of moments (GMM) estimator that targets the \emph{minimal} bridge function, \ie, the bridge function whose unknown parameters have the smallest size among all valid ones. 
We prove that the final average causal effect estimator based on this approach is consistent and asymptotically normal, and provide asymptotically valid confidence intervals based on a simple plug-in asymptotic variance estimator. 
We also prove that using the inverse moment covariance matrix as the weighting matrix in a class of regularized GMM estimators is asymptotically optimal in a certain sense. 
Notably, our estimator and inferential procedure are agnostic to whether bridge functions are nonunique nor not. 
When the bridge function happens to be unique, our results recover standard GMM estimation theory, but when this is not true, our results are still valid but standard GMM estimators may be ill-performing.

This paper is organized as follows. In \cref{sec: setup}, we set up the linear factor model and average causal effect estimand. Then we use the TWFE model as a special example to motivate the concept of bridge functions (\cref{sec: twfe}), review existing methods for panel data causal inference, and show their limitations in requiring infinitely many cross-sectional units and time periods to consistently estimate causal effects (\cref{sec: previous}). 
In \cref{sec: identification}, we derive bridge functions for the linear factor model and use them to identify the average causal effect based on a fixed number of pre-treatment and post-treatment outcomes. 
In \cref{sec: estimation}, we propose a regularized estimator for the minimal bridge function and present the estimation and inferential theory. In \cref{sec: discuss}, we extend our methodology to alternative causal effect estimands and linear factor models with time-varying unmeasured factors, and connect this paper to the negative control framework proposed in recent literature \citep[\eg, ][]{cui2020semiparametric,tchetgen2020introduction,miao2018identifying,deaner2021proxy}. 
We finally conclude in \cref{sec: conclusion}.

\section{Problem Setup}\label{sec: setup}
\begin{table}[t]
\centering
\begin{tabular}{c | cccc ccc}
 \toprule 
 \diagbox[width=2cm]{Group}{Time}
 & $-T_0$ & $\cdots$ & $-1$ & \cellcolor{pink} $0$ & \cellcolor{Gray} $1$ & \cellcolor{Gray} $\cdots$ & \cellcolor{Gray} $T_1$  \\
 \midrule 
 \multirow{3}{*}{$A=0$} & \checkmark & $\hdots$ & \checkmark & \cellcolor{pink} \checkmark  & \cellcolor{Gray} \checkmark & \cellcolor{Gray} $\hdots$ & \cellcolor{Gray} \checkmark    \\
  & $\vdots$ & $\ddots$  & $\vdots$  &$  \cellcolor{pink} \vdots$  & \cellcolor{Gray} $\vdots$  & \cellcolor{Gray} $\ddots$  & \cellcolor{Gray} $\vdots$  \\
  & \checkmark & $\hdots$ & \checkmark &  \cellcolor{pink} \checkmark  & \cellcolor{Gray} \checkmark & \cellcolor{Gray} $\hdots$  & \cellcolor{Gray} \checkmark  \\
  \hline 
  \multirow{3}{*}{$A=1$} & \checkmark & $\hdots$ & \checkmark & \cellcolor{pink} $?$  & \cellcolor{Gray} $?$ & \cellcolor{Gray} $\hdots$  & \cellcolor{Gray} ? \\ 
  & $\vdots$ & $\ddots$  & $\vdots$  &$ \cellcolor{pink} \vdots$  & \cellcolor{Gray} $\vdots$  & \cellcolor{Gray} $\ddots$  &\cellcolor{Gray} $\vdots$ \\
  & \checkmark & $\hdots$ & \checkmark & \cellcolor{pink} $?$  & \cellcolor{Gray} $?$ & \cellcolor{Gray} $\hdots$  & \cellcolor{Gray} ? \\
 \bottomrule 
\end{tabular}
\caption{Observation patterns for the control potential outcomes $Y_{i, t}\prns{0}$. ``$\checkmark$'' means $Y_{i, t}\prns{0}$ is observed while ``$?$'' means otherwise. White cells denote pre-treatment periods, pink cells denote the treatment period, and gray cells denote post-treatment periods. }\label{table: data}
\end{table}

Consider a balanced panel with $N$ cross-sectional units observed in time periods $t \in \braces{-T_0, \dots, -1, 0, 1, \dots, T_1}$. A treated unit $i$, with treatment indicator $A_i = 1$,  is exposed to a treatment of interest for the first time at time $t = 0$ and remains in the treatment condition until the end of time horizon, whereas a control unit $i'$, with $A_{i'} = 0$,  is never exposed to the treatment. We use $Y_{i, t}(0)$ to denote the potential outcome that would be observed for unit $i$ at time $t$ in absence of the treatment condition, and use $Y_{i, t}(1)$ for the potential outcome in the treatment condition. Accordingly, we observe $Y_{i, t} = Y_{i, t}\prns{1}$ only when $A_i = 1$ and $t \ge 0$, and observe $Y_{i, t} = Y_{i, t}\prns{0}$ otherwise. See \cref{table: data} for an illustration. Additionally, we may observe covariates $X_i \in \R{d}$ for each unit. We assume that $\prns{A_i, X_i, Y_{i, t}\prns{0}, Y_{i, t}\prns{1}: -T_0 \le t \le T_1}$ for $i = 1, \dots, N$ are independent and identically distributed ({i.i.d}) draws from a common (infinite) population denoted as $\prns{A, X, Y_{t}\prns{0}, Y_{t}\prns{1}: -T_0 \le t \le T_1}$.

In this paper, we are interested in the average causal effect for the treated  at time $t = 0$:
\begin{align*}
\ATT = \Eb{Y_{0}\prns{1} - Y_{0}\prns{0} \mid A = 1} = \Eb{Y_{0} \mid A = 1}- \Eb{Y_{0}\prns{0} \mid A = 1}.
\end{align*}
Because the first term $\Eb{Y_{0} (1)\mid A = 1}$ is trivially identified, we focus on studying the identification and estimation of the second term, \ie, the counterfactual mean for the treated:
\begin{align}\label{eq: target}
    \gamma^* = \Eb{Y_{0}\prns{0} \mid A = 1}.
\end{align}

In this paper, we assume that the control potential outcomes satisfy a linear factor model.
\begin{assumption}[Linear Factor Model]\label{assump: outcome}
\begin{align}\label{eq: factor-model}
 &Y_{i, t}(0) = V_t^\top U_i + b^\top_t X_{i} + \epsilon_{i, t}.
\end{align}
where $U_i \in \R{r}$ are unobserved confounders, $V_t \in \mathbb R^r$ characterize their unknown effects on the potential outcomes, and $\epsilon_{i, t}$ are mean-zero idiosyncratic error terms. In \cref{eq: factor-model},  $\prns{U_i, X_i, \epsilon_{i, -T_0}, \dots, \epsilon_{i, T_1}}$ are i.i.d draws from a common population $\prns{U, X, \epsilon_{-T_0}, \dots, \epsilon_{-T_1}}$ with finite second order moments, and $\epsilon_{t} \perp \prns{U, X, A}$ for any $t$.
\end{assumption}
Linear factor models for causal effects have been the focus of a growing literature \citep[\eg, ][]{doudchenko2016balancing,Abadie2007,bai2009panel,athey2021matrix,xu_2017,ruoxuan2020}. One of its key property is that the unobserved confounders $U_i$ has \emph{time-varying} effects on the potential outcomes, captured by the vectors $V_t$. Although here the covariates $X_i$ appear to be time invariant, actually we can still allow time-varying covariates. For example, we can consider $X_i = \prns{X_{i, -T_0}, \dots, X_{i, T_1}}$, and set coefficients $b_{t}$ to have nonzero entries only for the covariates $X_{i, t}$. 
The condition $\epsilon_{t} \perp \prns{U, X, A}$ requires the idisyncratic errors to be exogenous. Under this condition, we have
\begin{align*}
 Y_{0}\prns{0} \perp A \mid X, U, 
 \end{align*} 
 but we allow $U$ to be unmeasured confounders, \ie, $A \not \perp U \mid X$ and  
 \begin{align*}
 Y_{0}\prns{0} \not\perp A \mid X.
 \end{align*} 
We also impose a standard positivity assumption on the treatment assignment. 
\begin{assumption}[Positivity]\label{assump: treatment}
The treatment assignment $A$ satisfies that $\Prb{A = a \mid X, U} > 0$  almost surely for $a = 0, 1$. 
\end{assumption}

In this paper, we aim to study the identification and estimation of the counterfactual mean for the treated, \ie, $\gamma^*$ in \cref{eq: target}, based on the following observed data:
\begin{align*}
    \braces{O_i = \prns{A_i, X_i, Y_{i, \pre}, Y_{i, 0}, Y_{i, \post}}: i = 1,\dots, N },
\end{align*}
where $Y_{i, \pre} = \prns{Y_{i, t}: t < 0}$ and $Y_{i, \post} = \prns{Y_{i, t}: t > 0}$ denote pre-treatment outcomes and post-treatment outcomes, respectively. 

\paragraph{Notation.} 
We use $\C = \braces{i: A_i = 0}$ and $\T = \braces{i: A_i = 1}$ to denote the  sets of control units and treated units respectively, and denote their sizes as $N_0$ and $N_1$ respectively. We let $\pre = \braces{t: t < 0}$ and $\post = \braces{t: t > 0}$  denote the pre-treatment and post-treatment periods, respectively.
We use them in subscripts to represent vectors or matrices whose corresponding components belong to these index subsets. For example, we define a  vector $Y_{\C, 0} \in \R{N_0}$ as  $\prns{Y_{i, 0}: i \in \C}$, and define $\mb Y_{\C, \pre}$ as a $N_0 \times T_0$ matrix whose rows correspond to control units in $\C$ and columns correspond to pre-treatment periods in $\pre$. We also define $\mb{U}_\C \in \mathbb{R}^{N_0 \times r}$, $\mb{V}_\pre \in \mathbb{R}^{T_0 \times r}$, and $\mb{B}_\pre \in \mathbb{R}^{T_0 \times d}$
as matrices  whose rows correspond to $\prns{U_i^\top: i \in \C}$, $\prns{V_t^\top: t \in \pre}$, and $\prns{\beta_t^\top: t \in \pre}$  respectively. For positive integers $n$ and $n'$, we define $\mb{1}_{n}$ as an all-one vector of length $n$, $\mathbf{0}_{n \times n'}$ as an all-zero matrix of size $n \times n'$, and $I_{n \times n}$ as an $n \times n$ identity matrix. 
Other vectors and matrices with similar subscripts can be understood analogously.
For a  function $f(O)$ of observed variables $O$, we denote its sample average with respect to data $\prns{O_1, \dots, O_N}$ as $\En{f(O)} = \frac{1}{N}\sum_{i=1}^N f(O_i)$. 

\subsection{Special Example: Two-Way Fixed Effects}\label{sec: twfe}
One special example of our  model in \cref{eq: factor-model} is the two-way fixed effect model (TWFE):
\begin{align}\label{eq: twfe}
    Y_{i, t}\prns{0} = {U_i + b_t}  +  \epsilon_{i, t},
\end{align}
where $U_i, b_t \in \mathbb{R}$ denote unit effects and time effects respectively and additional covariates are ignored for simplicity. This corresponds to \cref{eq: factor-model} with $V_t = X_t = 1$. 
This model imposes a strong additivity assumption on the fixed effects, which requires the effects of unobserved confounders on the counterfactual outcomes to be time invariant (\ie, $V_t = 1$). 
This also implies the so-called ``parallel trend'' assumption, \ie, the average counterfactual outcomes of treated and control units follow parallel paths. 
However, this assumption is often controversial in practice
({\it e.g.,} 
\citet{callaway2020difference, goodman2018difference, sun2020estimating}). 
The linear factor model in \cref{assump: outcome} is substantially more general and in particular does not impose this ``parallel trend'' assumption. 

The TWFE model accommodates a particularly simple estimation procedure, known as the difference-in-difference (DID) estimator for  $\Eb{Y_{0}\prns{0} \mid A = 1}$: 
\begin{align}\label{eq: did}
    \hat \gamma_{\did} = \frac{1}{T_0N_1}\sum_{t \in \pre}\sum_{i \in \T}Y_{i, t} - \frac{1}{T_0N_0}\sum_{t \in \pre}\sum_{i \in \C}Y_{i, t} + \frac{1}{N_0}\sum_{i \in \C}Y_{i, 0}.
\end{align}
To interpret the DID estimator, first we note that it is a sample analogue
of the population estimand
\begin{align}\label{eq: did2}
    \gamma^*_{\did} = 
    \Eb{\frac{1}{T_0}\mb{1}^\top_{T_0}Y_{\pre} \mid A = 1}-
    \Eb{\frac{1}{T_0}\mb{1}^\top_{T_0} Y_{\pre} \mid A = 0}
    + \Eb{Y_{0} \mid A = 0}.
\end{align}
Next, we note that this estimand can be written as:
\begin{align}
\gamma^*_\did = \Eb{h(Y_{\pre}; \theta^*_{\did}) \mid A=1}, \label{eq: did-bridge}
\end{align}
where $h(Y_{\pre}; \theta) = \theta^{\top}_{1} Y_{\pre} + \theta_2$, $\theta^*_{\did, 1} = \mb{1}_{T_0}/T_0$, and $\theta^*_{\did, 2} = \Eb{-\frac{1}{T_0}\sum_{t\in\pre}Y_{t} + Y_{0} \mid A = 0} = b_0 - b^\top_\pre \theta_1^*$. 
Here $h(Y_{\pre}; \theta^*_{\did})$ is a transformation of the pre-treatment outcomes, whose expectation for the treated units is exactly equal to the target parameter $\gamma^*$.  
The DID estimator learns this transformation based on  control units' observed outcomes, \ie, $Y_{i, t}$ for $i \in \C$ and $t \in \pre \cup \braces{0}$.

Actually, the transformation $h(Y_{\pre};\theta^*_{\did})$ learned by the DID estimator is only one among a many valid transformations: we in fact have
\begin{align}
     \gamma^* = \Eb{h(Y_{\pre}; \theta^*) \mid A=1}, ~\text{for any}~ 
     \theta^* \in \Theta^*_{\op{FE}} = \braces{\theta^*: \mb{1}_{T_0}^\top\theta_1^* = 1, ~ 
    \theta_2^* = b_0 - b_\pre^\top\theta_1^*}. \label{eq: bridge-fe} 
\end{align}
These functions all have a special property that
\begin{align}\label{eq: did-bridge-2}
\Eb{Y_{0}(0) - h(Y_{\pre}; \theta^*) \mid U, A = 0} = 0, ~\text{for any}~ 
     \theta^* \in \Theta^*_{\op{FE}}.
\end{align}
The condition in \cref{eq: did-bridge-2} says that the effect of unmeasured confounders $U$ on the transformed pre-treatment outcomes $h(Y_{\pre})$ is exactly the same as the unmeasured confounding effect on $Y_{0}(0)$. Consequently, any such transformation, which we call a \emph{bridge function}, can recover the parameter $\gamma^*$.

According to \cref{eq: bridge-fe} , whenever $T_0 > 1$, there are infinitely many different bridge functions, whose coefficients on the pre-treatment outcomes correspond to different solutions to the equation $\mb{1}_{T_0}^\top\theta_1^* = 1$. Among them, the DID bridge function $h(Y_{\pre}; \theta^*_{\did})$ is particularly simple, in that its coefficient $\theta^*_{1, \did}$ has the smallest $L_2$ norm. 

In this paper, we generalize this bridge function approach to the linear factor model in \cref{assump: outcome}. 
We will show that for this more general model, there also exist (usually nonunique) linear bridge functions of pre-treatment outcomes that can control for the unmeasured confounding effects on the primary outcome.
However, learning the bridge functions becomes more difficult: the coefficients on the pre-treatment outcomes depend on unknown parameters, so we can no longer directly pick a particular one as we do in the DID estimator. 
Instead, we have to estimate them from data, which we realize by leveraging post-treatment outcomes (see \cref{thm: identification}).
The nonuniqueness of the bridge functions poses a unique challenge to estimation and inference, which we tackle using regularization 
(see \cref{sec: estimation}).

\subsection{Existing Approaches and Their Limitations}\label{sec: previous}
\begin{table}[t]
\begin{subtable}{0.45\textwidth}
\centering
\begin{tabular}{c | cccc }
 \toprule 
 \diagbox[width=2cm]{Group}{Time}
 & $-T_0$ & $\cdots$ & $-1$ & $0$   \\
 \midrule 
 \multirow{3}{*}{$A=0$} & \multicolumn{3}{c}{\horzbar  ~ $Y_{1, \pre}$  \horzbar} & $Y_{1, 0}$   \\
  & $\vdots$ & $\vdots$  & $\vdots$  &$  \vdots$  \\
  &\multicolumn{3}{c}{\horzbar  ~ $Y_{N_0, \pre}$  \horzbar} & $Y_{N_0, 0}$   \\
  \hline 
  \multirow{3}{*}{$A=1$} & \multicolumn{3}{c}{\horzbar  ~ $Y_{N_0+1, \pre}$  \horzbar}  & $?$ \\ 
  & $\vdots$ & $\vdots$  & $\vdots$  &$  \vdots$  \\
  & \multicolumn{3}{c}{\horzbar  ~ $Y_{N, \pre}$  \horzbar}  & $?$ \\ 
 \bottomrule 
\end{tabular}
\caption{Horizontal Regressons}\label{table: horizontal-regression}
\end{subtable}
\hspace{0.1cm}
\begin{subtable}{0.45\textwidth}
\centering
\begin{tabular}{c | cccc }
 \toprule 
 \diagbox[width=2cm]{Group}{Time}
 & $-T_0$ & $\cdots$ & $-1$ & $0$   \\
 \midrule 
 \multirow{3}{*}{$A=0$} & \vertbar & $\hdots$ & \vertbar & \vertbar   \\
  & $Y_{\C, -T_0}$ & $\hdots$ & $Y_{\C, -1}$ & $Y_{\C, 0}$   \\
  & \vertbar & $\hdots$ & \vertbar & \vertbar \\
  \hline 
  $A=1$ & $\overline{Y}_{\T, -T_0}$ & $\hdots$ & $\overline{Y}_{\T, -1}$ & ? \\
 \bottomrule 
\end{tabular}
\caption{Vertical Regressons}\label{table: vertical-regression}
\end{subtable}
\caption{An illustration for the data used in horizontal regressions and vertical regressions. Without loss of generality, we sort the data table so the first $N_0$ units are all control units. In figure (b), elements in the last row are average pre-treatment outcomes for the treated, \ie,  $\overline{Y}_{\T, t} = \frac{1}{N_1}\sum_{i \in \T} Y_{i, t}$ for $t \le -1$.}\label{table: regression}
\end{table}

In this part, we review some existing methods for estimating causal effects under the linear factor model and show their limitations. For simplicity, we ignore the covariates $X$ in \cref{assump: outcome}, \ie, setting $b_t = 0$ for all $t$. 
\paragraph*{Horizontal Regressions.}
In \cref{eq: did-bridge}, we show that under the TWFE model, some transformations of the pre-treatment outcomes, \ie, the bridge functions, can control for the unmeasured confounding effects.
 For the linear factor model, we also have similar transformations:  in \cref{sec: app-existing} \cref{lemma: horizontal-vertical-form}, we show that when $\mb{V}_\pre$ has full column rank (\ie, equal to the number of unmeasured confounders $r$), there exist $\theta^*_1 \in \mathbb R^{T_0}$ such that 
\begin{align}\label{eq: vertical}
Y_{i, 0}\prns{0}  =  \theta^{*\top}_1 Y_{i, \pre} + \xi_i, ~ \xi_i = \epsilon_{i, 0} - \theta^{*\top}_1 \epsilon_{i, \pre}, ~~ \forall i.
\end{align}
Obviously, once we can learn a coefficient $\theta^*_1$ satisfying \cref{eq: vertical} above, we can immediately estimate $\gamma^* = \Eb{Y_{0}\prns{0}\mid A = 1}$. 
One straightforward idea is to view \cref{eq: vertical} as a regression equation, and run a linear regression of $Y_{i, 0}$ against $Y_{i, \pre}$, based on the data for control units up to time $t = 0$, possibly with additional regularization. This amounts to using $N_0$ data points to estimate the $T_0$-dimensional parameter $\theta^*$ in \cref{eq: vertical}. We call this as a \emph{horizontal regression}, because the regressors $Y_{i, \pre}$ are horizontally laid out in the observed data matrix (see \cref{table: horizontal-regression}). Variants of horizontal regressions are also considered in \cite{hazlett2018trajectory,athey2021matrix}.

However, this horizontal regression approach is susceptible to estimation bias. Indeed, the regressors $Y_{i, \pre}$ are dependent with errors $\xi_i$, since both include common components $\epsilon_{i, \pre}$. 
This is analogous to the well-known problem of error-in-variable regressions, where using proxy variables in place of unobserved true variables as regressors leads to coefficient estimates with nonvanishing bias \citep[\eg, ][Section 4.4]{wooldridge2010econometric}. 
Similarly, here we can view pre-treatment outcomes as proxy variables for the unmeasured confounders, so even when the horizontal regression has infinite sample size, \ie, $N_0 = \infty$, the resulting counterfactual mean estimator will still have persistent bias. 
In \cref{sec: app-existing} \cref{lemma: horizontal-bias}, we prove that the bias can vanish when the dimension of regressors in the horizontal regression, \ie, $T_0$, also grows to infinity. 
This means that we need both $N_0 \to \infty$ and $T_0 \to \infty$ to consistently estimate $\gamma^*$ based on horizontal regressions, which is infeasible with only observations in a limited number of time periods. 

\paragraph*{Vertical Regressions.} In \cref{sec: app-existing} \cref{lemma: horizontal-vertical-form}, we also show that when $\mb{U}_{\C}$ has full column rank (which holds with high probability if components of $U$ are not collinear), there exists $w^* \in \R{N_0}$ such that 
 \begin{align}
\frac{1}{N_1}\sum_{i \in \T} Y_{i, t}\prns{0}  = w^{*\top} Y_{\C, t} + \nu_t, ~ \nu_t = {\frac{1}{N_1}\sum_{i \in \T} \epsilon_{i, t} - w^{*\top} \epsilon_{\C,t}}, ~~ \forall t. \label{eq: transform-2}
\end{align}
We may also view \cref{eq: transform-2} as a regression function, and run a linear regression of $\frac{1}{N_1}\sum_{i \in \T} Y_{i, t}$ against $Y_{\C, t}$ based on data up to time $t = -1$.  This amounts to using $T_0$ data points to estimate the $N_0$ dimensional parameter $w^*$ in \cref{eq: transform-2}.  We call this as a \emph{vertical regression}, because the regressors $Y_{\C, t}$ are vertically laid out in the observed data matrix (see \cref{table: vertical-regression}). The vertical regression recovers the synthetic control method \citep{Abadie2007,abadie2003conflict} if the  regression coefficients are additionally constrained to be nonnegative and sum to one. Other variants of vertical regressions also appear in \cite{doudchenko2016balancing,chernozhukov2021ttest,BenMichael2021}.

However, the vertical regression approach also has the error-in-variable regression problem, because the regressors $Y_{\C, t}$ and the errors $\nu_t$ share common components $\epsilon_{\C, t}$. Therefore, the resulting counterfactual mean estimators also have nonvanishing biases when the sample size of the vertical regressions, \ie, $T_0$, grows to infinity. Instead, we show in \cref{sec: app-existing} \cref{lemma: vertical-bias}
 that 
 we also need the dimension of regressors $N_0 \to \infty$ to consistently estimate the counterfactual mean. Similar observations were also noted by \cite{ferman2021synthetic,ferman2020properties,Gobillon2016}, using different analyses.

\paragraph{Matrix Estimation.} Alternatively, some recent literature propose to impute missing counterfactuals by directly learning the factor model structure in \cref{eq: factor-model}, \eg, by estimating $U_i$ and $V_t$ factors for all $i$ and $t$  \citep[\eg, ][]{xiong2019large,xu_2017,bai2021matrix}, by matrix norm regularization methods \citep[\eg][]{athey2021matrix,farias2021learning}, or by singular value thresholding \citep[\eg][]{amjad2018robust}.
However, learning the factor model structure is a very difficult high-dimensional estimation problem. To consistently estimate the factor model structure and causal parameters, these existing estimators need both $N_0$ and $T_0$ to grow to infinity (see \cref{sec: app-existing} \cref{lemma: factor} for details), with post-treatment outcomes or not. 
Note that this is very different from the TWFE model: although learning all fixed effects $U_i, b_t$ for $i, t$ is also a difficult high-dimensional estimation problem, we do not need to estimate them at all. Instead, we can directly estimate the causal parameters consistently when $N \to \infty$ but $T_0$ is fixed. 

\paragraph{Summary.} These existing approaches require both $N_0 \to \infty$ and $T_0 \to \infty$ to consistently estimate the causal parameter, either because of the error-in-variable regression problem, or because of the need to learn the factor model directly. Note that the post-treatment data (\ie, gray cells in \cref{table: data}) are not important in these approaches. 
Actually, it is not immediately clear how to use post-treatment data in horizontal/vertical regressions, since in post-treatment periods the counterfactual outcomes for the treated units are all missing. 
 In the matrix estimation methods, although we can indeed incorporate post-treatment outcomes, they also bring in more missing entries and cannot relax the requirement of $N_0 \to \infty$ and $T_0 \to \infty$  for consistent estimation of $\gamma^*$ (see discussions in \cref{sec: app-existing}). 

In this paper, we will show that even when the number of pre-treatment outcomes $T_0$ is fixed, we can still identify the counterfactual mean parameter $\gamma^*$ and estimate it consistently, provided that we have access to some post-treatment observations for a \emph{fixed} number of periods (\ie, fixed $T_1$). 
We will build on a generalization of \cref{eq: vertical} that relates the missing counterfactual to the pre-treatment outcomes and covariates, and show that 
post-treatment outcomes are valuable in addressing the problem of ``error-in-variable'' regressions.
Importantly, our results reveal that even though the factor model structure cannot be identified or consistently estimated with  data only in  a fixed number of time periods, identification and consistent estimation of causal effects is still possible.
Our methods thus uniquely enable effective causal inference under the popular linear factor model with big-$N$-small-$T$ panel data.

\section{Identification via Bridge Functions}\label{sec: identification}
In \cref{sec: twfe}, we show that in the TWFE model, the so-called bridge functions can effectively control for unmeasured confounding and lead to the familiar DID estimator.  In this section, we derive bridge functions for the linear factor model in \cref{assump: outcome}, by generalizing the formulation in \cref{eq: vertical}. 
Although horizontal regression methods based on this formulation may not consistently the target parameter when $T_0$ is fixed, we show that this can be realized by leveraging post-treatment outcomes. 

We first introduce the definition of bridge functions. 
\begin{definition}[Bridge Functions]\label{def: bridge}
A function $h\prns{Y_{\pre}, X}$ is called a bridge function if 
 \begin{align}
    \Eb{Y_0(0) - h\prns{Y_{\pre}, X} \mid U, A = 0, X} = 0, ~~ \text{ almost surely. } \label{eq: bridge}
\end{align}
\end{definition}
According to \cref{eq: bridge}, bridge functions give some transformations of the pre-treatment outcomes and covariates, such that the unmeasured confounding effects on this transformation exactly reproduce those on the counterfactual outcome. This formalizes the requirement that bridge functions can control for unmeasured confounding. 

In the following lemma, we show that under our assumptions, the treatment has no direct causal effects on the pre-treatment outcomes.
As a result, their bridge function transformations, despite being defined in terms of the \emph{control} population in \cref{eq: bridge}, can be applied to the \emph{treated} population to recover the counterfactual mean for the treated. 
\begin{lemma}\label{lemma: bridge}
Under \cref{assump: outcome}, we have 
\begin{align}\label{eq: no-direct-effect}
Y_{\pre} \perp A \mid U, X.
\end{align}
If further \cref{assump: treatment} holds, then for any bridge function $h\prns{Y_{\pre}, X}$ in \cref{def: bridge}, 
\begin{align}
    \gamma^* = \Eb{Y_0\prns{0} \mid A = 1} = \Eb{h\prns{Y_{\pre}, X} \mid A = 1}. \label{eq: identification}
\end{align}
\end{lemma}
In the following lemma, we prove the existence of bridge functions in the linear factor model, by generalizing \cref{eq: vertical} to incorporate additional covariates. 
\begin{lemma}\label{lemma: bridge-prelim}
 Suppose Assumption \ref{assump: outcome} holds. If $\mb{V}_\pre \in \R{T_0 \times r}$ has full column rank, \ie, $\op{Rank}\prns{\mb{V}_\pre} = r$, then for any solution $\theta^{*}_1 \in \mathbb{R}^{T_0}$ to the equation $\mb{V}_\pre^\top \theta^{*}_1 = V_0$, we have 
\begin{align}
    &Y_{i, 0}\prns{0}  =   \theta^{*\top}_1 Y_{i, \pre} + \prns{b_0 - \mb{B}_\pre^\top\theta^*_1}^\top X_i + \xi_i, ~ \xi_i = \epsilon_{i, 0} - \theta^{*\top}_1 \epsilon_{i, \pre}, ~~  \forall i. \label{eq: nc-bridge} 
\end{align}
Moreover, if further assuming \cref{assump: treatment}, then any such $\theta^{*}_1$ satisfies that 
\begin{align}
    &\Eb{Y_{0}(0)  -  \theta^{*\top}_1 Y_{\pre} - \prns{b_0 - \mb{B}_\pre^\top\theta^*_1}^\top X \mid  U, A = 0, X} = 0.\notag
\end{align}
\end{lemma}
\cref{lemma: bridge-prelim} shows the form of bridge functions in the linear factor model:
\begin{align}
    &h(Y_{\pre}, X; \theta^*) = \theta^{*\top}_1 Y_{\pre} + \theta_2^* X, \label{eq: bridge-form} \\ &\text{ for any } \theta^* \in \Theta^* = \braces{\theta^*: \mb{V}_\pre^\top \theta_1^* = V_0, \theta_2^* = {b_0 - \mb{B}_\pre^\top\theta^*_1}}. \label{eq: bridge-theta}
\end{align}
Any of these bridge functions can identify the target counterfactual mean parameter $\gamma^*$ by \cref{eq: identification}, or equivalently by the following moment equation: 
\begin{align}
    \Eb{A\prns{h(Y_{\pre}, X; \theta^*) - \gamma^*}} = 0. \label{eq: nc-identification} 
\end{align}
These bridge functions exist when solutions $\theta_1^*$ to the equation $\mb{V}_\pre^\top \theta_1^* = V_0$ exist, which is ensured for $\mb{V}_\pre$ with full column rank $r$. 
Intuitively, this rank condition means that pre-treatment outcomes are informative proxy variables for the unmeasured confounders: their number must be no smaller than the number of unmeasured confounders, \ie, $T_0 \ge r$,
and the unmeasured confounding effects on them, \ie, $U^\top \mb{V}_\pre$, captures the variations of confounders $U$ in any direction. 
This is why some linear transformations of pre-treatment outcomes can control for the confounding effects on the counterfactual outcome $Y_{0}\prns{0}$. Moreover, it is obvious that the confounding bridge function is unique  only when the number of pre-treatment outcomes is equal to the number of unmeasured confounders, \ie, $T_0 = r$ (also see \cref{lemma: nonunique}).

However, it remains unclear how to learn the bridge functions from observed data, since their definition involves unmeasured confounders $U$ (see \cref{eq: bridge}), and their coefficients depend on unknown factors $\mb{V}_\pre, V_0$ and unknown coefficients $\mb{B}_\pre, b_0$. This is in stark contrast to bridge functions in the TWFE model, where all but one coefficients are specified by fully known equations (see \cref{eq: bridge-fe}). 
We already show that 
directly regressing $Y_{i, 0}$ against $Y_{i, \pre}$ and $X_i$, just like the horizontal regressions in \cref{sec: previous}, cannot learn bridge functions without bias due to the error-in-variable regression problem.
In the following lemma, we show that post-treatment outcomes provide new opportunities for learning the bridge functions from observed data. 
\begin{lemma}\label{lemma: post-treatment}
Let assumptions in \cref{lemma: bridge-prelim} hold. 
If ${\epsilon_{\post} \perp \prns{\epsilon_{\pre}, \epsilon_0}}$, then 
\begin{align}\label{eq: post-indep0}
Y_{\post} \perp \prns{Y_\pre, Y_0} \mid X, U, A = 0.
\end{align}
It follows that for any $\theta^* \in \Theta^*$, 
\begin{align}\label{eq: moment-marginal} \Eb{\prns{1-A}\prns{Y_{0}  -  h\prns{Y_{\pre}, X; \theta^{*}}}
        \begin{bmatrix}
        Y_{\post} \\
        X
        \end{bmatrix}
        } = \mb{0}_{\prns{T_1 + d} \times 1}.
\end{align}
\end{lemma}
In \cref{lemma: post-treatment}, we assume that the idiosyncratic errors in the post-treatment periods are independent with those up to period $0$. 
This condition trivally holds when $\epsilon_t$ for $t = -T_0, \dots, T_1$ are all serially independent. Under this condition, post-treatment outcomes $Y_{\post}$ are conditionally independent with the pre-treatment outcomes $Y_{\pre}$ and target outcome $Y_0$. 
Serial independence assumption is often not assumed in the TWFE model (\cref{sec: twfe}), or in many previous literature on regression-based methods and matrix estimation (\cref{sec: previous}), with some exceptions like \citet{Abadie2007,amjad2018robust}.
However, given the general factor model with only limited pre-treatment outcomes, additional assumptions like this become important in causal effect  identification. 
Note that this assumption does not rule out dependence among the outcomes in different periods. The post-treatment outcomes  $Y_{\post}$ can be still dependent with the pre-treatment outcomes $Y_{\pre}$ and the target outcome $Y_0$, but their dependence is completely mediated by unmeasured confounders $U$ and covariates $X$. 
In \cref{sec: discuss}, we further allow confounders to be time-varying, so that dependence structure of the outcomes can be even more complex. 

\cref{lemma: post-treatment} gives a moment equation characterization of bridge functions in \cref{eq: moment-marginal}, which only depends on observed data. 
Although this moment equation looks very similar to moment equations in instrumental variable estimation \citep{wooldridge2010econometric}, 
it is based on substantially different assumptions. 
Actually, post-treatment outcomes $Y_{\post}$ must not be valid instrumental variables, since below they are assumed to be strongly dependent with the unmeasured confounders.
Note that this moment equation characterization does not suffice to show the identifiability of counterfactual mean $\gamma^*$: some of its solutions may not be valid bridge function coefficients, and based on only observed data, there is no way to distinguish invalid coefficients from valid ones. 
In the following theorem,  we stregthen \cref{lemma: post-treatment} by showing that when post-treatment outcomes are also informative proxy variables for the unmeasured confounders, the moment equation in \cref{eq: moment-marginal} \emph{sharply} characterizes all valid bridge functions, which proves the identifiability of the target parameter $\gamma^*$.
\begin{theorem}\label{thm: identification}
Assume conditions in \cref{lemma: bridge-prelim,lemma: post-treatment} and two additional conditions:
\begin{enumerate}
    \item The following $\prns{r + d} \times \prns{r + d}$ second order moment matrix has full rank: 
\begin{align*}
\begin{bmatrix}
\Eb{UU^\top \mid A = 0} & \Eb{UX^\top \mid A = 0} \\
\Eb{XU^\top \mid A = 0} & \Eb{XX^\top \mid A = 0}
\end{bmatrix}.    
\end{align*}
    \item The following $\prns{T_1 + d}\times\prns{r + d}$ matrix has full column rank $r+d$:
    \begin{align*}
        \begin{bmatrix}
        \mb{V}_\post & \mb{B}_\post \\
        \mb{0}_{d\times r} & I_{d\times d}
        \end{bmatrix}.
    \end{align*}
\end{enumerate}
Then $\theta^* \in \Theta^*$  if and only if it satisfies \cref{eq: moment-marginal}, and $\gamma^*$ in \cref{eq: target} is identifiable. 
\end{theorem}
Here condition 1 rules out multicollinearity of the unmeasured confounders $U$ and $X$, which is a common identification condition in linear models. Condition 2 roughly means that the unobserved confounding effects on the post-treatment outcomes, \ie, $U^\top \mb{V}_\post$, after accounting for the covariates $X$, can still capture variations of confounders $U$ in any direction. This condition implicitly requires the number of post-treatment outcomes to be no smaller than the number of unmeasured confounders, \ie, $T_1 \ge r$. Under these two conditions, we can characterize all valid bridge functions by the moment equation in \cref{eq: moment-marginal} that only involves observed data, and plug any of them into \cref{eq: nc-identification}   to identify the counterfactual mean parameter $\gamma^*$. This shows that  $\gamma^*$ can be completely determined by observed data so it is identifiable.

In \cref{lemma: post-treatment,thm: identification}, we show that when ${\epsilon_{\post} \perp \prns{\epsilon_{\pre}, \epsilon_0}}$, post-treatment outcomes can be used to learn bridge functions and identify the causal parameter. If we further assume that the idiosyncratic errors are serially independent, then we can achieve this with other alternative observations. 
Indeed, if this is the case and given that the unmeasured confounders $U$ are time-invariant, then the temporal order of data is not important. We may use additional pre-treatment outcomes not in $Y_{\pre}$ (\eg, outcomes before time $-T_0$) or a mix of these additional pre-treatment outcomes and the post-treatment outcomes $Y_{\post}$ to form the marginal moments in \cref{eq: moment-marginal}. 
Nevertheless, in \cref{sec: discuss}, we show that when unmeasured confounders themselves are time-varing, the temporal order of data is indeed important, and we \emph{must} only use post-treatment outcomes to learn the bridge functions (see discussions below \cref{thm: bridge-varying-obs}). 
Thus we focus on using post-treatment outcomes to learn bridge functions as it is robust to time-varying unmeasured confounding. 

\section{Regularized GMM Estimation}\label{sec: estimation}
In \cref{sec: identification}, we prove the identifiability of the counterfactual mean for the treated parameter $\gamma^*$ by bridge functions, based on both pre-treatment outcomes and post-treatment outcomes. In particular, \cref{lemma: bridge-prelim,lemma: bridge,thm: identification} show the identification of $\gamma^*$ by the following moment equations of $\prns{\theta, \gamma}$:
\begin{align}
    &\Eb{g\prns{O; \theta, \gamma}} = 0, ~~ \Eb{m\prns{O; \theta}} = 0 , ~ \text{ for } O = \prns{X, A, Y_{\pre}, Y_{0}, Y_{\post}}, \label{eq: equations} \\
    &\text{where } g\prns{O; \theta, \gamma} \coloneqq A\prns{\theta^{\top}_1 Y_{\pre} + \theta_2^\top X - \gamma}, \notag \\
    &\phantom{\text{where }} m\prns{O; \theta} \coloneqq \prns{1-A}\prns{Y_{0}  -  
    \prns{\theta^{\top}_1 Y_{\pre} + \theta_2^\top X}
    }
        \begin{bmatrix}
        Y_{\post} \\
        X
        \end{bmatrix}. \notag 
\end{align}
In this section, we study how to estimate $\gamma^*$ based on these moment equations. One immediate challenge in this estimation task is that solutions to $\Eb{m\prns{O; \theta}} = 0$, \ie, valid bridge functions, may not be unique. 
\begin{lemma}\label{lemma: nonunique}
 Under assumptions in \cref{lemma: bridge-prelim,thm: identification}, if $T_0 > r$, then the set of confounding bridge function coefficients $\Theta^*$ contains infinitely many elements, and for any $\theta^* \in \Theta^* $,  the $\prns{T_1 + d} \times \prns{T_0 + d}$ Jacobian matrix $\nabla \Eb{m\prns{O; \theta^*}}$ has rank at most $r + d$, strictly smaller than the dimension of $\theta^* \in \mathbb{R}^{T_0 + d}$. 
\end{lemma}

Actually, nonunique bridge functions are likely to be prevalent. In practice, we never know the number of unmeasured confounders, so we can at best use as many pre-treatment outcomes as possible to safeguard the existence of bridge functions. It is rather unlikely that the number of negative controls is exactly the same as the number of unmeasured confounders, so that bridge functions uniquely exist. Instead, we may tend to use more than enough, \ie, $T_0 > r$, and end up with nonunique bridge functions. 

Although any valid bridge function  identifies the true counterfactual mean parameter $\gamma^*$, nonunique bridge functions pose serious estimation challenges.  
Standard methods, such as Generalized Method of Moments (GMM) \citep{hansen1982large} that solves sample analogues of \cref{eq: equations}, may fail to perform well. Their estimates are not guaranteed to converge to any fixed elements in $\Theta^*$. Worse yet, the set of valid bridge function coefficients $\Theta^*$ is an unbounded linear subspace, so standard GMM estimators may give very  extreme values, which causes the resulting counterfactual mean estimator to be highly unstable. Technically, \cref{lemma: nonunique} shows that the rank of the Jacobian matrix of the moment equation corresponding to the function $m$ is strictly smaller than the dimension of the bridge function coefficients to be estimated, which violates key conditions in the asymptotic guarantees for standard GMM estimators  \citep{newey1994large}.


To understand how to overcome this challenge, recall the TWFE model in \cref{eq: twfe}. In \cref{sec: twfe}, we show that the bridge functions under the TWFE model, characterized by $\Theta^*_{\op{FE}}$ in \cref{eq: bridge-fe}, are also nonunique whenever the number of pre-treatment outcomes $T_0 > 1$. But we can easily target a specific bridge function and obtain the DID estimator.
Motivated by this, we propose to target the \emph{minimal} bridge function, whose coefficients have the smallest norm among all valid bridge function coefficients:
\begin{align}
    &\theta_{\min}^* \coloneqq \argmin  \braces{\|\theta\|_2: \Eb{m\prns{O; \theta}} = 0} = \prns{{\Eb{\prns{1-A}\tilde Z \tilde W^\top}}}^+\Eb{\prns{1-A}\tilde Z Y_{0}}, \label{eq: minimum-norm} \\
    &\text{where } \tilde Z =
    \begin{bmatrix}
    {Y_{\post}} \\
    X
    \end{bmatrix},
    ~~
    \tilde W =
    \begin{bmatrix}
    Y_{\pre} \\
    X
    \end{bmatrix},
    \notag
\end{align}
and $[\cdot]^+$ denotes the Moore–Penrose pseudoinverse. 
Obviously, when there happen to be a unique bridge function, \ie, $T_0 = r$, the minimal bridge function is trivally the only valid bridge function. 
Thus our approach of targeting the minimal bridge function is valid regardless of whether bridge functions are unique or not.
In \cref{sec: app-reg}, we also show that it is possible to target alternative bridge functions achieving smallest $\theta^\top M\theta$ for a positive semidefinite matrix $M$. 

It might be tempting to estimate $\theta_{\min}^*$ by substituting empirical averages for all true expectations in \cref{eq: minimum-norm}. However, this simple estimator may not even  be  consistent. Note that in \cref{eq: minimum-norm}, we need to the pseudoinverse of a $(T_1 + d)\times(T_0 + d)$ matrix with rank at most $r + d$ (which corresponds to $\nabla\Eb{m\prns{Z; \theta^{*}_{\min}}}$ in \cref{lemma: nonunique}).
It is well known that Moore–Penrose pseudoinverse  is not a continuous operation at singular matrices.
Thus the pseudoinverse of the empirical average matrix may not converge to the pseudoinverse of the limiting expectation matrix. Therefore, this simple estimator can be often ill-performing. 

Instead, we propose a regularized GMM estimator:
\begin{align}\label{eq: GMM}
\hat\theta 
&= \argmin_\theta \prns{\En{m\prns{O; \theta}}}^\top \mathcal{W}_{m, N} \prns{\En{m\prns{O; \theta}}} +  \lambda_N\|\theta\|^2_2,
\end{align}
where $\mathcal{W}_{m, N} \in \mathbb{R}^{\prns{T_1 + d} \times \prns{T_1 + d}}$ is a (possibly data-dependent) positive definite weighting matrix 
that converges (in probability) to a fixed positive definite matrix $\mathcal{W}_{m, \infty}$, and $\lambda_N > 0$ is a regularization parameter that converges to $0$ as $N \to \infty$.
The $L_2$ norm regularization in \cref{eq: GMM} encourages small norm solutions to the sample moment equations. So it is reasonable to expect the regularized estimator $\hat\theta$ to target the minimal bridge function whose coefficients are given in \cref{eq: minimum-norm}.

In the following lemma, we prove that this is indeed the case: when the regularization parameter $\lambda_N$ converges at an appropriate rate, the regularized estimator converges to the minimal bridge function coefficient. 
\begin{lemma}\label{lemma: theta-asymp}
Suppose conditions in \cref{thm: identification} hold, and let $\lambda_N \to 0$ and $\mathcal{W}_{m, N}$ converges in probability to a positive definite matrix $\mathcal{W}_{m, \infty}$  as $N \to \infty$.
Then regularized GMM estimator $\hat\theta$ in \cref{eq: GMM} satisfies that 
\begin{align*}
\|\hat\theta - \theta_{\min}^*\|_2 = \mathcal{O}_p\prns{\lambda_N +\frac{1}{N \lambda_N} + \frac{1}{\sqrt{N}}}.
\end{align*}
\end{lemma}
According to \cref{lemma: theta-asymp}, $\hat\theta$ converges to the minimal bridge function coefficient $\theta_{\min}^*$ in \cref{eq: minimum-norm} when $\lambda_N \to 0$ but $N\lambda_N \to \infty$, which shows the validity of the regularized GMM estimator $\hat\theta$.
Then we can use $\hat\theta$ to get the final estimator for $\gamma^*$:  
\begin{align}
    \En{g\prns{O; \hat\theta, \hat\gamma}} = 0 \implies \hat\gamma = \frac{1}{N_1}\sum_{i\in\T} \tilde W_i^\top \hat\theta.  \label{eq: mean-est}
\end{align}
In the following theorem, we further show that this counterfactual mean estimator has an asymptotically linear expansion with a closed-form influence function. 
\begin{theorem}\label{thm: mu-asymp}
Suppose conditions in \cref{lemma: theta-asymp} hold. Then the counterfactual mean estimator $\hat\gamma$ in \cref{eq: mean-est} satisfies that 
\begin{align}\label{eq: asymp-linear}
\sqrt{N}\prns{\hat\gamma - \gamma^*} 	&= \frac{1}{\sqrt{N}}\sum_{i = 1}^N \psi\prns{O_i; \theta_{\min}^*, \gamma^*, \mathcal{W}_{m, \infty}}+ \mathcal{O}_p\prns{\lambda_N \sqrt{N} + \frac{1}{\sqrt{\lambda_N N}}},
\end{align}
where 
\begin{align*}
&\psi\prns{O_i; \theta_{\min}^*, \gamma^*, \mathcal{W}_{m, \infty}} = -\frac{1}{\Eb{A}} 
\braces{g\prns{O_i; \theta^*_{\min}, \gamma^*}+ \Psi\prns{\mathcal{W}_{m, \infty}} m\prns{O_i ; \theta^*_{\min}}}, \\
&\Psi(\mathcal{W}_{m, \infty}) = \Eb{A \tilde W^\top}\braces{{\Eb{(1-A)\tilde W \tilde Z^\top}} \mathcal{W}_{m, \infty}\Eb{(1-A)\tilde Z \tilde W^\top}}^{+}{\Eb{(1-A)\tilde W \tilde Z^\top}}\mathcal{W}_{m, \infty}.
\end{align*}
\end{theorem}
\cref{thm: mu-asymp} shows that our estimator $\hat\gamma$ based on the regularized GMM estimator for bridge functions has desirable asymptotic properties. Note that the influence function $\psi\prns{O_i; \theta_{\min}^*, \gamma^*, \mathcal{W}_{m, \infty}}$ of estimator $\hat\gamma$ has mean zero. So when $\lambda_N \sqrt{N} \to 0$ and $\lambda_N N \to \infty$, we can use Law of Large Number and Central Limit Theorem to show that estimator $\hat\gamma$ is $\sqrt{N}$-consistent with an asymptotic normal distribution. The asymptotic variance of $\hat\gamma$ is given by the variance of the influence function, \ie, 
\begin{align}
    \sigma^2\prns{\mathcal{W}_{m, \infty}} = \Eb{\psi^2\prns{O; \theta_{\min}^*, \gamma^*, \mathcal{W}_{m, \infty}}}.
\end{align}
This can be estimated by a straightforward plug-in estimator:
\begin{align}
    &\hat\sigma^2\prns{\mathcal{W}_{m, N}} = \En{\hat\psi^2\prns{O; \hat\theta, \hat\gamma, \mathcal{W}_{m, N}}}, \text{ where }\\
    &
    \hat\psi\prns{O_i; \hat\theta, \hat\gamma, \mathcal{W}_{m, N}} = -\frac{1}{\En{A}} 
\braces{g\prns{O_i; \hat\theta, \hat\gamma}+ \hat \Psi\prns{\mathcal{W}_{m, N}} m\prns{O_i; \hat\theta}}, \nonumber \\
    &
    \hat\Psi\prns{\mathcal{W}_{m, N}} = \En{A \tilde W^\top}\braces{{\En{(1-A)\tilde W \tilde Z^\top}} \mathcal{W}_{m, N}\En{(1-A)\tilde Z \tilde W^\top} + \lambda_N I}^{-1}\\
    &\qquad\qquad\qquad\qquad\qquad\qquad\qquad\qquad\qquad\qquad\qquad\qquad \times {\En{(1-A)\tilde W \tilde Z^\top}}\mathcal{W}_{m, N}.  \nonumber 
\end{align}
In the following theorem, we further prove that the variance estimator is consistent and it can be used to construct asymptotically valid confidence intervals. 
\begin{theorem}\label{thm: CI}
Suppose that conditions in \cref{lemma: theta-asymp} hold, and $\Eb{\psi^q\prns{O; \theta_{\min}^*, \gamma^*, \mathcal{W}_{m, \infty}}} < \infty$ for $q = 2, 4$. If  $\lambda_N \sqrt{N} \to 0$, $\lambda_N N \to \infty$, then
\begin{align*}
    \hat\sigma^2\prns{\mathcal{W}_{m, N}} - \sigma^2\prns{\mathcal{W}_{m, \infty}} \to 0, \text{ in probability. } 
\end{align*}
Moreover, given the cumulative distribution function of the standard normal distribution $\Phi$, the $1-\rho$ confidence interval 
\begin{align*}
    \op{CI} \coloneqq \prns{\hat\gamma \pm \Phi^{-1}\prns{1-\frac{\rho}{2}}\sqrt{\hat\sigma^2\prns{\mathcal{W}_{m, N}}/N}},
\end{align*}
obeys that 
\begin{align*}
    \Prb{\gamma^* \in \op{CI}} \to 1 - \rho.
\end{align*}
\end{theorem}
In \cref{thm: mu-asymp,thm: CI}, we derive the asymptotic property of estimator $\hat\gamma$ and the associated confidence intervals, both based on a generic weighting matrix $\mathcal{W}_{m, N}$ with a limit $\mathcal{W}_{m, \infty}$. We may wonder how to choose this weighting matrix. In standard GMM estimation, it is well known that the asymptotically optimal weighting matrix is the inverse moment covariance matrix \citep{hansen1982large}. In the following theorem, we show that the inverse moment covariance matrix is also optimal for regularized GMM estimation. 
\begin{theorem}\label{thm: GMM-optimality}
For any positive definite matrix $\mathcal{W}_{m, \infty}$, we have that 
\begin{align*}
    \sigma^2\prns{\mathcal{W}_{m, \infty}} \ge \sigma^2\prns{\Sigma_{m}^{-1}}, ~~ \Sigma_{m} \coloneqq \Eb{m\prns{O; \theta^*_{\min}}m^\top\prns{O; \theta^*_{\min}}},
\end{align*}
where the estimating function $m$ is defined in \cref{eq: equations} and the matrix $\Sigma_{m}$ is invertible. 
\end{theorem}
\cref{thm: GMM-optimality} shows that among the class of regularized GMM estimators given in \cref{eq: GMM,eq: mean-est}, the ones that achieve the smallest asymptotic variance have $\Sigma_{m}^{-1}$ as the limit of their weighting matrix in \cref{eq: GMM}. 
Note that the optimal asymptotic variance $\sigma^2\prns{\Sigma_{m}^{-1}}$ depends on the minimal bridge function coefficient $\theta^*_{\min}$ that we choose to target. 
If we target a different one, then the optimal asymptotic variance will also change accordingly, and the best target  depends on the unknown covariance of pre-treatment idiosyncratic errors (see \cref{sec: app-reg} \cref{prop: asymp-variance}).

We can construct an asymptotically optimal estimator by a two-stage approach:
\begin{enumerate}
    \item First, construct an  estimator $\hat\theta_{\op{init}}$ by solving \cref{eq: GMM} with a known weighting matrix, such as the $\prns{T_1 + d} \times \prns{T_1 + d}$ identity matrix; 
    \item Second, construct the covariance matrix estimator $\hat\Sigma_m = \En{m({O; \hat\theta_{\op{init}}})m^\top({O; \hat\theta_{\op{init}}})}$ and then obtain the estimator $\hat\theta$ by solving \cref{eq: GMM} with the weighting matrix $\hat\Sigma_m^{-1}$. Finally, plug $\hat\theta$ into \cref{eq: mean-est} to solve for the estimator $\hat\gamma$. 
\end{enumerate}

\section{Discussions}\label{sec: discuss}
\paragraph*{Effects on Multiple Time Periods. } In previous sections, we only study the treatment effects on the outcome in period $t = 0$. We can also consider effects on outcomes aggregated from multiple time periods, \eg, the ATT parameter $\frac{1}{L}\sum_{t=0}^L \Eb{Y_{t}\prns{1} - Y_{t}\prns{0} \mid A = 1}$ for an integer $0 < L < T_1$. In this case, we can straightforwardly adapt the results in \cref{sec: identification,sec: estimation} to the  identification and estimation of the aggregated counterfactual mean  $\frac{1}{L}\sum_{t=0}^L \Eb{Y_{t}\prns{0} \mid A = 1}$. In particular, the corresponding bridge functions still have the form in \cref{eq: bridge-form}, with the set of coefficients being 
$$\Theta^* = \braces{\theta^*: \mb{V}_\pre^\top \theta_1^* = \frac{1}{L}\sum_{t=0}^L V_t, \theta_2^* = {b_0 - \mb{B}_\pre^\top\theta^*_1}}.$$ 
For the rest of part, we only need to redefine $\post = \braces{L+1, \dots, T_1}$ and revise the assumptions and estimation procedures accordingly. 


\paragraph*{Counterfactual Mean for the Whole Population.} In previous sections, we focus on the counterfactual mean for a subpopulation, \ie, the treated units. We can also study the counterfactual mean for the whole population, \ie, 
\begin{align}\label{eq: gamma-whole}
    \tilde\gamma^* = \Eb{Y_0(0)} = \Eb{h\prns{Y_\pre, X; \theta^*}}, ~~ \forall \theta^* \in \Theta^* \text{ given in \cref{eq: bridge-theta}.}  
\end{align}
This means that we need to solve the following moment equations:
\begin{align}
    &\Eb{\tilde g\prns{O; \theta, \tilde \gamma}} = 0, ~~ \Eb{m\prns{O; \theta}} = 0 , ~ \text{ for } O = \prns{X, A, Y_{\pre}, Y_{0}, Y_{\post}}, \label{eq: moment-whole} \\
    &\text{where } \tilde g\prns{O; \theta, \tilde \gamma} \coloneqq  \tilde W^\top \theta - \tilde \gamma, \notag \\
    &\phantom{\text{where }} m\prns{O; \theta} \coloneqq \prns{1-A}\prns{Y_{0}  -  
    \prns{\theta^{\top}_1 Y_{\pre} + \theta_2 X}
    }
        \begin{bmatrix}
        Y_{\post} \\
        X
        \end{bmatrix}. \notag 
\end{align}
Note that here the moment function $\tilde g$ and $m$ are generally correlated, while the previous moment function $g$ and $m$ in \cref{eq: equations} have zero correlations. Because of the latter, when estimating the counterfactual mean for the treated, we can solve the two moment equations in \cref{eq: equations}  separately without loss of efficiency. However, when estimating the counterfactual mean for the whole population, explicitly accounting for the correlations between the two moment equations in \cref{eq: moment-whole} can improve the asymptotic estimation efficiency. See \cref{sec: app-opt} for  details. 


\paragraph*{Time-Varying Unmeasured Confounders.} In the linear factor model in \cref{assump: outcome}, although the unmeasured confounders have time-varying effects (characterized by matrices $\braces{V_t: t= - T_0, \dots, T_1}$ for all $t$), the confounders themselves are time-invariant. Now we drop this restriction, and consider unmeasured confounders that follow an autoregression model. Under this model, the unmeasured confounders are also time-varying, so the serial dependence structure of counterfactual outcomes also becomes more complex than before. 

\begin{assumption}[Time-varying Confounders]\label{assump: time-varying}
\begin{align}
 Y_{i, t}(0) &= V_t^\top U_{i, t} + b^\top_t X_{i} + \epsilon_{i, t}, \notag \\
 U_{i, t} &= \Gamma_{t-1} U_{i, t-1} + \eta_{i, t-1}, ~~ \forall i, t, \label{eq: time-varying} 
\end{align}
where the dynamics of unmeasured confounders $U_{i, t} \in \R{r}$ are governed by transition matrices $\Gamma_{t-1}$ and mean-zero  innovations $\eta_{i, t-1} \in \R{r \times r}$. Here $\prns{X_i, U_{i, t}, \epsilon_{i, t}, \eta_{i, t-1}: -T_0 \le t \le T_1}$ are i.i.d draws from a common population  $\prns{X, U_{t}, \epsilon_{t}, \eta_{t-1}: -T_0 \le t \le T_1}$ with finite second order moments. Without loss of generality, assume that the first component of $X$ is the constant $1$ (\ie, intercept). For any $t$ and $s$, we assume  $\epsilon_t \perp \prns{U_s, X, A}$, $A \perp {U_t} \mid X, U_0$, and $\Prb{A = a \mid X, U_0} > 0$ for $a = 0, 1$. 
\end{assumption}
In \cref{assump: time-varying}, we assume that $A  \perp U_t \mid X, U_0$, which means that the treatment assignment $A$ only depends on confounders $U_0$ at the time of the treatment, but not any past or future confounder. 
This condition implies that $A \perp Y_{\pre}  \mid X, U_0$, an analogue of \cref{eq: no-direct-effect}, ensuring that bridge functions can be applied to the treated units to recover the counterfactual mean parameter (see \cref{lemma: bridge}).

Note that we should not deal with this model by redefining $U_i = \prns{U_{i, -1}, \dots, U_{i, -T_0}} \in \R{T_0 r}$ and casting it as a special example of \cref{assump: outcome}. Otherwise the dimensionality of $U_i$ exceeds the number of pre-treatment outcomes $T_0$ unless $r = 1$, so the rank condition in \cref{lemma: bridge} is violated. Therefore, we have to handle the model in \cref{assump: time-varying} directly. 

In the following theorem, we show that under additional assumptions for the unmeasured confounders and the transition innovations, this model again has linear bridge functions. 

\begin{lemma}\label{lemma: time-varying-bridge-simple}
Let \cref{assump: time-varying} hold and further assume the following assumptions:
\begin{itemize}
\item For any $t \in \pre$, $\Eb{\eta_t \mid U_0, X, A = 0} = \Eb{\eta_t \mid U_0, A = 0}$ and $\Eb{U_{-T_0} \mid U_0, A = 0, X} = \Eb{U_{-T_0} \mid U_0, A = 0}$ are linear functions of $U_0$;
\item The $r \times r$ matrix $\Sigma_0 = \Eb{U_{0}U_0^\top \mid A = 0}$ is invertible. 
\item The following $T_0 \times r$ matrix has full column rank equal to $r$:
\begin{align*}
\begin{bmatrix}
V_{-1}^\top \Gamma_{\prns{-2}:\prns{-T_0}} {\Sigma_{U_{-T_0}}\Gamma^\top_{\prns{-1}:\prns{-T_0}}} + V_{-1}^\top \sum_{k=1}^{T_0 - 1}\Gamma_{\prns{-2}:\prns{-k}}{\Sigma_{\eta_{-1-k}}\Gamma^\top_{\prns{-1}:\prns{ - k}}} \\
\vdots \\
V_{-t}^\top \Gamma_{\prns{-t-1}:\prns{-T_0}} {\Sigma_{U_{-T_0}}\Gamma^\top_{\prns{-1}:\prns{-T_0}}} + V_{-t}^\top \sum_{k=1}^{-t + T_0}\Gamma_{\prns{-t-1}:\prns{-t-k+1}}{\Sigma_{\eta_{-t-k}}\Gamma^\top_{\prns{-1}:\prns{-t - k + 1}}} \\ 
\vdots \\
V_{-T_0}^\top \Gamma_{\prns{-T_0-1}:\prns{-T_0}} {\Sigma_{U_{-T_0}}\Gamma^\top_{\prns{-1}:\prns{-T_0}}}
\end{bmatrix}
\end{align*}
where $\Gamma_{t_1:t_2}$ is equal to $\prod_{t = t_1}^{t_2}\Gamma_t$ when $t_1 \ge t_2$ and the $r\times r$ identity matrix $I_r$ otherwise, and $\Sigma_{\eta_t}, \Sigma_{U_{-T_0}}$ are the covariance matrices of $\eta_t$ and $U_{-T_0}$ in the control population. 
\end{itemize}
Then there exist $\theta^* = \prns{\theta_1^*, \theta_2^*} \in \R{T_0 + d}$ such that 
\begin{align}
    &\Eb{ Y_{0}\prns{0} - \prns{\theta^{*\top}_1 Y_{\pre} + \theta_2^{*\top} X} \mid U_0, A = 0, X} = 0, \label{eq: bridge-2-simple}
\end{align}
and any such $\theta^*$ satisfies \cref{eq: identification}.
\end{lemma}
In \cref{lemma: time-varying-bridge-simple}, condition 1 is the key condition for the bridge functions to be \emph{linear} in the pre-treatment outcomes and the covariates.  
It assumes linear regression functions of transition innovations $\eta_t$ and the unmeasured confounders $U_{-T_0}$ with respect to unmeasured confounders $U_0$, in the control subpopulation.  
This condition is satisfied, for example, when the innovations and unmeasured confounders in the control subpopulation have a  joint normally distribution.
In \cref{sec: app-dynamic} \cref{lemma: time-varying-bridge-general}, we relax this condition to allow the two conditional expectations to depend on covariates $X$ as well, and prove a similar conclusion, albeit with more complex notations. 
Condition 2 rules out collinear components in the unmeasured confounders $U_0$. 
The matrix in condition 3 characterizes the effects of unmeasured confounders on the pre-treatment outcomes that can be attributed to $U_0$. 
We require this matrix to be invertible, as an anologue of the rank condition in \cref{lemma: bridge-prelim}, to ensure that pre-treatment outcomes are sufficiently informative proxies for confounders $U_0$.
In \cref{sec: app-dynamic}, we discuss that condition 1 in \cref{lemma: time-varying-bridge-simple} is not necessary for the existence of bridge functions, but without this condition bridge functions may not be linear, which is out of the scope of this paper.

In the following theorem, we further show that under conditions analogous to those in \cref{thm: identification}, we can again use post-treatment outcomes to learn the bridge functions. 
\begin{theorem}\label{thm: bridge-varying-obs}
Suppose that assumptions in \cref{lemma: time-varying-bridge-simple} hold. If $(\eta_t: t \in \post) \perp \prns{\eta_s: s \in \pre} \mid X$, $(\eta_t: t \in \post) \perp U_{-T_0} \mid X$, and $\epsilon_{\post} \perp \prns{\epsilon_{\pre}, \epsilon_0}$, then
\begin{align}\label{eq: post-indep}
Y_{\post} \perp \prns{Y_{\pre}, Y_0} \mid U_0, A = 0, X, 
\end{align}
and any $\theta^*$ that solves \cref{eq: bridge-2-simple} must also satisfy 
\begin{align}\label{eq: bridge-2-obs}
    \Eb{\begin{bmatrix}
    Y_{\post} \\
    X 
    \end{bmatrix}
    \prns{1-A}\prns{Y_{0} - \prns{\theta^{*\top}_1 Y_{\pre} + \theta_2^{*\top} X}}} = 0 .
\end{align}
Suppose that condition 1 in \cref{thm: identification} also holds for $U = U_0$, and the following matrix has rank $r+d$:
\begin{align*}
    \begin{bmatrix}
    \tilde{\mb{V}}_{\post} & \mb{B}_{\post}  \\
    0 & I  
    \end{bmatrix} \in \R{\prns{T_1 + d} \times \prns{r + d}},
\end{align*}
where the $t^{\text{th}}$ row of matrix $\tilde{\mb{V}}_{\post} \in \R{T_1 \times r}$ is equal to $V_t^\top \Gamma_{\prns{t-1}:0}$ for $t = 1, \dots, T_1$.  

Then $\theta^*$ satisfies \cref{eq: bridge-2-simple} if and only if it satisfies \cref{eq: bridge-2-obs}, and $\gamma^*$ in \cref{eq: target} is identifiable. 
\end{theorem}
In \cref{thm: bridge-varying-obs}, we assume that confounding innovations in post-treatment periods are conditionally independent with those in pre-treatment periods and the initial unmeasured confounders $U_{-T_0}$. 
The former condition holds trivially when the innovations are serially independent.  
Under these two condition, the dependence between unmeasured confounders in the  pre-treatment periods and those in the post-treatment periods is fully mediated by $U_0$ and covariates $X$, \ie, $\prns{U_t, t \in \post} \perp \prns{U_s, s \in \pre} \mid U_0, X$, which in turn ensures \cref{eq: post-indep}. According to \cref{lemma: bridge}, \cref{thm: bridge-varying-obs} shows the identifiability of the target counterfactual mean parameter $\gamma^*$. Then we can estimate it by applying the regularized GMM estimation procedure in \cref{sec: estimation} to the moment equation in \cref{eq: bridge-2-obs}. 

Below \cref{thm: identification}, we mention that when unmeasured confounders are time-invariant and idiosyncratic errors are serially independent, it is possible to use some extra pre-treatment outcomes not in $Y_{\pre}$ to learn the bridge functions. 
However, this is no longer feasible when unmeasured confounders are time-varying. 
In this time-varying setting, unmeasured confounders in pre-treatment periods follow the autoregressive process in \cref{eq: time-varying}, so their dependence cannot be fully mediated by covariates $X$ and confounders $U_0$ taking place \emph{after} the pre-treatment periods. 
As a result, pre-treatment outcomes are all dependent even conditionally on $U_0$ and $X$, so we cannot use extra pre-treatment outcomes in place of $Y_{\post}$ in \cref{eq: bridge-2-obs}. 
Therefore,  when unmeasured confounders are time-varying, we must use \emph{only} post-treatment outcomes to learn the bridge functions.

\paragraph{Connections to Negative Controls.}

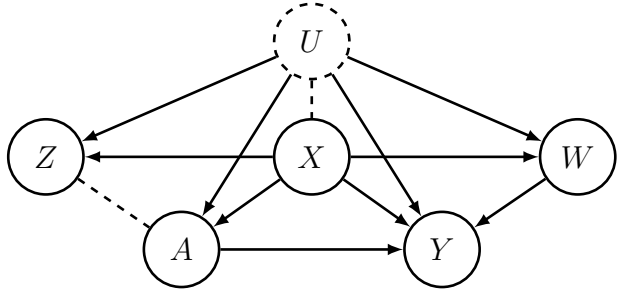
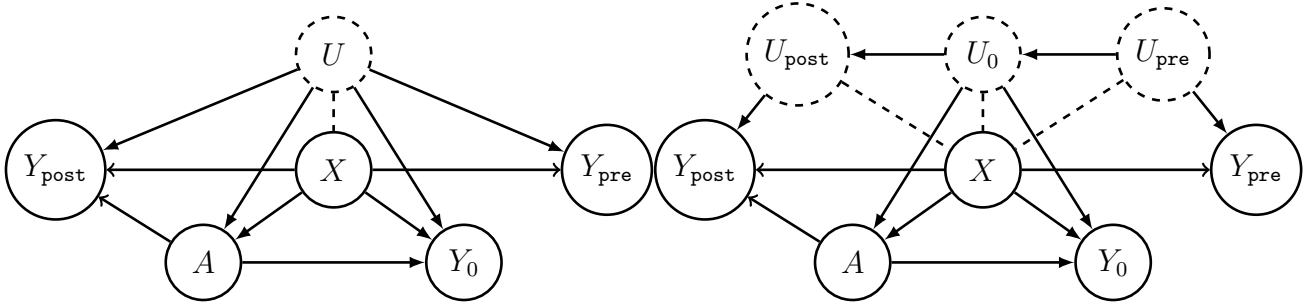
\begin{figure}[t]
    \begin{subfigure}[b]{\linewidth}{
    \centering
    \begin{tikzpicture}
    \node[draw, circle, text centered, minimum size=1cm, line width= 1](x) {$X$};
    \node[draw, circle, above=0.5 of x,
    text centered, minimum size=1cm, dashed,line width= 1] (u) {$U$};
    \node[draw, circle,left=2.5 of x, minimum size=1cm, text centered,line width= 1] (z) {$Z$};
    \node[draw, circle,right=2.5 of x, minimum size=1cm, text centered,line width= 1] (w) {$W$};
    \node[draw, circle, below left = 0.5 and 1 of x, minimum size=1cm, text centered,line width= 1] (a) {$A$};
    \node[draw, circle, below right = 0.5 and 1 of x, minimum size=1cm, text centered,line width= 1] (y) {$Y$};
    
    \draw[-latex, line width= 1] (x) -- (a);
    \draw[-latex, line width= 1] (x) -- (y);
    \draw[-latex, line width= 1] (x) -- (z);
    \draw[-latex, line width= 1] (x) -- (w);
    \draw[line width= 1, dashed] (u) -- (x);
    \draw[-latex, line width= 1] (u) -- (z);
    \draw[-latex, line width= 1] (u) -- (w);
    \draw[-latex, line width= 1] (u) -- (a);
    \draw[-latex, line width= 1] (u) -- (y);
    \draw[-latex, line width= 1] (a) -- (y);
    \draw[line width= 1, dashed] (z) -- (a);
    \draw[-latex, line width= 1] (w) -- (y);
    \end{tikzpicture}
    \caption{Negative controls.}
    \label{fig: nc}
    }
    \end{subfigure}
\newline
    \begin{subfigure}[b]{0.48\linewidth}{
    \centering
    \begin{tikzpicture}
    \node[draw, circle, text centered, minimum size=1cm, line width= 1]  (x) {$X$};
    \node[draw, circle, above=0.5 of x,
    text centered, minimum size=1cm, dashed,line width= 1] (u) {$U$};
    \node[draw, circle,left=2.5 of x, minimum size=1cm, text centered,line width= 1] (z) {$Y_{\post}$};
    \node[draw, circle,right=2.5 of x, minimum size=1cm, text centered,line width= 1] (w) {$Y_{\pre}$};
    \node[draw, circle, below left = 0.5 and 1 of x, minimum size=1cm, text centered,line width= 1] (a) {$A$};
    \node[draw, circle, below right = 0.5 and 1 of x, minimum size=1cm, text centered,line width= 1] (y) {$Y_0$};
    
    \draw[-latex, line width= 1] (x) -- (a);
    \draw[-latex, line width= 1] (x) -- (y);
    \draw[->, line width= 1] (x) -- (z);
    \draw[->, line width= 1] (x) -- (w);
    \draw[line width= 1, dashed] (u) -- (x);
    \draw[-latex, line width= 1] (u) -- (z);
    \draw[-latex, line width= 1] (u) -- (w);
    \draw[-latex, line width= 1] (u) -- (a);
    \draw[-latex, line width= 1] (u) -- (y);
    \draw[-latex, line width= 1] (a) -- (y);
    \draw[->, line width= 1] (a) -- (z);
    \end{tikzpicture}
    }
    \caption{Panels with time-invariant counfounders.}
    \label{fig: ours-nc-1}
    \end{subfigure}
\hspace{0.04\textwidth}
    \begin{subfigure}[b]{0.48\linewidth}{
    \begin{tikzpicture}
    \node[draw, circle, text centered, minimum size=1cm, line width= 1] (x) {$X$};
    \node[draw, circle, above=0.5 of x,
    text centered, minimum size=1cm, dashed,line width= 1] (u) {$U_0$};
    \node[draw, circle, right=1.25 of u,
    text centered, minimum size=1cm, dashed,line width= 1] (u-) {$U_{\pre}$};
    \node[draw, circle, left=1.25 of u,
    text centered, minimum size=1cm, dashed,line width= 1] (u+) {$U_{\post}$};
    \node[draw, circle,left=2.5 of x, minimum size=1cm, text centered,line width= 1] (z) {$Y_{\post}$};
    \node[draw, circle,right=2.5 of x, minimum size=1cm, text centered,line width= 1] (w) {$Y_{\pre}$};
    \node[draw, circle, below left = 0.5 and 1 of x, minimum size=1cm, text centered,line width= 1] (a) {$A$};
    \node[draw, circle, below right = 0.5 and 1 of x, minimum size=1cm, text centered,line width= 1] (y) {$Y_0$};
    
    \draw[-latex, line width= 1] (x) -- (a);
    \draw[-latex, line width= 1] (x) -- (y);
    \draw[->, line width= 1] (x) -- (z);
    \draw[->, line width= 1] (x) -- (w);
    \draw[line width= 1, dashed] (u) -- (x);
    \draw[line width= 1, dashed] (u+) -- (x);
    \draw[line width= 1, dashed] (u-) -- (x);
    \draw[-latex, line width= 1] (u+) -- (z);
    \draw[-latex, line width= 1] (u-) -- (w);
    \draw[-latex, line width= 1] (u) -- (a);
    \draw[-latex, line width= 1] (u-) -- (u);
    \draw[-latex, line width= 1] (u) -- (u+);
    \draw[-latex, line width= 1] (u) -- (y);
    \draw[-latex, line width= 1] (a) -- (y);
    \draw[->, line width= 1] (a) -- (z);
    \end{tikzpicture}
    }
    \caption{Panels with time-varying confounders.}
    \label{fig: ours-nc-2}
    \end{subfigure}
\caption{Causal diagrams for negative controls and their instances in panel data setting. Here dashed circles represent unobserved variables, and dashed lines represent causal edges (along either direction) that may or may not exist.}
\end{figure}

Recently, a series of works propose a negative control framework to deal with the challenge of unmeasured confounding \citep[\eg, ][]{cui2020semiparametric,tchetgen2020introduction,miao2018identifying,deaner2021proxy,shi2020multiply}. This framework requires two types of proxy variables for the unmeasured confounders: negative control outcomes $W$ and  negative control treatments $Z$. These proxy variables are informative in that they are dependent with the unmeasured confounders.
 Importantly, they have special causal relations with other variables: the negative control outcomes $W$ \emph{cannot} be directly caused by the primary treatment of interest, and the negative control treatments $Z$ cannot directly cause either the negative control outcomes or the primary outcome of interest. 
 In \cref{fig: nc}, with slight abuses of notations, we show a typical causal diagram of negative controls when studying  the causal effect of a primary treatment $A$ on a primary outcome $Y$ in presence of unmeasured confounders $U$. 

In our panel data setting, we have natural candidates for the negative control variables (see \cref{fig: ours-nc-1,fig: ours-nc-2} for illustrations): the pre-treatment outcomes $Y_{\pre}$ can be considered as negative control outcomes $W$ and the post-treatment outcomes $Y_{\post}$ can be considered as  negative control treatments $Z$. Indeed, pre-treatment outcomes are realized before the treatment takes place so they may not be directly caused by the treatment $A$, and post-treatment outcomes are realized after the primary outcome $Y_0$ and pre-treatment outcomes $Y_\pre$, so they may not directly cause the latter.
These conditions are formalized in \cref{eq: no-direct-effect,eq: post-indep0} respectively.

In \cref{sec: identification}, we identify the causal parameters through bridge functions. This concept is originally proposed in the negative control literature \citep{miao2018a,cui2020semiparametric,tchetgen2020introduction,deaner2021proxy,Miao2016}, where some also apply it to panel data. 
The connection between negative controls and (nonlinear) difference-in-differences is also noted by \citep{Sofer2016}. 
Our identification results in \cref{sec: identification} can be viewed as an application of the general negative control identification strategy to linear factor models in panel data. 
By focusing on this important model, we explicitly characterize when bridge functions exist and when post-treatment outcomes can be used to learn them, shedding light on the abstract conditions assumed in previous literature. 
More importantly, our analyses elucidate the benefits  and prices of the negative control identification relative to existing ones: using negative controls only needs  finite time periods, but has to additionally assume some serial independence assumptions on the idiosyncratic errors (see \cref{lemma: post-treatment}).

However, our estimation method in \cref{sec: estimation} is distinct from those in previous negative control literature, in that we use  regularization to deal with the prevalent problem of nonunique bridge functions (see discussions below \cref{lemma: nonunique}). 
Remarkbly, our estimator enjoys $\sqrt{n}$-consistency, asymptotic normality, and simple plug-in confidence intervals, regardless of whether bridge functions are unique. 
In contrast, previous negative control literature often explicitly or implicitly assume a unique bridge function to facilitate estimation \citep[\eg, ][]{miao2018a,cui2020semiparametric,shi2020multiply,QiZhengling2021PLfI,pmlr-v139-mastouri21a,SinghRahul2020KMfU,GhassamiAmirEmad2021MKML}. \cite{deaner2021proxy,kallus2021proxy} recognize this problem and derive convergence rates of their proposed causal effect estimators even when bridge functions are nonunique, but they either 
rely on inefficient sample splitting and computationally intensive bootstrap methods, or 
do not have inferential procedures.
Our regularization approach that targets the minimal bridge function provides a new solution to this important problem. Extending it to more general negative control settings is an exciting future direction.


\section{Conclusions}\label{sec: conclusion}
In this paper, we study the identification and estimation of average causal effects in panel data under a linear factor model. 
Previous regression-based methods (\eg, synthetic controls) and matrix estimation methods (\eg, matrix completion) all require both the number of units and the number of time periods to grow to infinity to consistently estimate the causal effects. 
So they may not be suitable when only observations in a relatively small number of time periods are available. 

Motivated by the differencing transformation in the DID estimator for the simpler TWFE model, we propose to identify the causal effects using bridge functions, namely some tansformations of pre-treatment outcomes to control for unmeasured confounding. 
Learning bridge functions for causal effect estimation requires sufficiently many informative pre-treatment and post-treatment outcomes to account for all unmeasured confounders, but not infinitely many. 
Noting that bridge functions are often nonunique in practice, we propose an novel regularized GMM estimator to target the minimal bridge function. We prove that the resulting causal effect estimators and confidence intervals have  desirable asymptotic guarantees, regardless of whether bridge functions are unique or not. 
Our proposal thus features a novel approach to handle unmeasured confounding in panel data with observations in only a limited number of time periods. 

\bibliographystyle{plainnat}
\bibliography{semiparametric}

\appendix
\section{More Details on Existing Approaches}\label{sec: app-existing}
In this section, we compare existing approaches in \cref{sec: previous} in more detail. We focus on estimating the sample counterfactual mean for the treated, so that different methods are more comparable:  
\begin{align*}
 {\gamma}^{S} = \frac{1}{N_1}\sum_{i \in \T} Y_{i, t}\prns{0}.
 \end{align*} 
For simplicity, we also ignore covariates in this section, \ie, $b_t = 0$ for any $t$ in \cref{assump: outcome}.

We first prove \cref{eq: vertical,eq: transform-2} that motivate the horizontal and vertical regressions.

\begin{lemma}\label{lemma: horizontal-vertical-form}
 Suppose $\mb{V}_{\pre}$ has full column rank, \ie,  $\op{Rank}\prns{\mb{V}_{\pre}} = r$. Then for any $\theta^*_1 \in \mathbb R^{T_0}$ such that $\mb{V}^\top_{\pre}\theta^*_1 = V_{0}$, 
 \begin{align*}
Y_{i, 0}\prns{0}  =  \theta_1^{*\top} Y_{i, \pre} + \epsilon_{i, 0} - \theta_1^{*\top} \epsilon_{i, \pre}, ~~ \forall i.
\end{align*}
Suppose $\mb{U}_\C$ has full column rank, \ie,  $\op{Rank}\prns{\mb{U}_\C} = r$. Then for any $w^* \in \R{N_0}$ such that $\mb{U}^\top_\C w^* = \frac{1}{N_1}\sum_{i \in \T} U_i$, 
 \begin{align*}
 \frac{1}{N_1}\sum_{i \in \T} Y_{i, t}\prns{0}  = w^{*\top} Y_{\C, t} +  {\frac{1}{N_1}\sum_{i \in \T} \epsilon_{i, t} - w^{*\top} \epsilon_{\C,t}}, ~~ \forall t. 
\end{align*}
\end{lemma}
In \cref{lemma: horizontal-vertical-form}, we require $\mb{V}_{\pre}$ or $\mb{U}_\C$ to have full column rank. As we discuss below \cref{lemma: bridge-prelim}, the rank condition on $\mb{V}_{\pre}$ means that pre-treatment outcomes are informative proxy variables for the unmeasured confounders, and the number pre-treatment outcomes $T_0$ has to be no smaller than the number of unmeasured confounders $r$. Similarly, the rank condition on $\mb{U}_\C$ requires that control units' outcomes are informative proxies for the time-varying coefficients of unmeasured confounders, \ie, $Y_{\C, t}$  captures the information of $V_t$ in all directions for any $t$. This condition holds with high probability when components of the unmeasured confounders are not collinear in the control unit  subpopulation. It also requires the number of control units $N_0$ to be no smaller than the number of unmeasured confounders $r$. 

In \cref{sec: previous}, we show that the two representations in \cref{lemma: horizontal-vertical-form} motivate horizontal regression estimator and vertical regression estimator respectively, and many existing methods can be viewed as variants of these two approaches. 
Here we further show the properties of these two approaches by analyzing the bias and variance of the resulting causal estimators. 
For simplicity, we focus on the simplest forms of these two approaches based on linear regressions without any regularization. 
\begin{description}
\item[Horizontal Regression.] First  run a linear regression of $Y_{i, 0}$ against $Y_{i, \pre}$, based on data for control units up to time $t = 0$ to obtain the coefficient estimator $\hat\theta_{\HR}$, where we assume $N_0 > T_0$ for $\hat\theta_{\HR}$ to be well-defined. Then estimate $\gamma^S$ by 
\[
\hat \gamma_{\HR} = \frac{1}{N_1}\sum_{i\in\T}\hat\theta_\HR^\top Y_{i, \pre}, ~~ \hat\theta_\HR = \prns{\mb{Y}^\top_{\C, \pre}\mb{Y}_{\C, \pre}}^{-1}\mb{Y}^\top_{\C, \pre}Y_{\C, 0}.
\]
\item[Vertical Regression.] First run a linear regression of $\frac{1}{N_1}\sum_{i\in \T}Y_{i,t}$ against $Y_{\C,t}$ based on data up to time $t = 1$ to and obtain the coefficient estimator $\hat w_{\VR}$, where we assume $T_0 > N_0$ for $\hat w_{\VR}$ to be well-defined. Then estimate $\gamma^S$ by 
\[
\hat \gamma_{\VR} = \hat w_{\VR}^\top Y_{\C, 0}, ~~  \hat w_{\VR} = \prns{\mb{Y}_{\C, \pre}\mb{Y}_{\C, \pre}^\top}^{-1}\mb{Y}_{\C, \pre}\prns{\frac{1}{N_1}\sum_{i \in \T} Y_{\T, \pre}}.
\]
\end{description}

 We denote $\overline U_{\T} = \frac{1}{N_1}\sum_{i\in\T}U_{i}$, $\overline \epsilon_{\T, \pre} = \frac{1}{N_1}\sum_{i\in\T}\epsilon_{i, \pre}$, $\overline \epsilon_{\T, 0} = \frac{1}{N_1}\sum_{i\in\T}\epsilon_{i, 0}$, and $\Sigma_{U \mid 0} = \op{Cov}\prns{U, U \mid A = 0}$. For any matrix $A$, we denote its smallest and largest singular values as $\sigma_{\min}\prns{A}$ and $\sigma_{\max}\prns{A}$ respectively. For simplicitly, we also assume that $\Eb{\epsilon_t^2}$ is a constant $\sigma_{\epsilon}^2$ that does not vary with $t$. 
\begin{lemma}\label{lemma: horizontal-bias}
Suppose $\mb{V}_{\pre}$ has full column rank, \ie,  $\op{Rank}\prns{\mb{V}_{\pre}} = r$, and $\Sigma_{U \mid 0}$ is an invertible matrix. When $N_0 \to \infty$, 
\begin{align*}
\abs{\hat \gamma_{\HR} - \gamma^S  - \mathcal B_{\HR} - \mathcal V_{\HR}} \to 0,
\end{align*}
where the bias term $\mathcal B_{\HR}$ and variance term $\mathcal V_{\HR}$ are 
\begin{align*}
\mathcal B_{\HR} = - V_0^\top \prns{\frac{1}{\sigma^2_\epsilon}\mb{V}_{\pre}^\top \mb{V}_{\pre}\Sigma_{U \mid 0} + I }^{-1}\overline U_{\T} , ~~ 
\mathcal V_{\HR} =   \overline\epsilon_{\T, \pre}^\top \mb{V}_{\pre}\prns{\sigma_\epsilon^2 \Sigma_{U\mid 0}^{-1} + \mb{V}_{\pre}^\top\mb{V}_{\pre}}^{-1}V_0  - \overline \epsilon_{\T, 0}.
\end{align*}
Moreover, as $\sigma_{\min}\prns{\mb{V}_\pre} \to \infty$, 
\begin{align*}
|\mathcal B_\HR| 
    &\le \frac{\sigma^2_\epsilon}{\sigma_{\min}\prns{\Sigma_{U \mid 0}}\sigma_{\min}^2\prns{\mb{V}_{\pre}} - \sigma^2_\epsilon}\|\overline U_{\T}\|\|V_0\| \to 0 \\
\abs{\mathcal V_{\HR} + \overline \epsilon_{\T, 0}}
    &\le \frac{\sigma_{\max}\prns{\Sigma_{U\mid 0}}\sigma_{\min}\prns{\mb{V}_{\pre}}}{\sigma_{\max}\prns{\Sigma_{U \mid 0}}\sigma^2_{\min}\prns{\mb{V}_{\pre}} - \sigma_{\epsilon}^2}\|\overline\epsilon_{\T, \pre}\|\|V_0\| \to 0.  
\end{align*}  
\end{lemma}

\cref{lemma: horizontal-bias} shows that the horizontal regression estimator $\hat\gamma_{\HR}$ has nonvanishing bias $\mathcal B_\HR$ when only the sample size in the horizontal regression (\ie, the number of control units $N_0$) grows to infinity. However, this bias shrinks to $0$ when the smallest singular value of matrix  $\mb{V}_{\pre}$ can converge to infinity. This requires that the dimension of regressors in the horizontal regression (\ie, number of pre-treatment outcomes $T_0$) grows to infinity, and  asymptotically there are  infinitely many pre-treatment outcomes with nonzero coefficients on each of the unmeasured confounders \citep[][Corollary 1.1]{NEURIPSKallus2018}. Conversely, if there are only a fixed number of pre-treatment outcomes, then the horizontal regression estimator generally has persistent bias. 

\begin{lemma}\label{lemma: vertical-bias}
Suppose $\mb{U}_\C$ has full column rank, \ie,  $\op{Rank}\prns{\mb{U}_\C} = r$ and $\frac{1}{T_0}\sum_{t = -T_0}^{-1}V_tV_t^\top \to \overline V^\otimes$ where $\overline V^\otimes$ has full rank $r$. 
Conditionally on $\mb{U}_\C, \mb{U}_\T$, when $T_0 \to \infty$,
\[
\abs{\hat \gamma_{\VR} - \gamma^S  - \mathcal B_{\VR}- \mathcal V_{\VR}} \to 0,
\]
where the bias term $\mathcal B_{\VR}$ and variance term $\mathcal V_{\VR}$ are 
\begin{align*}
\mathcal B_{\VR} = - V_0^\top \prns{\frac{1}{\sigma^2_\epsilon}\mb{U}_\C^\top\mb{U}_\C \overline V^\otimes  + I}^{-1}\overline{U}_\T, ~~ \mathcal V_{\VR} =  \epsilon_{\C, 0}^\top \mb{U}_\C\prns{\mb{U}_\C^\top \mb{U}_\C + \sigma^2_\epsilon {\overline V^{\otimes-1}}}^{-1}\overline{U}_\T - \overline\epsilon_{\T, 0}
\end{align*}
Moreover, if $\liminf_{N_0 \to \infty}\sigma_{\min}\prns{\sum_{i \in \C}U_iU_i^\top/N_0} > 0$, and $N_0 \to \infty$, then 
\begin{align*}
\abs{\mathcal B_{\VR}} 
    &\le \frac{1}{N_0}\frac{\sigma_\epsilon^2}{\sigma_{\min}\prns{\sum_{i \in \C}U_iU_i^\top/N_0}\sigma_{\min}\prns{\overline V^\otimes} - \sigma_{\epsilon}^2/N_0}\|V_0\|\|\overline U_\T\| \to 0, \\
\abs{\mathcal V_{\VR} + \overline\epsilon_{\T, 0}}  &\le \prns{\sum_{i \in \C}\epsilon_{i, 0}U_{i}/N_0}\frac{\sigma_\epsilon^2}{\sigma_{\min}\prns{\sum_{i \in \C}U_iU_i^\top/N_0}\sigma_{\min}\prns{\overline V^{\otimes}} - \sigma_{\epsilon}^2/N_0}\|\|\overline V^{\otimes}\|\|\overline U_\T\| \to 0.
\end{align*}
\end{lemma}

\cref{lemma: vertical-bias} shows that the vertical regression estimator  $\hat\gamma_{\VR}$ also has a nonvanishing bias when only the sample size in the vertical regression (\ie, the number of pre-treatment outcomes $T_0$) grows to infinity. However, this bias shrinks to $0$ when the dimension of regressors in the vertical regression (\ie, the number of control units $N_0$) grows to infinity. Note that we require the smallest singular value of ${\sum_{i \in \C}U_iU_i^\top/N_0}$ to be bounded away from $0$, which holds when $\Sigma_{U \mid 0}$ has full rank, \ie, the unmeasured confounders are not collinear. 

\paragraph*{Matrix Estimation.} As we discuss in \cref{sec: previous}, some previous matrix estimation literature learns the linear factor model structure in \cref{assump: outcome} directly. 
To illustrate the property of matrix estimation 
methods, we use the factor model approach in \cite{ruoxuan2020} as an example, and for simplicity, we only use pre-treatment data.
Similar conclusions hold when we also use post-treatment data, noting that in post-treatment periods the counterfactual outcomes of all treated units (which constitute at least a constant fraction of units according to \cref{assump: treatment}) are missing.

The approach in \cite{ruoxuan2020} consists of the following four steps:
\begin{enumerate}
\item Estimate the second order moment of the observed outcome by $\tilde\Sigma = [\tilde\Sigma_{ij}]_{i, j=1}^N$ where 
    \begin{align*}
    \tilde\Sigma_{ij} =
    \begin{cases}
    \frac{1}{T_0 + 1}\sum_{t=-T_0}^{0}Y_{i, t}Y_{j, t} & \text{if } i \in \C, j \in \C, \\
    \frac{1}{T_0}\sum_{t=-T_0}^{-1}Y_{i, t}Y_{j, t} & \text{if } i \in \T \text{ or } j \in \T. \\
    \end{cases}
     \end{align*} 
\item Compute the eigen-decomposition of $\frac{1}{N}\tilde\Sigma$, and obtain the eigenvector matrix $\tilde{\mb{U}} \in \R{N \times r}$ corresponding to the top $r$ eigenvalues:
\begin{align*}
\frac{1}{N} \tilde \Sigma \tilde{\mb{U}} = \tilde{\mb{U}}\tilde{\Lambda},
\end{align*}
where $r$ is assumed to be known or it can be consistently estimated by the approach in \cite{Bai2002}.
\item Regress the observed components of $\prns{Y_{1, t}, \dots, Y_{N, t}}$ against $\tilde U_r$ to estimate $V_0$:
\begin{align*}
\tilde V_t = 
\prns{\sum_{i\in\C} \tilde U_{i}\tilde U_{i}^\top}^{-1}\prns{\sum_{i\in\C} \tilde U_{i} Y_{i, t}}.
\end{align*}
\item Estimate $\gamma^S$ by the following imputation estimator:
\begin{align*}
\hat\gamma_{\op{F}} = \frac{1}{N_1}\sum_{i\in\T} \tilde U_i^\top \tilde V_0.
\end{align*}
\end{enumerate}

\begin{lemma}\label{lemma: factor}
Let $\Sigma_U$ and $\Sigma_{U \mid 0}$ be the marginal covariance matrix of $U$ and conditional covariance matrix of $U$ given $A = 0$ and assume that they are both positive definite. Suppose the following conditions hold:
\begin{enumerate}
\item $\frac{1}{T_0}\sum_{t = -T_0}^{-1}V_tV_t^\top \to \overline V^\otimes$ where $\overline V^\otimes$ has full rank $r$. 
\item The fourth moment of $U$ and eighth moment of $\epsilon_{i, t}$ are finite. 
\item  The eigenvalues of $\Sigma_U V^\otimes$ are distinct.
\end{enumerate}
Then for any $i = 1, \dots, N$,
\begin{align}\label{eq: factor-1}
\abs{\tilde U_i^\top \tilde V_0 - U_i^\top V_0} = O_p\prns{N_0^{-1/2} + T_0^{-1/2}},
\end{align}
and 
\begin{align}\label{eq: factor-2}
\abs{\hat\gamma_{\fac} - \overline{\gamma}^S + \overline{\epsilon}_{\T, 0}} = O_p\prns{N_0^{-1/2} + T_0^{-1/2}}.
\end{align}
\end{lemma}
\cref{lemma: factor} show that we need both $N_0 \to \infty$ and $T_0 \to \infty$ to estimate the unobserved factors and the counterfactual mean parameter consistently.

\section{Optimal GMM Weighting Matrix}\label{sec: app-opt}
In \cref{thm: GMM-optimality}, we prove the asymptotic optimality of using the inverse covariance matrix as the weighting matrix among the class of regularized GMM estimators in \cref{eq: GMM}.
In this section, we generalizing this conclusion to a larger class of regularized GMM estimators. Through our analysis, we also show that the asymptotically optimal weighting matrix for 
the counterfactual mean of the whole population in \cref{eq: gamma-whole} has to additionally account for moment equation covariances.

Recall the moment equations that are used to estimate the bridge functions and counterfactual mean parameter:
\begin{align}
    &\Eb{g\prns{O; \theta, \gamma}} = 0, ~~ \Eb{m\prns{O; \theta}} = 0 , ~ \text{ for } O = \prns{X, A, Y_{\pre}, Y_{0}, Y_{\post}}, \notag \\
    &\text{where } g\prns{O; \theta, \gamma} \coloneqq A\prns{\theta^{\top}_1 Y_{\pre} + \theta_2 X - \gamma}, \notag \\
    &\phantom{\text{where }} m\prns{O; \theta} \coloneqq \prns{1-A}\prns{Y_{0}  -  
    \prns{\theta^{\top}_1 Y_{\pre} + \theta_2 X}
    }
        \begin{bmatrix}
        Y_{\post} \\
        X
        \end{bmatrix}. \notag 
\end{align}
In this section, we prove a stronger conclusion of \cref{thm: GMM-optimality} by considering a more general class of regularized GMM estimators:
\begin{align}\label{eq: gmm-full}
\prns{\tilde\theta, \tilde \gamma} 
    &= \argmin_{\theta, \gamma}\begin{bmatrix}
    \En{m\prns{O;\theta}}^\top & \En{ g\prns{O;\theta, \gamma}}^\top
    \end{bmatrix}
    \W_N^{-1}
    \begin{bmatrix}
    \En{m\prns{O;\theta}} \\ \En{ g\prns{O;\theta, \gamma}}
    \end{bmatrix} + \lambda_N \|\theta\|_2^2,
\end{align}
where the weighting matrix $\W_N^{-1}$ is positive definite and it converges (in probability) to a fixed limiting matrix that is also positive definite:
\begin{align*}
\W_N^{-1} \coloneqq \begin{bmatrix}
     \W_{11, N} &\W_{12, N} \\
    \W_{21, N} & \W_{22, N}
    \end{bmatrix}^{-1} 
\overset{\mathrm{p}}{\to }
\W_\infty^{-1} \coloneqq \begin{bmatrix}
     \W_{11, \infty} &\W_{12, \infty} \\
    \W_{21, \infty} & \W_{22, \infty}
    \end{bmatrix}^{-1}, ~~ \text{as } N \to \infty.
\end{align*}
It is easy to show that these estimators have an equivalent two-stage representation:
\begin{align}
&\tilde\theta = \min_{\theta} ~~ \En{m^\top\prns{O; \theta}}\W_{11, N}^{-1}\En{m\prns{O; \theta}} + \lambda_N\|\theta\|^2,\label{eq: 2-stage-1} \\
&\tilde\gamma \text{ solves } \En{g\prns{O; \tilde\theta, \gamma} - \W_{21, N}\W_{11, N}^{-1}m\prns{O; \tilde\theta}} = 0.\label{eq: 2-stage-2}
\end{align}
Here, because $\gamma^*$ is exactly identified by the moment equation $\En{g\prns{O; \theta, \gamma}} = 0$, the choice of $\W_{22, N}$ does not influence the estimator $\tilde \gamma$. Obviously, If we set $\W_{11, N}^{-1} = \W_{m, N}$ and $\W_{12, N} = \W_{21, N}^\top = \mb{0}_{\prns{T_1 + d} \times 1}$ in \cref{eq: 2-stage-1,eq: 2-stage-2}, then we obtain estimators in \cref{eq: GMM,eq: mean-est}. 

In this section, I will show that an asymptotically optimal $\W_\infty$ is given by the covariance matrix of the moments: 
\begin{align}\label{eq: sigma}
\Sigma 
    &= 
    \begin{bmatrix}
    \Sigma_{m} & \Sigma_{mg} \\
    \Sigma_{gm} & \Sigma_{g}
    \end{bmatrix} =
    \begin{bmatrix}
    \Eb{m\prns{O; \theta^*_{\min}}m^\top\prns{O; \theta^*_{\min}}} & \Eb{m\prns{O; \theta^*_{\min}}g^\top\prns{O; \theta^*_{\min}, \gamma^*}} \\
    \Eb{g\prns{O; \theta^*_{\min}, \gamma^*}m^\top\prns{O; \theta^*_{\min}}} & \Eb{g\prns{O; \theta^*_{\min}, \gamma^*}g^\top\prns{O; \theta^*_{\min}, \gamma^*}}
    \end{bmatrix} 
\end{align} 
Obviously, in our setting we have  
\begin{align*}
\Sigma_{mg} = \Eb{m\prns{O; \theta^*_{\min}}g^\top\prns{O; \theta^*_{\min}, \gamma^*}} = \mb{0}_{\prns{T_1 + d} \times 1}.
\end{align*}
Therefore, restricting to $\W_{12, N} = \W_{21, N}^\top = \mb{0}_{\prns{T_1 + d} \times 1}$ as we do in \cref{sec: estimation} does not lose any asymptotic efficiency in estimating $\gamma^*$.

We first extend \cref{thm: mu-asymp} to the more general regularized GMM estimators given by 
\cref{eq: 2-stage-1,eq: 2-stage-2}.
\begin{corollary}\label{corollary: mu-asymp2}
Suppose conditions in \cref{thm: mu-asymp} hold, and the positive definite matrix $\W_{N}$ converges in probability to the positive definite matrix  $\W_{\infty}$  as $N \to \infty$. Then the counterfactual mean estimator $\tilde\gamma$ in \cref{eq: 2-stage-2} satisfies that 
\begin{align*}
\sqrt{N}\prns{\tilde\gamma - \gamma^*}    
    &= \frac{1}{\sqrt{N}}\sum_{i = 1}^N\tilde\psi\prns{O_i; \theta^*_{\min}, \gamma^*, \W_{\infty}} + \mathcal{O}_p\prns{\lambda_N\sqrt{N} + \frac{1}{\sqrt{\lambda_N N}}},
\end{align*}
where 
\begin{align*}
&\tilde\psi\prns{O_i; \theta^*_{\min}, \gamma^*, \W_{\infty}}  = -\frac{1}{\Eb{A}}\braces{{A_i\prns{\gamma^* -\tilde W^\top_i\theta_{\min}^*}} + \tilde\Psi(\mathcal{W}_{\infty})\prns{1-A_i}\tilde Z_i\prns{Y_{i, 0}  - \tilde W_i^\top\theta_{\min}^{*}}},\\
&\tilde\Psi(\mathcal{W}_{\infty}) = \prns{\Eb{A \tilde W^\top} + \W_{21, \infty}\W_{11, \infty}^{-1}\Eb{\prns{1-A}\tilde Z\tilde W^\top}} \\
&\qquad\qquad\times \braces{{\Eb{(1-A)\tilde W Z^\top}} \mathcal{W}^{-1}_{11, \infty}\Eb{(1-A)\tilde Z \tilde W^\top}}^{+}{\Eb{(1-A)\tilde W \tilde Z^\top}}\mathcal{W}^{-1}_{11, \infty}- \W_{21, \infty}\W_{11, \infty}^{-1}.
\end{align*}
\end{corollary}
Thus when $\lambda_N\sqrt{N}\to 0$ and ${{\lambda_N N}} \to \infty$, $\tilde\gamma$ is a $\sqrt{N}$-consistent estimator, and the asymptotic variance of $\sqrt{N}\prns{\tilde \gamma-\gamma^*}$ is given by
\begin{align*}
\tilde\sigma^*\prns{\W_{\infty}} = \Eb{\tilde\psi^2\prns{O;\theta^*_{\min}, \gamma^*, \W_{\infty}}}.
\end{align*}
In the following theorem, we show that the optimal choice of $\W_{\infty}$ is the moment covariance matrix $\Sigma$ in \cref{eq: sigma}. Applying this theorem to moment equations in \cref{eq: equations} immediately proves \cref{thm: GMM-optimality}. 
\begin{theorem}\label{thm: opt-general}
For any positive definite matrix $\W_{\infty}$, we have that 
\begin{align*}
\tilde\sigma^2\prns{\W_{\infty}} \ge \tilde\sigma^2\prns{\Sigma}.
\end{align*}
\end{theorem}
Note that the two moment equations in \cref{eq: equations} for the counterfactual mean for the treated $\gamma^*$ has zero covariance. Thus we can set $\W_{12, N} = \W_{21, N}^\top = \mb{0}_{\prns{T_1 + d} \times 1}$ without efficiency loss. However, if we consider estimating the counterfactual mean for the whole population $\tilde\gamma^*$ in \cref{eq: gamma-whole}, then the corresponding two moment equations in \cref{eq: moment-whole} generally have nonzero covariance. Therefore, we have to account for the covariance between the two moment equations to achieve the optimal asymptotic variance among the class of estimators in \cref{eq: 2-stage-1,eq: 2-stage-2}. 
Again, we can construct an asymptotically optimal estimator by a two-stage approach:
\begin{enumerate}
    \item First, construct  estimators $\prns{\hat\theta_{\op{init}}, \hat\gamma_{\op{init}}}$ by solving \cref{eq: gmm-full} with a known weighting matrix, such as the $\prns{T_1 + d + 1} \times \prns{T_1 + d + 1}$ identity matrix. 
    \item Second, construct the covariance matrix estimators
    \begin{align*}
    &\hat \Sigma_{m} = \En{m\prns{O; \hat\theta_{\op{init}}}m^\top\prns{O; \hat\theta_{\op{init}}}}, ~~ \hat \Sigma_{gm} = \En{g\prns{O; \hat\theta_{\op{init}}, \hat\gamma_{\op{init}}}m^\top\prns{O; \hat\theta_{\op{init}}}}.
    \end{align*}
 Then obtain the estimator $\prns{\hat\theta, \hat\gamma}$ by solving \cref{eq: 2-stage-1,eq: 2-stage-2} with $\W_{21, N} = \hat \Sigma_{gm}$ and $\W_{11, N} = \hat \Sigma_{m}$. 
\end{enumerate}

\section{General Regularized GMM}\label{sec: app-reg}
Recall that in \cref{eq: minimum-norm}, we target the minimum-norm bridge function coefficient 
\begin{align}\label{eq: theta-min}
\theta_{\min}^* \coloneqq \argmin  \braces{\|\theta\|_2: \Eb{m\prns{O; \theta}} = 0} = \braces{\|\theta\|_2: \theta\in\Theta^*}.
\end{align}
In the following proposition, we show that we may target any of a family of bridge function coefficients. 
\begin{proposition}\label{lemma: unique-target}
Let $M \in \R{\prns{T_0 + d}\times\prns{T_0 + d}}$ be a positive semidefinite matrix. If the following matrix has full rank $T_0$:
\begin{align}\label{eq: target-condition}
\begin{bmatrix}
I_{T_0 \times T_0} & -\mb{B}_{\pre} 
\end{bmatrix}
M 
\begin{bmatrix}
I_{T_0 \times T_0} \\
 -\mb{B}_{\pre}^\top 
\end{bmatrix},
\end{align}
then there is a unique solution $\theta^*_{M}$ to
\begin{align*}
\min_{\theta \in \Theta^*} ~~ \theta^\top M \theta. 
\end{align*}
\end{proposition}
There are several choices of matrix $M$ that satisfies \cref{eq: target-condition}:
\begin{itemize}
\item By setting $M = I_{\prns{T_0 + d}\times\prns{T_0+d}}$, we recover the minimum-norm bridge coefficient
$\theta_{\min}^*$ in \cref{eq: theta-min}.
\item By setting 
\begin{align*}
M = 
\begin{bmatrix}
I_{T_0\times T_0} & \mb{0}_{T_0 \times d} \\
\mb{0}_{d \times T_0} & \mb{0}_{d \times d}
\end{bmatrix},
\end{align*}
we can obtain 
$\theta^*_{{\op{PAR}}} = \argmin \braces{\|\theta_1\|_2: \theta \in \Theta^*}$, \ie, the partially minimum-norm coefficient vector whose coefficients on $Y_{\pre}$ achieve the smallest $L_2$ norm, regardless of the remaining coefficients on $X$.  
\item For any positive definite matrix $M_{11} \in \R{T_0 \times T_0}$ and  positive semidefinite matrix $M_{22} \in \R{d \times d}$, we may set 
\begin{align*}
M = 
\begin{bmatrix}
M_{11} & \mb{0}_{T_0 \times d} \\
\mb{0}_{d \times T_0} & M_{22}
\end{bmatrix}.
\end{align*}
\end{itemize}
Beyond these, it may be difficult to assess whether other choices are also valid since their 
validity depends on the unknown matrix $\mb{B}_{\pre}$. 

A natural regularized GMM estimator to target $\theta^*_M$ is the following:
\begin{align}\label{eq: GMM-partial}
\hat\theta_{M}
&= \argmin_{\theta = \prns{\theta_1, \theta_2}} \prns{\En{m\prns{O; \theta}}}^\top \mathcal{W}_{m, N} \prns{\En{m\prns{O; \theta}}} +  \lambda_N \theta^\top M\theta, 
\end{align}
whose closed-form solution is 
\begin{align*}
\hat\theta_{M} = \prns{{\En{\prns{1-A}\tilde W \tilde Z^\top}}\W_{m, N}{\En{\prns{1-A}\tilde Z \tilde W^\top}} + \lambda_N M}^{-1}\En{\prns{1-A}\tilde Z Y_{0}}.
\end{align*}
Then we can plug $\hat\theta_{M}$ into \cref{eq: mean-est} to get an estimator $\hat\gamma_{M}$ for the target parameter $\gamma^*$. 

However, deriving the asymptotic distribution of $\hat\gamma_M$ becomes more difficult for a general matrix $M$. 
To extend results in \cref{sec: estimation} (\ie, \cref{thm: mu-asymp,lemma: theta-asymp,thm: CI,thm: GMM-optimality}) to estimator $\hat\theta_M$, we need to adapt statement 1 of \cref{lemma: two-key-facts} to coefficients $\theta_{M}^{*}$, and show that \cref{lemma: support,lemma: pseudo-inverse} hold when we substitute $\lambda_NM$ for all regularization matrices. Once these conditions hold, \cref{thm: mu-asymp} also holds for $\hat\theta_M$ by substituting $\theta_M^*$ for $\theta^*$ in \cref{eq: asymp-linear}. This will require more intricate analysis so we focus on targeting $\theta^*_{\min}$ in \cref{sec: identification}.

In the following proposition, we show that choosing to target some alternative bridge functions may actually lead to better asymptotic variance. 
\begin{proposition}\label{prop: asymp-variance}
For any $\theta^*_M \in \Theta^*$, we have
\begin{align*}
&\Eb{\psi^2\prns{O; \theta^*_M, \gamma^*, \mathcal{W}_{m, \infty}}} \\
=& \begin{bmatrix}
        V_0 ^\top  & b_0^\top
        \end{bmatrix}
        \begin{bmatrix}
        \op{Cov}\prns{U, U \mid A = 1} & \op{Cov}\prns{U, X \mid A = 1} \\
        \op{Cov}\prns{X, U \mid A = 1} & \op{Cov}\prns{X, X \mid A = 1}
        \end{bmatrix}
    \begin{bmatrix}
        V_0  \\ b_0 
        \end{bmatrix} \\
+& \Prb{A = 0}\Psi(\mathcal{W}_{m, \infty})\begin{bmatrix}
  \mb{V}_{\post} & \mb{B}_{\post} \\
  \mb{0} & \mb{I}
  \end{bmatrix}
       \begin{bmatrix}
        \Eb{UU^\top \mid A = 0} & \Eb{UX^\top \mid A = 0} \\
        \Eb{XU^\top \mid A = 0} & \Eb{XX^\top \mid A = 0}
        \end{bmatrix}
        \begin{bmatrix}
  \mb{V}_{\post}^\top &  \mb{0} \\
  \mb{B}_{\post}^\top  & \mb{I}
  \end{bmatrix}\Psi^\top(\mathcal{W}_{m, \infty}) \\
\times & \prns{\Eb{\epsilon_0^2} + \theta_{M, 1}^{*\top}\Eb{\epsilon_{\pre}\epsilon_{\pre}^\top}\theta_{M, 1}^{*}},
\end{align*}
where
\begin{align*}
\Psi(\mathcal{W}_{m, \infty}) = \Eb{A \tilde W^\top}\braces{{\Eb{(1-A)\tilde W Z^\top}} \mathcal{W}_{m, \infty}\Eb{(1-A)\tilde Z \tilde W^\top}}^{+}{\Eb{(1-A)\tilde W \tilde Z^\top}}\mathcal{W}_{m, \infty}.
\end{align*}
\end{proposition}
\cref{prop: asymp-variance} shows that one optimal choice of $M$ is the following:
\begin{align*}
\begin{bmatrix}
\Eb{\epsilon_{\pre}\epsilon_{\pre}^\top} & \mb{0} \\
\mb{0} & \mb{0}
\end{bmatrix}.
\end{align*}
In particular, when $\epsilon_t$ is serially independent and its variance does not vary across time, then the optimal target is $ \theta^*_{{\op{PAR}}} = \argmin \braces{\|\theta_1\|_2: \theta \in \Theta^*}$.

\section{Time-Varing Unmeasured Confounders}\label{sec: app-dynamic}
In this section, we present a more general version of \cref{lemma: time-varying-bridge-simple}, showing the existence of linear bridge functions under weaker conditions.

We first introduce some notations. We define  $\Gamma_{t_2:t_1} = \Gamma_{t_2}\Gamma_{t_2-1}\cdots\Gamma_{t_1}$ for any $t_2 \ge t_1$ and $\Gamma_{t_2:t_1} = I$ for $t_2 < t_1$. We also define $\Sigma_{U_{-T_0}} = \Eb{U_{-T_0}U_{-T_0}^\top \mid A = 0}$, $\Sigma_{U_{-T_0}, X} = \Eb{U_{-T_0}X^\top \mid A = 0}$, $\Sigma_{\eta_t} = \Eb{\eta_t\eta_t^\top \mid A = 0}$, and  $\Sigma_{\eta_t, X} = \Eb{\eta_{t} X^\top \mid A = 0}$. 

In the following lemma, we generalize \cref{lemma: time-varying-bridge-simple} by allowing $\Eb{U_{-T_0} \mid U_0, A = 0, X}$ and $\Eb{\eta_t \mid U_0, A = 0, X}$ to depend on both $U_0$ and $X$.  
\begin{lemma}\label{lemma: time-varying-bridge-general}
Let \cref{assump: time-varying} and Condition 1 in \cref{thm: identification} hold and further assume the following assumptions:
\begin{enumerate}
\item For any $t \in \pre$, $\Eb{\eta_t \mid U_0, X, A = 0}$ and $\Eb{U_{-T_0} \mid U_0, A = 0, X}$ are linear functions of $U_0$ and $X$;
\item The following $\prns{r + d}\times\prns{r + d}$ matrix is invertible:
\begin{align*}
\begin{bmatrix}
\Eb{U_{0}U_0^\top \mid A = 0} & \Eb{U_{0}X^\top \mid A = 0} \\
\Eb{XU_0^\top \mid A = 0} & \Eb{X X^\top \mid A = 0}
\end{bmatrix}.
\end{align*}
We partition its inverse as follows:
\begin{align*}
\begin{bmatrix}
\Eb{U_{0}U_0^\top \mid A = 0} & \Eb{U_{0}X^\top \mid A = 0} \\
\Eb{XU_0^\top \mid A = 0} & \Eb{X X^\top \mid A = 0}
\end{bmatrix}^{-1} = \begin{bmatrix}
G_{11} & G_{12} \\
G_{21} & G_{22}
\end{bmatrix},
\end{align*}
where
\begin{align*}
G_{11} 
  &= \prns{\Eb{U_{0}U_0^\top \mid A = 0} - \Eb{U_{0}X^\top \mid A = 0}\prns{\Eb{X X^\top \mid A = 0}}^{-1}\Eb{XU_0^\top \mid A = 0}}^{-1}  \\
G_{12}
  &= -G_{11}\Eb{U_0 X^\top \mid A = 0}\prns{\Eb{XX^\top \mid A = 0}}^{-1} \\
G_{21}
  &= -\prns{\Eb{XX^\top \mid A = 0}}^{-1}\Eb{XU_0^\top \mid A = 0}G_{11} \\
G_{22}
  &= \prns{\Eb{XX^\top \mid A = 0}}^{-1} \\
  &+ \prns{\Eb{XX^\top \mid A = 0}}^{-1}\Eb{XU_0^\top \mid A = 0}G_{11} \Eb{U_{0}X^\top \mid A = 0}\prns{\Eb{XX^\top \mid A = 0}}^{-1}.
\end{align*}
\item Let $\tilde{\mb{V}}_{\pre}$ be a $T_0 \times d$ matrix whose  $t^{\text{th}}$th row is 
\begin{align*}
& V_{t}^\top \Gamma_{\prns{t-1}:\prns{-T_0}} \prns{\Sigma_{U_{-T_0}}\Gamma^\top_{\prns{-1}:\prns{-T_0}} G_{11} + \Sigma_{U_{-T_0}, X}G_{21}} \\
+& V_{t}^\top \sum_{k=1}^{t + T_0}\Gamma_{\prns{t-1}:\prns{t-k+1}}\prns{\Sigma_{\eta_{t-k}}\Gamma^\top_{\prns{-1}:\prns{t - k + 1}}G_{11} + \Sigma_{\eta_{t -k}, X}G_{21}}.
\end{align*}
Assume that matrix $\tilde{\mb{V}}_{\pre}$ has full column rank equal to $r$.
\end{enumerate}
Then there exist $\theta^* = \prns{\theta_1^*, \theta_2^*} \in \R{T_0 + d}$ such that 
\begin{align*}
    &\Eb{ Y_{0}\prns{0} - \prns{\theta^{*\top}_1 Y_{\pre} + \theta_2^{*\top} X} \mid U_0, A = 0, X} = 0,
\end{align*}
and any such $\theta^*$ satisfies \cref{eq: identification}.
\end{lemma}
In condition 1, we assume $\Eb{\eta_t \mid U_0, X, A = 0}$ and $\Eb{U_{-T_0} \mid U_0, A = 0, X}$ are linear functions to ensure that bridge functions are linear functions. 
Without this condition, there may still exist a bridge  function $h\prns{Y_{\pre}, X}$ that solves the following conditional moment equation:
\begin{align*}
    &\Eb{ Y_{0}\prns{0} - h\prns{Y_{\pre}, X} \mid U_0, A = 0, X} = 0. 
\end{align*}
This conditional moment equation is a Fredholm integral equation of the first kind \citep{Carrasco2007}. The existence of its solutions can be characterized by singular value decomposition of the corresponding compact
operators associated. See the appendix in \cite{miao2018identifying} and appendix B.1 in \cite{kallus2021proxy} for details.
However, without condition 1, even if bridge functions still exist, they may no longer be linear functions so the estimation and inferential theory in this paper no longer applies. 
In this case, the nonparametric estimation estimators developed in \cite{kallus2021proxy,GhassamiAmirEmad2021MKML,pmlr-v139-mastouri21a} may be more suitable. 
However, it is still an open problem how to construct confidence intervals based on these nonparametric estimators when bridge functions are nonunique.

\section{Proofs}
 

\subsection{Proofs in \cref{sec: identification}}
\begin{proof}[Proof of \cref{lemma: bridge}]
First, note that $\epsilon_t \perp A \mid X, U$ implies $Y_t\prns{0} \perp A \mid X, U$ for any $t$. To prove the conclusion, we note that 
\begin{align*}
\Eb{h\prns{Y_{\pre}, X} \mid A= 1} 
    &= \Eb{\Eb{h\prns{Y_{\pre}, X} \mid U, A = 1, X} \mid A= 1} \\
    &= \Eb{\Eb{h\prns{Y_{\pre}, X} \mid U, A = 0, X} \mid A= 1} \\
    &= \Eb{\Eb{Y_0(0) \mid U, A = 0, X} \mid A= 1} \\
    &= \Eb{\Eb{Y_0(0) \mid U, A = 1, X} \mid A= 1} \\
    &= \Eb{Y(0)\mid A= 1} = \gamma^*.
\end{align*}
Here the second equality follows from the fact that $A$ has no causal effects on the pre-treatment outcomes $Y_{\pre} = Y_{\pre}\prns{0}$ and $A \perp Y_{\pre}\prns{0} \mid X, U$. The third equality follows from the definition of bridge function in \cref{def: bridge}. The fourth equality follows from $A \perp Y_{0}\prns{0} \mid X, U$.
\end{proof}

\begin{proof}[Proof for \cref{lemma: bridge-prelim}]
Under \cref{assump: outcome}, we have 
\begin{align*}
 Y_{i, \pre} =  Y_{i, \pre}(0) = \mb{V}_{\pre} U_i + \mb{B}_\pre X_i + \epsilon_{i, \pre}   
    &\implies \mb{V}_{\pre} U_i =  Y_{i, \pre} - \mb{B}_\pre X_i - \epsilon_{i, \pre} \\
    &\implies \theta_1^{*\top} \mb{V}_{\pre} U_i = \theta_1^{*\top}\prns{Y_{i, \pre} - \mb{B}_\pre X_i - \epsilon_{i, \pre}}
\end{align*}
Since $\theta_1^{*\top} \mb{V}_{\pre} = V_{0}^\top$, we have 
$$
V_{0}^\top U_i = \theta_1^{*\top} Y_{i, \pre} -  \theta_1^{*\top}\mb{B}_\pre X_i - \theta_1^{*\top} \epsilon_{i, \pre}.
$$
Therefore, 
\[
Y_{i, 0}\prns{0} =  V_{0}^\top U_i  + \beta_0^\top U_i + \epsilon_{i, 0} =  \theta_1^{*\top} Y_{i, \pre}  + \prns{\beta_0 - \mb{B}_\pre^\top\theta_1^*}^\top U_i + \epsilon_{i, 0}- \theta_1^{*\top} \epsilon_{i, \pre}.
\]
It follows from $\epsilon_{t} \perp \prns{A, U, X}$ for any $t$ that 
\begin{align*}
&\Eb{Y_{i, 0} - \theta_1^{*\top} Y_{i, \pre} - \prns{\beta_0 - \mb{B}_\pre^\top\theta_1^*}^\top X_i \mid A_i = 0, U_i, X_i} \\
=& \Eb{\epsilon_{i, 0}- \theta_1^{*\top} \epsilon_{i, \pre} \mid A_i = 0, U_i, X_i} = 0.
\end{align*}
\end{proof}

\begin{proof}[Proof of \cref{lemma: post-treatment}]
Note that given $A = 0$, $Y_{\post}$ only depends on $X, U, \epsilon_{\post}$, $Y_{\pre}$ only depends on $X, U, \epsilon_{\pre}$, $Y_{0}$ only depends on $X, U, \epsilon_{0}$, and we know that $\epsilon_{\post} \perp \prns{\epsilon_{\pre}^\top, \epsilon_{0}} \mid A = 0$. Thus,
\begin{align*}
Y_{\post} \perp \prns{Y_{\pre}, Y_{0}} \mid X, U, A = 0.
\end{align*}
If follows that 
\begin{align*}
&\Eb{Y_{\post}\prns{Y_{\pre} - \Eb{Y_{\pre} \mid X, U, A = 0}}^\top \mid X, U, A = 0} \\
&\qquad\qquad\qquad\qquad\qquad\qquad\qquad\qquad\qquad = \op{Cov}\prns{Y_{\post}, Y_{\pre} \mid X, U, A = 0} = \mb{0}_{T_1 \times T_0}, \\
&\Eb{Y_{\post}\prns{Y_0 - \Eb{Y_0 \mid X, U, A= 0}} \mid X, U, A = 0} = \op{Cov}\prns{Y_{\post}, Y_0 \mid X, U, A = 0} = \mb{0}_{T_1 \times 1}.
\end{align*}
This means that 
\begin{align*}
&\Eb{Y_{\post}\prns{Y_0 - \theta_1^{*\top}Y_{\pre} - \theta_2^{*\top}X} \mid X, U, A = 0} \\
=& \Eb{Y_{\post}\prns{Y_0 - \theta_1^{*\top}Y_{\pre} - \theta_2^{*\top}X} \mid X, U, A = 0} \\
-& \Eb{Y_{\post}\Eb{Y_0 - \theta_1^{*\top}Y_{\pre} - \theta_2^{*\top}X \mid X, U, A = 0} \mid X, U, A = 0} \\
=& -\theta_1^{*\top} \Eb{Y_{\post}\prns{Y_{\pre} - \Eb{Y_{\pre} \mid X, U, A = 0}}^\top \mid X, U, A = 0} \\
&+ \Eb{Y_{\post}\prns{Y_0 - \Eb{Y_0 \mid X, U, A= 0}} \mid X, U, A = 0} = \mb{0}_{T_1 \times 1}.
\end{align*}
This obviously implies the asserted conclusion in \cref{eq: moment-marginal}.
\end{proof}

\begin{proof}[Proofs for \cref{thm: identification}] According to \cref{assump: outcome}, we have 
\begin{align*}
Y_{\pre}\prns{0} 
    &= \mb{V}_\pre U  + \mb{B}_\pre X + \epsilon_{\pre}, \\
Y_{0}\prns{0} 
    &= V_0^\top U + \beta_0^\top X + \epsilon_{0}.
\end{align*}
Thus for any $\theta = \prns{\theta_1, \theta_2} \in \R{T_0 + d}$, we have 
\begin{align*}
Y_{0} \prns{0}  - \theta_1^\top Y_{\pre} - \theta_2^\top X  
    &= \prns{V_0^\top - \theta_1^\top\mb{V}_\pre}U + \prns{\beta_0 - \mb{B}_\pre^\top \theta_1 - \theta_2}^\top X + \epsilon_0 - \theta_1^\top\epsilon_{\pre} \\
    &= \begin{bmatrix}
    U^\top & X^\top
    \end{bmatrix}
    \begin{bmatrix}
    V_0 - \mb{V}_\pre\theta_1 \\
    \beta_0 - \mb{B}_\pre^\top\theta_1 - \theta_2
    \end{bmatrix}
    + \epsilon_0 - \theta_1^\top\epsilon_{\pre},
\end{align*}
and 
\begin{align*}
\begin{bmatrix}
Y_{\post} \prns{0} \\
X
\end{bmatrix}
    &= 
\begin{bmatrix}
\mb{V}_\post & \mb{B}_\post \\
0_{d\times r} & I_{d\times d}
\end{bmatrix}
\begin{bmatrix}
U \\
X
\end{bmatrix}
+ 
\begin{bmatrix}
\epsilon_{\post} \\
0_{d\times 1}
\end{bmatrix}.
\end{align*}
Under \cref{assump: outcome}, we have  $(\epsilon_{\pre}, \epsilon_{0}, \epsilon_{\post}) \perp (A, U, X)$. Under the serial independence assumption, we have that $\Eb{\prns{1-A}\epsilon_{\post}\epsilon_{\pre}^\top} = \mb{0}_{T_1 \times T_0}$ and $\Eb{\prns{1-A}\epsilon_{\post}\epsilon_{0}} = \mb{0}_{T_1 \times 1}$. Therefore, 
\begin{align*} 
& \Eb{\prns{1-A}\prns{Y_{0}  -  h\prns{Y_{\pre}, X; \theta^{*}}}
        \begin{bmatrix}
        Y_{\post} \\
        X
        \end{bmatrix}
        } \\
=& 
\Prb{A = 0}
\begin{bmatrix}
\mb{V}_\post & \mb{B}_\post \\
0_{d\times r} & I_{d\times d}
\end{bmatrix}
\begin{bmatrix}
\Eb{UU^\top \mid A = 0} & \Eb{UX^\top \mid A = 0} \\
\Eb{XU^\top \mid A = 0} & \Eb{XX^\top \mid A = 0}
\end{bmatrix}
\begin{bmatrix}
    V_0 - \mb{V}_\pre\theta^*_1 \\
    \beta_0 - \mb{B}_\pre^\top\theta^*_1 - \theta^*_2
\end{bmatrix},
\end{align*}
Given the asserted conditions, the following matrix has full column rank:
\begin{align*}
\Prb{A = 0}
\begin{bmatrix}
\mb{V}_\post & \mb{B}_\post \\
0_{d\times r} & I_{d\times d}
\end{bmatrix}
\begin{bmatrix}
\Eb{UU^\top \mid A = 0} & \Eb{UX^\top \mid A = 0} \\
\Eb{XU^\top \mid A = 0} & \Eb{XX^\top \mid A = 0}
\end{bmatrix}.
\end{align*}
Thus $\theta^*$ solves \cref{eq: moment-marginal} if and only if 
\begin{align*}
 0 = \begin{bmatrix}
    V_0 - \mb{V}_\pre\theta^*_1 \\
    \beta_0 - \mb{B}_\pre^\top\theta^*_1 - \theta^*_2
\end{bmatrix},
\end{align*}
namely, $\theta^* \in \Theta^*$.

This also implies that for any $\theta^*$ that solves \cref{eq: moment-marginal}, \cref{lemma: bridge} gives that 
\begin{align*}
\gamma^* = \Eb{h\prns{Y_{\pre}, X; \theta^*} \mid A = 1}.
\end{align*}
Therefore, $\gamma^*$ is identifiable. 
\end{proof}

\subsection{Proofs in \cref{sec: estimation}}\label{sec: proof-est}

\subsubsection{Notations and Supporting Lemmas}
We first denote 
\begin{align*}
&\hat K = \En{(1-A)\tilde Z\tilde W^\top}, ~~ K = \Eb{(1-A)\tilde Z \tilde W^\top}, \\
&\hat b = \En{(1-A)\tilde ZY_{0}}, ~~ b = \Eb{(1-A)\tilde ZY_{0}}.
\end{align*}
According to \cref{thm: identification}, the set of bridge function coefficients is 
\begin{align*}
\Theta^* = \braces{\theta \in \R{T_0 + d}: K\theta = b},
\end{align*}
and the minimal one is $\theta^*_{\min} = \argmin_{\theta\in\Theta^*} \braces{\|\theta\|_2: K\theta = b} = K^+ b$, 
where $K^+$ is the Moore–Penrose inverse of matrix $K$. Moreover, the regularized GMM estimator $\hat\theta$ in \cref{eq: GMM} can be written as 
\begin{align*}
\hat\theta = \prns{\hat K^\top \mathcal{W}_{m, N} \hat K + \lambda_N I}^{-1}\hat K^\top \mathcal{W}_{m, N}  \hat b.
\end{align*}

We first give two simple lemmas that can be easily proved so their proofs are omitted.  
 \begin{lemma}\label{lemma: inverse-diff}
 For invertible matrices $A $ and $B$, 
\begin{align}
A^{-1} - B^{-1} 
    &= A^{-1}\prns{B - A}B^{-1} \label{eq: matrix-err1} \\
    &= B^{-1}\prns{A - B}A^{-1}\label{eq: matrix-err2}
\end{align}
 \end{lemma}
 
 \begin{lemma}\label{lemma: svd-bound}
 For two matrices $A, B \in \R{m \times n}$, denote their singular values in decreasing order as $\sigma_j\prns{A}, \sigma_j\prns{B}$ for $j = 1, \dots, \min\prns{m, n}$. Then 
 \begin{align*}
 \sigma_{\min\prns{m, n}}\prns{A + B} \ge \sigma_{\min\prns{m, n}}\prns{A} - \sigma_1\prns{B}.
 \end{align*}
 \end{lemma}

Next, we prove several lemmas that play an important role in bounding the errors of our regularized GMM estimator. 

\begin{lemma}\label{lemma: two-key-facts}
Suppose that conditions in \cref{thm: identification} hold.  Then 
\begin{enumerate}
\item there exists $\phi_{\theta}\in \R{T_0 + d}$ such that the minimum norm solution  $\theta_{\min}^* = K^\top  K \phi_{\theta}$.
\item there exists   $\phi_{W}\in \R{T_1 + d}$ such that $\Eb{A\tilde W} = K^\top K \phi_{W}$. 
\end{enumerate}
\end{lemma} 

\begin{proof}[Proof of \cref{lemma: two-key-facts}]
For the matrix $K \in \R{\prns{T_1 + d} \times \prns{T_0 + d}}$, we denote its column space and kernel space as $\mathcal R\prns{K} \subseteq \R{T_1 + d}$ and  $\mathcal N\prns{K} \subseteq \R{T_0 + d}$ respectively.
We similarly define $\mathcal R\prns{K^\top} \subseteq \R{T_0 + d}$ and $\mathcal N\prns{K^\top} \subseteq \R{T_1 + d}$. For spaces $\mathcal N\prns{K^\top}$ and $\mathcal N\prns{K^\top}^\perp$, we define the projection of an element $z \in \R{T_1 + d}$ into these two spaces as $P_{\mathcal N(K^\top)}z$ and $P_{\mathcal N(K^\top)^\perp}z$ respectively.

 By the relationship of four fundamental spaces, we know that $\mathcal N(K)^\perp = \overline{\mathcal R(K^\top)} = \mathcal R(K^\top)$. Moreover, we can prove that $\mathcal R\prns{K^\top} =  \mathcal R\prns{K^\top  K}$ so that $\mathcal N(K)^\perp = \mathcal R\prns{K^\top  K}$. Indeed, for any $y \in \mathcal R(K^\top)$, there exists $z \in \R{T_1 + d}$ such that 
\[
y = K^\top   z = K^\top  \bracks{P_{\mathcal N(K^\top)}{ z} + P_{\mathcal N(K^\top)^\perp}{ z}} = K^\top P_{\mathcal R(K)}\prns{ z}. 
\]
Since $P_{\mathcal R(K)}{ z} \in \mathcal R(K)$, there must exist $x \in \R{T_d + d}$ such that $
y = K^\top  P_{\mathcal R(K)}{ z} = K^\top  K x$. It follows that  $\mathcal R(K^\top)\subseteq \mathcal R(K^\top  K)$. In addition, we obviously have $\mathcal R(K^\top  K)\subseteq \mathcal R(K^\top)$. Therefore, $\mathcal R(K^\top  K) =  \mathcal R(K^\top)$.

Any solution $\theta \in \R{T_0 + d}$ of $K\theta = b$ can be written as $\theta = \theta_0 + \theta_1$ where $\theta_0 \in \mathcal N(K)^\perp$ satisfies $K\theta_0 = b$ and $\theta_1 \in \mathcal N(K)$. The minimum norm solution $\theta^*_{\min} \in \mathcal N(K)^\perp \subseteq \mathcal R\prns{K^\top  K}$. Therefore, there exist $\phi_\theta \in \R{T_0 + d}$ such that  $\theta^*_{\min} = K^\top  K \phi_\theta$. This finishes the proof of statement 1. 

Moreover, according to the proof of \cref{lemma: nonunique}, we have
\begin{align*}
K^\top 
    &= \Eb{\prns{1-A_i}\tilde W\tilde Z^\top} \\
    &= \Prb{A = 0}
\begin{bmatrix}
\mb{V}_\pre &  \mb{B}_\pre \\
\mb{0}_{r \times d} & I_{d\times d}
\end{bmatrix}
\begin{bmatrix}
\Eb{UU^\top \mid A = 0} & \Eb{UX^\top \mid A = 0} \\
\Eb{XU^\top \mid A = 0} & \Eb{XX^\top \mid A = 0}
\end{bmatrix}
\begin{bmatrix}
\mb{V}_\post^\top &  \mb{0}_{d\times r} \\
\mb{B}_\post^\top 
 & I_{d\times d}
\end{bmatrix}.
\end{align*}
In addition, 
\begin{align*}
\Eb{A\tilde W} = 
\begin{bmatrix}
\Eb{AY_{pre}^\top} \\
\Eb{AX^\top}
\end{bmatrix}
 = 
\begin{bmatrix}
\mb{V}_\pre &  \mb{B}_\pre \\
\mb{0}_{r \times d} & I_{d\times d}
\end{bmatrix}
\begin{bmatrix}
\Eb{AU} \\
\Eb{AX}
\end{bmatrix}.
\end{align*}
Under conditions in \cref{thm: identification}, we know that the following matrix has full row rank: 
\begin{align*}
\begin{bmatrix}
\Eb{UU^\top \mid A = 0} & \Eb{UX^\top \mid A = 0} \\
\Eb{XU^\top \mid A = 0} & \Eb{XX^\top \mid A = 0}
\end{bmatrix}
\begin{bmatrix}
\mb{V}_\post^\top &  \mb{0}_{d\times r} \\
\mb{B}_\post^\top 
 & I_{d\times d}
\end{bmatrix}.
\end{align*}
Therefore, 
\begin{align*}
\Eb{A\tilde W} \subseteq \mathcal{R}\prns{\begin{bmatrix}
\mb{V}_\pre &  \mb{B}_\pre \\
\mb{0}_{r \times d} & I_{d\times d}
\end{bmatrix}} = \mathcal{R}\prns{K^\top}
\end{align*}
It follows that there exists $\phi_{W}\in \R{T_1 + d}$ such that $\Eb{A\tilde W } = K^\top K \phi_{W}$, which finishes the proof of statement 2. 
\end{proof}

\begin{lemma}\label{lemma: support}
Suppose that conditions in \cref{thm: identification} hold. Then for a weighting matrix $\mathcal{W}_{m, N}$ that converges to a fixed invertible matrix $\mathcal{W}_{m, \infty}$ in probability as $N \to \infty$, we have 
\begin{align*}
&\|\prns{ K^\top \mathcal{W}_{m, N} K + \lambda_N I}^{-1} K^\top K\| = \mathcal{O}_p\prns{1}, ~~ \|\prns{ K^\top \mathcal{W}_{m, N} K + \lambda_N I}^{-1} K^\top \| = \mathcal{O}_p\prns{1}, \\
&\left\|\hat K \prns{\hat K^\top\mathcal{W}_{m, N} \hat K + \lambda_N I}^{-1}\hat K^\top\right\| = \mathcal{O}_p(1), ~~ \left\|\prns{\hat K^\top\mathcal{W}_{m, N} \hat K + \lambda_N I}^{-1}\hat K^\top\right\| = \mathcal{O}_p\prns{\frac{1}{\sqrt{\lambda_N}}}.
\end{align*}  
\end{lemma} 
\begin{proof}[Proof of \cref{lemma: support}]
We prove the conclusion when $T_0 \ge T_1$. The conclusion for $T_0 < T_1$ can be proved analogously. For simplicity, we use $\mb{0}$ to represent generic all-zero matrices of different sizes comformable to different contexts. For any matrix $A$, we denote its smallest and largest singular value as $\sigma_{\min}\prns{A}$ and $\sigma_{\max}\prns{A}$ respectively.

Let the singular value decomposition of $K \in \R{\prns{T_1 + d} \times \prns{T_0 + d}}$ be $K = L\Sigma R^\top$ with $L \in \R{\prns{T_1 + d} \times \prns{T_1 + d}}$, $R \in \R{\prns{T_0 + d} \times \prns{T_0 + d}}$, and $\Sigma \in \R{\prns{T_1 + d} \times \prns{T_0 + d}}$. According to \cref{lemma: nonunique}, the matrix $K$ has rank at most $r + d$. Thus we can partition $L$ and $\Sigma$ as follows: 
\begin{align*}
L^\top = 
\begin{bmatrix}
L^\top_{1:(r+d)} \\
L^\top_{(r+d+1):(T_1 + d)}
\end{bmatrix},
~~
\Sigma 
= 
\begin{bmatrix}
\Sigma_{1:\prns{r+d}} & \mb{0}_{\prns{r+d} \times \prns{T_0 - r}}\\
\mb{0}_{\prns{T_1 - r} \times \prns{r+d}} & \mb{0}_{\prns{T_1 - r} \times \prns{T_0 - r}}
\end{bmatrix},
\end{align*}
where $\Sigma_{1:\prns{r+d}}$ is a $\prns{r+d}\times\prns{r+d}$ diagonal matrix whose diagonal elements are singular values $\sigma_{j}$ for $j = 1, \dots, r + d$. 

It follows that 
\begin{align*}
&\prns{ K^\top \mathcal{W}_{m, N} K + \lambda_N I}^{-1} K^\top K \\
    =& \prns{R\Sigma^\top L^\top \mathcal{W}_{m, N} L\Sigma R^\top + \lambda_N I }^{-1}R\Sigma^\top\Sigma R^\top = R \prns{\Sigma^\top L^\top \mathcal{W}_{m, N} L\Sigma  + \lambda_N I }^{-1} \Sigma^\top\Sigma R^\top \\
    =& R
\begin{bmatrix}
\prns{\Sigma^\top_{1:\prns{r+d}}L_{1:\prns{r+d}}^\top \mathcal{W}_{m, N}L_{1:\prns{r+d}}\Sigma_{1:\prns{r+d}} + \lambda_N I_{r+d}}^{-1} & \mb{0} \\
\mb{0} & \frac{1}{\lambda_N }I_{\prns{T_0 - r}}
\end{bmatrix}
\begin{bmatrix}
\Sigma_{1:\prns{r+d}}^2 & \mb{0}\\
\mb{0} & \mb{0}
\end{bmatrix}
R^\top \\
    =& R
\begin{bmatrix}
\prns{\Sigma^\top_{1:\prns{r+d}}L_{1:\prns{r+d}}^\top \mathcal{W}_{m, N} L_{1:\prns{r+d}}\Sigma_{1:\prns{r+d}} + \lambda_N I_{r+d}}^{-1}\Sigma_{1:\prns{r+d}}^2 & \mb{0} \\
\mb{0} & \mb{0}
\end{bmatrix}
R^\top.
\end{align*}
Under the asserted assumptions, we have that $L_{1:\prns{r+d}}^\top  \mathcal{W}_{m, N} L_{1:\prns{r+d}}$ has full rank and $\|L_{1:\prns{r+d}}^\top  \mathcal{W}_{m, N} L_{1:\prns{r+d}}\| = \mathcal{O}_p\prns{1}$. It is easy to show that 
\begin{align*}
&\left\|\prns{ K^\top  \mathcal{W}_{m, N} K + \lambda_N I}^{-1} K^\top K\right\| \\
&\qquad = \left\|\prns{\Sigma^\top_{1:\prns{r+d}}L_{1:\prns{r+d}}^\top  \mathcal{W}_{m, N} L_{1:\prns{r+d}}\Sigma_{1:\prns{r+d}} + \lambda_N I_{r+d}}^{-1}\Sigma_{1:\prns{r+d}}^2\right\| = \mathcal{O}_p\prns{1}.
\end{align*}
Similarly, we can show that
\begin{align*}
&\left\|\prns{ K^\top \mathcal{W}_{m, N}  K + \lambda_N I}^{-1} K^\top \right\| \\
&\qquad = \left\|\prns{\Sigma^\top_{1:\prns{r+d}}L_{1:\prns{r+d}}^\top   \mathcal{W}_{m, N} L_{1:\prns{r+d}}\Sigma_{1:\prns{r+d}} + \lambda_N I_{r+d}}^{-1}\Sigma_{1:\prns{r+d}} \right\| = \mathcal{O}_p\prns{1}. 
\end{align*} 
Next, denote the singular value decomposition of $\hat K$ and $K$ as $\hat K = \hat L\hat \Sigma\hat R^\top$ with $\hat L \in \R{\prns{T_1 + d} \times \prns{T_1 + d}}$, $\hat R \in \R{\prns{T_0 + d} \times \prns{T_0 + d}}$, and $\hat \Sigma \in \R{\prns{T_1 + d} \times \prns{T_0 + d}}$. Note that we can write $\hat\Sigma = [\hat\Sigma_0, \mb{0}]$ where $\hat\Sigma_0 \in \R{\prns{T_1+d}\times\prns{T_1+d}}$ is a diagonal matrix whose diagonal elements $\hat\sigma_j$ for $j = 1, \dots, T_1 + d$ are singulr values of $\hat K$ ordered in decreasing order. The matrix $\hat K$ typically has full rank in finite sample even though its limit $K$ is rank-deficient, so $\hat \sigma_{j}$ for $j = r+d+1, \dots, T_1 + d$ are typically strictly positive for finite $N$ but their limits are $0$ as $N \to \infty$. So we assume that the rank of $\hat K$  is $T_1 + d$ for any finite $N$ for simplicity. When this is not true, we can adapt the proof below by replacing $\hat\sigma_{T_1 + d}$ by the smallest nonzero singular value of $\hat K$.

It follows that 
\begin{align*}
\|\hat K \prns{\hat K^\top \mathcal{W}_{m, N} \hat K + \lambda_N I}^{-1}\hat K^\top\| 
    &= 
    \|\hat L \hat\Sigma\prns{\hat\Sigma^\top \hat L^\top \mathcal{W}_{m, N}\hat L\hat \Sigma  + \lambda_N I }^{-1} \hat\Sigma^\top \hat L^\top\| \\
    &= \|\hat\Sigma_0\prns{\hat\Sigma^\top_0 \hat L^\top \mathcal{W}_{m, N}\hat L\hat \Sigma_0  + \lambda_N I }^{-1} \hat\Sigma^\top_0 \| \\
    &= \|\prns{\hat L^\top \mathcal{W}_{m, N}\hat L + \lambda_N \hat\Sigma_0^{-2} }^{-1} \| \\
    &\le \frac{1}{\sigma_{\min}\prns{\hat L^\top \mathcal{W}_{m, N}\hat L} - \lambda_N\sigma_{\max}\prns{\hat\Sigma^{-2}_0}}\tag{\cref{lemma: svd-bound}} \\
    &= \frac{1}{\sigma_{\min}\prns{\mathcal{W}_{m, N}} - \frac{\lambda_N}{\hat\sigma_{T_1 +d }^2}} = \mathcal{O}_p\prns{1},
\end{align*}
where the last equality follows from the fact that $\sigma_{\min}\prns{\mathcal{W}_{m, N}} = \mathcal{O}_p(1)$ since $\mathcal{W}_{m, N}$ converges to a fixed invertible matrix. 

Similarly, we have 
\begin{align*}
\|\prns{\hat K^\top \mathcal{W}_{m, N} \hat K + \lambda_N I}^{-1}\hat K^\top\| 
    &= \|\hat\Sigma^{-1}_0\prns{ \hat L^\top \mathcal{W}_{m, N}\hat L + \lambda_N \hat\Sigma_0^{-2} }^{-1}  \| \\
    &\le \frac{\frac{1}{\hat\sigma_{T_1 + d}}}{\sigma_{\min}\prns{\mathcal{W}_{m, N}} - \frac{\lambda_N}{\hat\sigma^2_{T_1 + d}}} 
        =  \frac{1}{2\sqrt{\lambda_N}}\frac{2\sqrt{\lambda_N}{\hat\sigma_{T_1 + d}}}{\hat\sigma^2_{T_1 + d}\sigma_{\min}\prns{\mathcal{W}_{m, N}} -{\lambda_N}}\\
        &\le \frac{1}{2\sqrt{\lambda_N}}\frac{\lambda_N + {\hat\sigma_{T_1 + d}}^2}{\hat\sigma^2_{T_1 + d}\sigma_{\min}\prns{\mathcal{W}_{m, N}} -{\lambda_N}} = \mathcal{O}_p\prns{\frac{1}{\sqrt{\lambda}_N}}.
\end{align*} 
\end{proof}

\begin{lemma}\label{lemma: two-decompose}
We have the following two equivalent decompositions:
\begin{align}\label{eq: two-decompose-1}
& \prns{\hat K^\top  \mathcal{W}_{m, N} \hat K + \lambda_N I}^{-1} - \prns{ K^\top  \mathcal{W}_{m, N} K + \lambda_N I}^{-1} \notag \\
=&  -\prns{\hat K^\top \mathcal{W}_{m, N} \hat K + \lambda_N I}^{-1}
    \hat K^\top \mathcal{W}_{m, N}\prns{\hat K - K}
    \prns{ K^\top \mathcal{W}_{m, N} K + \lambda_N I}^{-1}  \nonumber \\
-& \prns{\hat K^\top \mathcal{W}_{m, N} \hat K + \lambda_N I}^{-1}\prns{\hat K - K}^\top \mathcal{W}_{m, N} K\prns{ K^\top \mathcal{W}_{m, N} K + \lambda_N I}^{-1},
\end{align}
and 
\begin{align}\label{eq: two-decompose-2}
&\prns{\hat K^\top\mathcal{W}_{m, N}\hat K + \lambda_N I}^{-1} - \prns{ K^\top \mathcal{W}_{m, N} K + \lambda_N I}^{-1} \notag   \\
=& \prns{ K^\top \mathcal{W}_{m, N} K + \lambda_N I}^{-1}\prns{\hat K - K}^\top \mathcal{W}_{m, N} \hat K \prns{\hat K^\top \mathcal{W}_{m, N} \hat K + \lambda_N I}^{-1} \nonumber \\
+& \prns{ K^\top \mathcal{W}_{m, N} K + \lambda_N I}^{-1}K^\top \mathcal{W}_{m, N} \prns{\hat K - K}\prns{\hat K^\top \mathcal{W}_{m, N} \hat K + \lambda_N I}^{-1}.
\end{align}
\end{lemma}
\begin{proof}[Proof of \cref{lemma: two-decompose}]
First, according to \cref{eq: matrix-err1} in \cref{lemma: inverse-diff}
\begin{align*}
&\prns{\hat K^\top  \mathcal{W}_{m, N} \hat K + \lambda_N I}^{-1} - \prns{ K^\top  \mathcal{W}_{m, N} K + \lambda_N I}^{-1} \\
=& -\prns{\hat K^\top \mathcal{W}_{m, N} \hat K + \lambda_N I}^{-1}\bracks{\hat K^\top \mathcal{W}_{m, N} \hat K - K^\top \mathcal{W}_{m, N} K}\prns{ K^\top \mathcal{W}_{m, N} K + \lambda_N I}^{-1} \\
=& -\prns{\hat K^\top \mathcal{W}_{m, N} \hat K + \lambda_N I}^{-1}
    \hat K^\top \mathcal{W}_{m, N}\prns{\hat K - K}
    \prns{ K^\top \mathcal{W}_{m, N} K + \lambda_N I}^{-1} \\
 -& \prns{\hat K^\top \mathcal{W}_{m, N} \hat K + \lambda_N I}^{-1}\prns{\hat K - K}^\top \mathcal{W}_{m, N} K\prns{ K^\top \mathcal{W}_{m, N} K + \lambda_N I}^{-1}.
\end{align*}
Second, according to \cref{eq: matrix-err2}  in \cref{lemma: inverse-diff}, 
\begin{align*}
&\prns{\hat K^\top\mathcal{W}_{m, N}\hat K + \lambda_N I}^{-1} - \prns{ K^\top \mathcal{W}_{m, N} K + \lambda_N I}^{-1}  \\
    =& \prns{ K^\top \mathcal{W}_{m, N} K + \lambda_N I}^{-1} \prns{\hat K^\top \mathcal{W}_{m, N} \hat K - K^\top \mathcal{W}_{m, N} K}\prns{\hat K^\top \mathcal{W}_{m, N}\hat K + \lambda_N I}^{-1} \\
    =& \prns{ K^\top \mathcal{W}_{m, N} K + \lambda_N I}^{-1}\prns{\hat K - K}^\top \mathcal{W}_{m, N} \hat K \prns{\hat K^\top \mathcal{W}_{m, N} \hat K + \lambda_N I}^{-1} \\
    +& \prns{ K^\top \mathcal{W}_{m, N} K + \lambda_N I}^{-1}K^\top \mathcal{W}_{m, N} \prns{\hat K - K}\prns{\hat K^\top \mathcal{W}_{m, N} \hat K + \lambda_N I}^{-1}
\end{align*}
\end{proof}

\begin{lemma}\label{lemma: pseudo-inverse}
If $\mathcal{W}_{m, \infty}$ has full rank, then
\begin{align*}
\norm{\braces{\prns{ K^\top  \mathcal{W}_{m, \infty} K + \lambda_N I}^{-1} - \prns{ K^\top  \mathcal{W}_{m, \infty} K}^{+}}K^\top} = \mathcal{O}\prns{\lambda_N}.
\end{align*}
\end{lemma}
\begin{proof}[Proof of \cref{lemma: pseudo-inverse}]
According to \cref{lemma: nonunique}, matrix $K$ has rank at most $r + d$, and we assume $\op{Rank}\prns{K} = r+ d$ for simplicity. We also assume that $T_0 \ge T_1$. Proving this lemma for matrix $K$ of rank smaller than $r + d$ and $T_0 < T_1$ is analogous. 

Consider the singular value decomposition  $K = L\Sigma R^\top$ with $L \in \R{\prns{T_1 + d} \times \prns{T_1 + d}}$ in the proof of \cref{lemma: support} with the partition
\begin{align*}
L^\top = 
\begin{bmatrix}
L^\top_{1:(r+d)} \\
L^\top_{(r+d+1):(T_1 + d)}
\end{bmatrix},
~~
\Sigma 
= 
\begin{bmatrix}
\Sigma_{1:\prns{r+d}} & \mb{0}_{\prns{r+d} \times \prns{T_0 - r}}\\
\mb{0}_{\prns{T_1 - r} \times \prns{r+d}} & \mb{0}_{\prns{T_1 - r} \times \prns{T_0 - r}}
\end{bmatrix}.
\end{align*}
Then according to the assumptions, we have 
\begin{align*}
\norm{\braces{\prns{ K^\top  \mathcal{W}_{m, \infty} K + \lambda_N I}^{-1} - \prns{ K^\top  \mathcal{W}_{m, \infty} K}^{+}}K^\top} = \norm{R^\top 
\begin{bmatrix}
\prns{*} & \mb{0} \\
\mb{0} & \mb{0}
\end{bmatrix}} = \norm{\prns{*}},
\end{align*}
where 
\begin{align*}
\norm{\prns{*}} = 
    &\bigg\|\prns{\Sigma_{1:\prns{r+d}}L^\top_{1:(r+d)}\mathcal{W}_{m, \infty}L_{1:(r+d)}\Sigma^\top_{1:\prns{r+d}} + \lambda_N I_{r+d}}^{-1}\Sigma_{1:\prns{r+d}}L^\top_{1:\prns{r+d}} \\
    &\qquad\qquad\qquad\qquad - \prns{\Sigma_{1:\prns{r+d}}L^\top_{1:(r+d)}\mathcal{W}_{m, \infty}L_{1:(r+d)}\Sigma^\top_{1:\prns{r+d}}}^{-1}\Sigma_{1:\prns{r+d}}L^\top_{1:\prns{r+d}}\bigg\| \\
    &= \lambda_N \norm{\prns{\Sigma_{1:\prns{r+d}}L^\top_{1:(r+d)}\mathcal{W}_{m, \infty}L_{1:(r+d)}\Sigma^\top_{1:\prns{r+d}} + \lambda_N I_{r+d}}^{-1}}\tag{\cref{lemma: inverse-diff}} \\
    &\qquad\qquad\qquad\qquad \times \norm{\prns{\Sigma_{1:\prns{r+d}}L^\top_{1:(r+d)}\mathcal{W}_{m, \infty}L_{1:(r+d)}\Sigma^\top_{1:\prns{r+d}}}^{-1}\Sigma_{1:\prns{r+d}}L^\top_{1:\prns{r+d}}} \\
    &= \mathcal{O}\prns{\lambda_N}.
\end{align*}
\end{proof}

\subsubsection{Proof of \cref{lemma: nonunique}}
\begin{proof}[Proof for \cref{lemma: nonunique}]
Note that when $T_0 > r$, there are infinitely many solutions $\theta_1^*$ to the equation $\mb{V}_\pre^\top \theta_1^* = V_0$. Thus $\Theta^*$ contains infinitely many elements. For any $\theta^* \in \Theta^*$, 
\begin{align*}
\nabla \Eb{m\prns{O; \theta^*}} 
    &= -
        \begin{bmatrix}
                \Eb{\prns{1-A}Y_{\post}Y_{\pre}^\top} & \Eb{\prns{1-A}Y_{\post}X} \\
                \Eb{\prns{1-A}XY_{\pre}^\top}  & \Eb{\prns{1-A}XX^\top} 
        \end{bmatrix}\\
    &= 
\Prb{A = 0}
\begin{bmatrix}
\mb{V}_\post & \mb{B}_\post \\
0_{d\times r} & I_{d\times d}
\end{bmatrix}
\begin{bmatrix}
\Eb{UU^\top \mid A = 0} & \Eb{UX^\top \mid A = 0} \\
\Eb{XU^\top \mid A = 0} & \Eb{XX^\top \mid A = 0}
\end{bmatrix}
\begin{bmatrix}
\mb{V}_\pre^\top &  0_{r \times d} \\
\mb{B}_\pre^\top  & I_{d\times d}.
\end{bmatrix}
\end{align*}
Under the asserted conditions, we have that the following matrix has rank $r+ d$:
\begin{align*}
\begin{bmatrix}
\mb{V}_\post & \mb{B}_\post \\
0_{d\times r} & I_{d\times d}
\end{bmatrix}
\begin{bmatrix}
\Eb{UU^\top \mid A = 0} & \Eb{UX^\top \mid A = 0} \\
\Eb{XU^\top \mid A = 0} & \Eb{XX^\top \mid A = 0}
\end{bmatrix}.
\end{align*}
Therefore, $\nabla \Eb{m\prns{O; \theta^*}}$ has rank at most $r+d$. 
\end{proof}

\subsubsection{Proof of \cref{lemma: theta-asymp}}
\begin{proof}[Proof of \cref{lemma: theta-asymp}]
\textbf{Step I: Decomposing Estimation Errors. }
\begin{align*}
\hat\theta-\theta_{\min}^* 
    &= \prns{\hat K^\top \mathcal{W}_{m,N} \hat K + \lambda_N I}^{-1}\hat K^\top \mathcal{W}_{m,N} \hat b - \theta_{\min}^* \\
    &= \prns{\hat K^\top \mathcal{W}_{m,N} \hat K + \lambda_N I}^{-1}\hat K^\top \mathcal{W}_{m,N}  \prns{\hat b - \hat K\theta_{\min}^*} + \prns{\hat K^\top \mathcal{W}_{m,N} \hat K + \lambda_N I}^{-1}\hat K^\top \mathcal{W}_{m,N}  \hat K\theta_{\min}^* \\
    &- \prns{ K^\top \mathcal{W}_{m, N} K + \lambda_N I}^{-1} K^\top \mathcal{W}_{m, N} K\theta_{\min}^* + \prns{ K^\top  \mathcal{W}_{m, N} K + \lambda_N I}^{-1} K^\top  \mathcal{W}_{m, N} K\theta_{\min}^* -\theta_{\min}^*.
\end{align*}
We can convert the second and third terms as follows: 
 \begin{align}
 &\prns{\hat K^\top\mathcal{W}_{m,N}\hat K + \lambda_N I}^{-1}\hat K^\top \mathcal{W}_{m,N}\hat K\theta_{\min}^* - \prns{ K^\top \mathcal{W}_{m, N} K + \lambda_N I}^{-1} K^\top \mathcal{W}_{m, N} K\theta_{\min}^* \notag \\
    =& \braces{I - \lambda_N \prns{\hat K^\top \mathcal{W}_{m,N}\hat K + \lambda_N I}^{-1}}\theta_{\min}^*  - \braces{I - \lambda_N \prns{ K^\top \mathcal{W}_{m, N} K + \lambda_N I}^{-1}}\theta_{\min}^* \notag\\
    =& -\lambda_N \bracks{\prns{\hat K^\top  \mathcal{W}_{m,N} \hat K + \lambda_N I}^{-1} - \prns{ K^\top  \mathcal{W}_{m, N} K + \lambda_N I}^{-1}}\theta_{\min}^* \notag \\
    =& -\lambda_N\prns{\hat K^\top \mathcal{W}_{m,N} \hat K + \lambda_N I}^{-1}
    \hat K^\top \mathcal{W}_{m, N} \prns{\hat K - K}
    \prns{ K^\top \mathcal{W}_{m, N} K + \lambda_N I}^{-1}\theta_{\min}^* \label{eq: thm1-eq1-a}\\
    -& \lambda_N\prns{\hat K^\top \mathcal{W}_{m,N} \hat K + \lambda_N I}^{-1}\prns{\hat K - K}^\top \mathcal{W}_{m, N} K\prns{ K^\top \mathcal{W}_{m, N} K + \lambda_N I}^{-1}\theta_{\min}^*, \label{eq: thm1-eq1-b}
 \end{align}
 where \cref{eq: thm1-eq1-a,eq: thm1-eq1-b} from \cref{eq: two-decompose-1} in \cref{lemma: two-decompose}. 

 We denote $\zeta_N = \prns{ K^\top \mathcal{W}_{m, N} K + \lambda_N I}^{-1}\theta_{\min}^*$. According to statement 1 in \cref{lemma: two-key-facts}, we have 
 $$\zeta_N = \prns{ K^\top  \mathcal{W}_{m, N} K + \lambda_N I}^{-1}K^\top K\phi_\theta.
 $$
 It follows that 
\begin{align*}
\cref{eq: thm1-eq1-a} 
    = & -\lambda_N  \prns{\hat K^\top \mathcal{W}_{m, N} \hat K + \lambda_N I}^{-1}
    \hat K^\top \mathcal{W}_{m, N} \prns{\hat K - K}\zeta_N  \\
    = &-\lambda_N\bracks{\prns{\hat K^\top \mathcal{W}_{m, N}\hat K + \lambda_N I}^{-1} - \prns{ K^\top\mathcal{W}_{m, N}  K + \lambda_N I}^{-1}}\hat K^\top \mathcal{W}_{m, N} \prns{\hat K - K}\zeta_N \\
    & - \lambda_N\prns{ K^\top \mathcal{W}_{m, N} K + \lambda_N I}^{-1} K^\top \mathcal{W}_{m, N}\prns{\hat K - K}\zeta_N  \\
    &-  \lambda_N\prns{ K^\top \mathcal{W}_{m, N} K + \lambda_N I}^{-1}\prns{\hat K - K}^\top \mathcal{W}_{m, N} \prns{\hat K - K}\zeta_N,
\end{align*}
and 
\begin{align*}
\cref{eq: thm1-eq1-b} 
    = & -\lambda_N\prns{\hat K^\top \mathcal{W}_{m, N} \hat K + \lambda_N I}^{-1}\prns{\hat K - K}^\top \mathcal{W}_{m, N} K \zeta_N \\
    = & -\lambda_N \bracks{\prns{\hat K^\top \mathcal{W}_{m, N}\hat K + \lambda_N I}^{-1} - \prns{ K^\top \mathcal{W}_{m, N} K + \lambda_N I}^{-1}}\prns{\hat K - K}^\top \mathcal{W}_{m, N} K\zeta_N \\
      &- \lambda_N \prns{ K^\top  \mathcal{W}_{m, N} K + \lambda_N I}^{-1}\prns{\hat K - K}^\top \mathcal{W}_{m, N} K\zeta_N,
\end{align*}
and 
\begin{align*}
 \prns{ K^\top \mathcal{W}_{m, N} K + \lambda_N I}^{-1} K^\top  \mathcal{W}_{m, N} K\theta_{\min}^* -\theta_{\min}^* = -\lambda_N \prns{ K^\top  \mathcal{W}_{m, N} K + \lambda_N I}^{-1}\theta_{\min}^* = - \lambda_N \zeta_N.
 \end{align*}
We can then decompose the estimation erros of $\hat\theta$ as follows: 
 \begin{align}\label{eq: thm1-eq3}
 \hat\theta-\theta_{\min}^* = \mathcal R_1 + \mathcal R_2 + \mathcal R_3 + \mathcal R_4,
 \end{align}
 where 
\begin{align*}
\mathcal R_1 &= -\lambda_N \zeta_N, 
\end{align*}
and 
\begin{align*}
\mathcal R_2 &= \prns{\hat K^\top \mathcal{W}_{m, N} \hat K + \lambda_N I}^{-1}\hat K^\top \mathcal{W}_{m, N}  \prns{\hat b - \hat K\theta_{\min}^*} \\ 
    =& \prns{ K^\top  \mathcal{W}_{m, N} K + \lambda_N I}^{-1}K^\top  \mathcal{W}_{m, N} \prns{\hat b - \hat K\theta_{\min}^*} \\
    &+ \bracks{\prns{\hat K^\top  \mathcal{W}_{m, N} \hat K + \lambda_N I}^{-1} - \prns{ K^\top  \mathcal{W}_{m, N} K + \lambda_N I}^{-1}}\hat K^\top  \mathcal{W}_{m, N} \prns{\hat b - \hat K\theta_{\min}^*} \\
    &+ \prns{ K^\top  \mathcal{W}_{m, N} K + \lambda_N I}^{-1}\prns{\hat K - K}^\top  {\mathcal{W}_{m, N}} \prns{\hat b - \hat K\theta_{\min}^*} = \mathcal R_{2, a} + \mathcal R_{2, b} + \mathcal R_{2, c},
\end{align*}
and 
\begin{align*}
\mathcal R_3 =& -\lambda_N  \prns{\hat K^\top \mathcal{W}_{m, N} \hat K + \lambda_N I}^{-1}
    \hat K^\top \mathcal{W}_{m, N} \prns{\hat K - K}\zeta_N  \\
    =& -\lambda_N\bracks{\prns{\hat K^\top \mathcal{W}_{m, N}\hat K + \lambda_N I}^{-1} - \prns{ K^\top\mathcal{W}_{m, N}  K + \lambda_N I}^{-1}}\hat K^\top \mathcal{W}_{m, N} \prns{\hat K - K}\zeta_N \\
     &- \lambda_N\prns{ K^\top \mathcal{W}_{m, N} K + \lambda_N I}^{-1} K^\top \mathcal{W}_{m, N}\prns{\hat K - K}\zeta_N  \\
    &-  \lambda_N\prns{ K^\top \mathcal{W}_{m, N} K + \lambda_N I}^{-1}\prns{\hat K - K}^\top \mathcal{W}_{m, N} \prns{\hat K - K}\zeta_N \\
    =& \mathcal R_{3, a} + \mathcal R_{3, b} + \mathcal R_{3, c}, 
\end{align*}
and 
\begin{align*}
\mathcal R_4 
    =& -\lambda_N\prns{\hat K^\top \mathcal{W}_{m, N} \hat K + \lambda_N I}^{-1}\prns{\hat K - K}^\top \mathcal{W}_{m, N} K \zeta_N \\
    =& -\lambda_N \bracks{\prns{\hat K^\top \mathcal{W}_{m, N}\hat K + \lambda_N I}^{-1} - \prns{ K^\top \mathcal{W}_{m, N} K + \lambda_N I}^{-1}}\prns{\hat K - K}^\top \mathcal{W}_{m, N} K\zeta_N \\
    &- \lambda_N \prns{ K^\top  \mathcal{W}_{m, N} K + \lambda_N I}^{-1}\prns{\hat K - K}^\top \mathcal{W}_{m, N} K\zeta_N \\
    =& \mathcal R_{4, a} + \mathcal R_{4, b}.
\end{align*}

\textbf{Step II: Bounding Error Terms. }
Before we bound each error term in \cref{eq: thm1-eq3} respectively, we first note that according to \cref{lemma: support}, we have 
\begin{align*}
&\|\zeta_N\| = \|\prns{ K^\top  \mathcal{W}_{m, N} K + \lambda_N I}^{-1}K^\top K\phi_\theta\| = \mathcal{O}_p(1), \\
&\|\prns{ K^\top  \mathcal{W}_{m, N} K + \lambda_N I}^{-1}K^\top  \mathcal{W}_{m, N}\| = \mathcal{O}_p(1),\\
&\left\|\mathcal{W}_{m, N}\hat K \prns{\hat K^\top \mathcal{W}_{m, N} \hat K + \lambda_N I}^{-1}\hat K^\top\mathcal{W}_{m, N}\right\| = \mathcal{O}_p(1), \\
&\left\|\prns{\hat K^\top \mathcal{W}_{m, N} \hat K + \lambda_N I}^{-1}\hat K^\top \mathcal{W}_{m, N}\right\| = \mathcal{O}_p\prns{\frac{1}{\sqrt{\lambda_N}}}.
\end{align*}
Moreover, we can easily show that 
\begin{align*}
&\left\| {\hat b - \hat K\theta^*_{\min}}\right\| = \left\| {\hat b - b - \prns{\hat K- K}\theta^*_{\min}}\right\| = \mathcal{O}_p\prns{\frac{1}{\sqrt{N}}}, ~~ \left\|{\hat K - K}\right\| = \mathcal{O}_p\prns{\frac{1}{\sqrt{N}}},  \\
&\|\prns{\hat K^\top \mathcal{W}_{m, N} \hat K + \lambda_N I}^{-1}\| =  \mathcal{O}_p\prns{\frac{1}{\lambda_N}}, ~~ \|\prns{ K^\top \mathcal{W}_{m, N} K + \lambda_N I}^{-1}\| =  \mathcal{O}_p\prns{\frac{1}{\lambda_N}}, \\
&\qquad\qquad\qquad\qquad \left\|\prns{ K^\top \mathcal{W}_{m, N} K + \lambda_N I}^{-1}K^\top \mathcal{W}_{m, N}\right\| = \mathcal{O}_p(1).
\end{align*}
In addition,
\begin{align*}
&\left\|\bracks{\prns{\hat K^\top \mathcal{W}_{m, N}\hat K + \lambda_N I}^{-1} - \prns{ K^\top \mathcal{W}_{m, N} K + \lambda_N I}^{-1}}\hat K^\top\mathcal{W}_{m, N}\right\| \\
    \overset{\cref{eq: two-decompose-2}}{\le} & 
        \left\|\prns{ K^\top \mathcal{W}_{m, N} K + \lambda_N I}^{-1}\right\| \left\|{\hat K - K}\right\|\left\|\mathcal{W}_{m, N}  \hat K \prns{\hat K^\top \mathcal{W}_{m, N} \hat K + \lambda_N I}^{-1}\hat K^\top \mathcal{W}_{m, N} \right\| \nonumber \\
    +& \left\|\prns{ K^\top \mathcal{W}_{m, N} K + \lambda_N I}^{-1}K^\top \mathcal{W}_{m, N}\right\|\left\|{\hat K - K}\right\|\left\|\prns{\hat K^\top \mathcal{W}_{m, N} \hat K + \lambda_N I}^{-1}\hat K^\top \mathcal{W}_{m, N} \right\| \\
    =&  \mathcal{O}_p\prns{\frac{1}{\sqrt{N}}\prns{\frac{1}{\lambda_N} + \frac{1}{\sqrt{\lambda_N}}}} = \mathcal{O}_p\prns{\frac{1}{\sqrt{N}\lambda_N}},
\end{align*}
and similarly, 
\begin{align*}
&\left\|{\prns{\hat K^\top \mathcal{W}_{m, N} \hat K + \lambda_N I}^{-1} - \prns{ K^\top \mathcal{W}_{m, N} K + \lambda_N I}^{-1}}\right\| \\
    \overset{\cref{eq: two-decompose-2}}{\le} & 
        \left\|\prns{ K^\top \mathcal{W}_{m, N} K + \lambda_N I}^{-1}\right\| \left\|{\hat K - K}\right\|\left\| \mathcal{W}_{m, N} \hat K \prns{\hat K^\top \mathcal{W}_{m, N} \hat K + \lambda_N I}^{-1}\right\| \nonumber \\
    +& \left\|\prns{ K^\top \mathcal{W}_{m, N} K + \lambda_N I}^{-1}K^\top \mathcal{W}_{m, N}\right\|\left\|{\hat K - K}\right\|\left\|\prns{\hat K^\top \mathcal{W}_{m, N} \hat K + \lambda_N I}^{-1}\right\| \\
    =& \mathcal{O}_p\prns{\frac{1}{\sqrt{N}}\prns{\frac{1}{\lambda_N^{3/2}} + \frac{1}{\lambda_N}}}.
\end{align*}
In addition,
\begin{align*}
&\left\|\bracks{\prns{\hat K^\top \mathcal{W}_{m, N}\hat K + \lambda_N I}^{-1} - \prns{ K^\top \mathcal{W}_{m, N} K + \lambda_N I}^{-1}}\hat K^\top\mathcal{W}_{m, N}\right\| \\
    \overset{\cref{eq: two-decompose-2}}{\le} & 
        \left\|\prns{ K^\top \mathcal{W}_{m, N} K + \lambda_N I}^{-1}\right\| \left\|{\hat K - K}\right\|\left\|\mathcal{W}_{m, N}  \hat K \prns{\hat K^\top \mathcal{W}_{m, N} \hat K + \lambda_N I}^{-1}\hat K^\top \mathcal{W}_{m, N} \right\| \nonumber \\
    +& \left\|\prns{ K^\top \mathcal{W}_{m, N} K + \lambda_N I}^{-1}K^\top \mathcal{W}_{m, N}\right\|\left\|{\hat K - K}\right\|\left\|\prns{\hat K^\top \mathcal{W}_{m, N} \hat K + \lambda_N I}^{-1}\hat K^\top \mathcal{W}_{m, N} \right\| \\
    =&  \mathcal{O}_p\prns{\frac{1}{\sqrt{N}}\prns{\frac{1}{\lambda_N} + \frac{1}{\sqrt{\lambda_N}}}} = \mathcal{O}_p\prns{\frac{1}{\sqrt{N}\lambda_N}},
\end{align*}
and similarly, 
\begin{align*}
&\left\|{\prns{\hat K^\top \mathcal{W}_{m, N} \hat K + \lambda_N I}^{-1} - \prns{ K^\top \mathcal{W}_{m, N} K + \lambda_N I}^{-1}}\right\| \\
    \overset{\cref{eq: two-decompose-2}}{\le} & 
        \left\|\prns{ K^\top \mathcal{W}_{m, N} K + \lambda_N I}^{-1}\right\| \left\|{\hat K - K}\right\|\left\| \mathcal{W}_{m, N} \hat K \prns{\hat K^\top \mathcal{W}_{m, N} \hat K + \lambda_N I}^{-1}\right\| \nonumber \\
    +& \left\|\prns{ K^\top \mathcal{W}_{m, N} K + \lambda_N I}^{-1}K^\top \mathcal{W}_{m, N}\right\|\left\|{\hat K - K}\right\|\left\|\prns{\hat K^\top \mathcal{W}_{m, N} \hat K + \lambda_N I}^{-1}\right\| \\
    =& \mathcal{O}_p\prns{\frac{1}{\sqrt{N}}\prns{\frac{1}{\lambda_N^{3/2}} + \frac{1}{\lambda_N}}}.
\end{align*}

Now we are ready to bound each error terms from $\mathcal R_1$ to $\mathcal R_4$. 

First, 
\begin{align*}
\|\mathcal R_1\| \le \lambda_N\|\zeta_N\|= \mathcal{O}_p(\lambda_N).
\end{align*}
Second, 
\begin{align*}
&\|\mathcal R_{2, a}\| \le  \|\prns{ K^\top  \mathcal{W}_{m, N} K + \lambda_N I}^{-1}K^\top\|\|  \mathcal{W}_{m, N}\|\left\| {\hat b - \hat K\theta_{\min}^*}\right\| = \mathcal{O}_p\prns{\frac{1}{\sqrt{N}}}, \\
&\|\mathcal R_{2, b}\| \le \left\|\bracks{\prns{\hat K^\top  \mathcal{W}_{m, N} \hat K + \lambda_N I}^{-1} - \prns{ K^\top  \mathcal{W}_{m, N} K + \lambda_N I}^{-1}}\hat K^\top  \mathcal{W}_{m, N}\right\| \left\|{\hat b - \hat K\theta_{\min}^*} \right\| =  \mathcal{O}_p\prns{\frac{1}{N\lambda_N}}, \\
&\|\mathcal R_{2, c}\| \le \|\prns{ K^\top  \mathcal{W}_{m, N} K + \lambda_N I}^{-1}\|\|{\hat K - K}\| \|{\mathcal{W}_{m, N}}\| \|{\hat b - \hat K\theta_{\min}^*}\| = \mathcal{O}_p\prns{\frac{1}{N\lambda_N}}.
\end{align*}
Third, 
\begin{align*}
&\|\mathcal R_{3, a}\| \le \lambda_N\left\|\bracks{\prns{\hat K^\top \mathcal{W}_{m, N}\hat K + \lambda_N I}^{-1} - \prns{ K^\top\mathcal{W}_{m, N}  K + \lambda_N I}^{-1}}\hat K^\top \mathcal{W}_{m, N}\right\| \left\|{\hat K - K}\right\| \left\| \zeta_N \right\| \\
&\qquad\qquad\qquad\qquad\qquad\qquad\qquad\qquad\qquad\qquad\qquad\qquad\qquad\qquad\qquad\qquad\qquad\qquad = \mathcal{O}_p\prns{\frac{1}{N}}, \\
&\|\mathcal R_{3, b}\| \le \lambda_N\|\prns{ K^\top \mathcal{W}_{m, N} K + \lambda_N I}^{-1} K^\top \mathcal{W}_{m, N}\|\|{\hat K - K}\|\|\zeta_N\| = \mathcal{O}_p\prns{\frac{\lambda_N}{\sqrt{N}}}, \\
&\|\mathcal R_{3, c}\|\le \lambda_N\left\|\prns{ K^\top \mathcal{W}_{m, N} K + \lambda_N I}^{-1}\right\|\|{\hat K - K}\|^2 \|\mathcal{W}_{m, N}\| \|\zeta_N\| = \mathcal{O}_p\prns{\frac{1}{N}}.
\end{align*}
Fourth, 
\begin{align*}
&\|\mathcal R_{4, a}\| \le \lambda_N \left\|{\prns{\hat K^\top \mathcal{W}_{m, N}\hat K + \lambda_N I}^{-1} - \prns{ K^\top \mathcal{W}_{m, N} K + \lambda_N I}^{-1}}\right\|\|{\hat K - K}\|\|\mathcal{W}_{m, N}\| \|K\|\|\zeta_N\| \\
&\qquad\qquad\qquad\qquad\qquad\qquad\qquad\qquad\qquad\qquad\qquad\qquad\qquad\qquad = \mathcal{O}_p\prns{\frac{1}{N\sqrt{\lambda_N}}},\\
&\|\mathcal{R}_{4, b}\| \le  \lambda_N \|\prns{ K^\top  \mathcal{W}_{m, N} K + \lambda_N I}^{-1}\|\|{\hat K - K}\| \|\mathcal{W}_{m, N}\|\| K\|\|\zeta_N\| = \mathcal{O}_p\prns{\frac{1}{\sqrt{N}}}.
\end{align*}
Plugging these error bounds into \cref{eq: thm1-eq3} leads to 
\begin{align*}
\|\hat\theta - \theta_{\min}^*\| = \mathcal{O}_p\prns{\lambda_N +\frac{1}{N\lambda_N} + \frac{1}{\sqrt{N}}}.
\end{align*}
\end{proof}

\subsubsection{Proof of \cref{thm: mu-asymp}}
\begin{proof}[Proof of \cref{thm: mu-asymp}]
It is easy to show that 
\begin{align*}
\hat\gamma - \gamma^* 
    &= \frac{1}{\En{A}}\braces{\En{A\prns{\tilde W^\top\theta^*_{\min} - \gamma^*}} + \En{A\tilde W^\top\prns{\hat\theta-\theta^*_{\min}}}} \\
    &= \frac{1}{\Eb{A}}\braces{\En{A\prns{\tilde W^\top\theta^*_{\min} - \gamma^*}} + \Eb{A\tilde W^\top}\prns{\hat\theta-\theta^*_{\min}}} + \smallO_p(N^{-1/2})
\end{align*}
In the rest of the proof, we first prove that 
\begin{align}
\sqrt{N} \Eb{A\tilde W^\top}\prns{\hat\theta-\theta^*_{\min}} 
    &= \frac{1}{\sqrt{N}}\sum_{i=1}^N \tilde\Psi\prns{\mathcal{W}_{m, N}} \prns{1-A_i}\tilde Z_i\prns{Y_{i, 0} - \tilde W_i^\top \theta_{\min}^*} \notag \\
    &\qquad\qquad\qquad\qquad + \mathcal{O}_p\prns{\lambda_N \sqrt{N} + \frac{1}{\sqrt{\lambda_N N}}},\label{eq: step1}
\end{align}
where
\begin{align*}
\tilde\Psi(\mathcal{W}_{m, N}) = \Eb{A \tilde W^\top}\braces{{\Eb{(1-A)\tilde W Z^\top}} \mathcal{W}_{m, N}\Eb{(1-A)\tilde Z \tilde W^\top} + \lambda_N I}^{-1}{\Eb{(1-A)\tilde W \tilde Z^\top}}\mathcal{W}_{m, N}.
\end{align*}
Then we prove the conclusion by showing that  
\begin{align}
&\frac{1}{\sqrt{N}}\sum_{i=1}^N \tilde\Psi\prns{\mathcal{W}_{m, N}} \prns{1-A_i}\tilde Z_i\prns{Y_{i, 0} - \tilde W_i^\top \theta_{\min}^*} \notag \\
=& \frac{1}{\sqrt{N}}\sum_{i=1}^N \Psi\prns{\mathcal{W}_{m, \infty}} \prns{1-A_i}\tilde Z_i\prns{Y_{i, 0} - \tilde W_i^\top \theta_{\min}^*} + \smallO_p\prns{1},\label{eq: step2}
\end{align}
where 
\begin{align*}
\Psi(\mathcal{W}_{m, \infty}) = \Eb{A \tilde W^\top}\braces{{\Eb{(1-A)\tilde W Z^\top}} \mathcal{W}_{m, \infty}\Eb{(1-A)\tilde Z \tilde W^\top}}^{+}{\Eb{(1-A)\tilde W \tilde Z^\top}}\mathcal{W}_{m, \infty}.
\end{align*}
To prove \cref{eq: step1}, we first note that by following the statement 2 in \cref{lemma: two-key-facts} and the proof of \cref{lemma: theta-asymp}, we have 
\begin{align*}
&\left\|\Eb{A\tilde W^\top}\bracks{\prns{\hat K^\top \Lambda_N^{-1}\hat K + \lambda_N I}^{-1} - \prns{ K^\top \mathcal{W}_{m, N} K + \lambda_N I}^{-1}}\hat K^\top\mathcal{W}_{m, N}\right\| \\
=&\left\|\phi_W^\top K^\top K\bracks{\prns{\hat K^\top \Lambda_N^{-1}\hat K + \lambda_N I}^{-1} - \prns{ K^\top \mathcal{W}_{m, N} K + \lambda_N I}^{-1}}\hat K^\top\mathcal{W}_{m, N}\right\| \\
\le & 
        \left\|\phi_W^\top K^\top K\prns{ K^\top \mathcal{W}_{m, N} K + \lambda_N I}^{-1}\right\| \left\|{\hat K - K}\right\|\left\|\mathcal{W}_{m, N} \hat K \prns{\hat K^\top \mathcal{W}_{m, N}\hat K + \lambda_N I}^{-1}\hat K^\top \mathcal{W}_{m, N}\right\| \nonumber \\
    +& \left\|\phi_W^\top K^\top K \prns{ K^\top \mathcal{W}_{m, N} K + \lambda_N I}^{-1}K^\top \mathcal{W}_{m, N}\right\|\left\|{\hat K - K}\right\|\left\|\prns{\hat K^\top \mathcal{W}_{m, N}\hat K + \lambda_N I}^{-1}\hat K^\top \mathcal{W}_{m, N}\right\| \\
    =&  O_p\prns{\frac{1}{\sqrt{N}}\prns{1 + \frac{1}{\sqrt{\lambda_N}}}} = O_p\prns{\frac{1}{\sqrt{N\lambda_N}}},
\end{align*}
and similarly, 
\begin{align*}
&\left\|\Eb{A\tilde W^\top}\bracks{\prns{\hat K^\top \mathcal{W}_{m, N}\hat K + \lambda_N I}^{-1} - \prns{ K^\top \mathcal{W}_{m, N} K + \lambda_N I}^{-1}}\right\| \\
    \le & 
        \left\|\phi_W^\top K^\top K\prns{ K^\top \mathcal{W}_{m, N} K + \lambda_N I}^{-1}\right\| \left\|{\hat K - K}\right\|\left\| \mathcal{W}_{m, N}\hat K \prns{\hat K^\top \mathcal{W}_{m, N}\hat K + \lambda_N I}^{-1}\right\| \nonumber \\
    +& \left\|\phi_W^\top K^\top K\prns{ K^\top \mathcal{W}_{m, N} K + \lambda_N I}^{-1}K^\top \mathcal{W}_{m, N}\right\|\left\|{\hat K - K}\right\|\left\|\prns{\hat K^\top \mathcal{W}_{m, N}\hat K + \lambda_N I}^{-1}\right\| \\
    &= O_p\prns{\frac{1}{\sqrt{N}}\prns{\frac{1}{\lambda_N^{1/2}} + \frac{1}{\lambda_N}}} = O_p\prns{\frac{1}{\sqrt{N}\lambda_N}}
\end{align*}
Next, we bound $\Eb{A\tilde W}\mathcal{R}_j$ for $\mathcal{R}_j, j = 1, \dots, 4$ in \cref{eq: thm1-eq3} in the proof of \cref{lemma: theta-asymp}:
\begin{align*}
\left\|\Eb{A\tilde W^\top}\mathcal R_{1}\right\| 
    &\le \left\|\Eb{A\tilde W^\top}\right\|\|\mathcal R_{1}\| = O\prns{\lambda_N},
\end{align*}
and 
\begin{align*}
 &\left\|\Eb{A \tilde W^\top}\mathcal R_{2, b}\right\| 
    \le \|\hat b - \hat K\theta^*_{\min}\|\\
 &\qquad \times\left\|\Eb{A \tilde W^\top}\bracks{\prns{\hat K^\top \mathcal{W}_{m, N} \hat K + \lambda_N I}^{-1} - \prns{ K^\top \mathcal{W}_{m, N} K + \lambda_N I}^{-1}}\hat K^\top \mathcal{W}_{m, N}\right\| =  \mathcal{O}_p\prns{\frac{1}{N\sqrt{\lambda_N}}}, \\
&\norm{\Eb{A \tilde W^\top}\mathcal{R}_{2c}}
    \le \|\prns{ K^\top \mathcal{W}_{m, N}K + \lambda_N I}^{-1} K^\top K \phi_W\|\|{\hat K - K}\|\|{\mathcal{W}_{m, N}}\| \|{\hat b - \hat K\theta^*_{\min}} \| = \mathcal{O}_p\prns{\frac{1}{N}},
 \end{align*}
 and 
\begin{align*}
&\norm{\Eb{A\tilde W^\top}\mathcal R_{3, a}}
    \le \|\lambda_N \zeta_N\|\|\hat K - K\| \\
    &\qquad\qquad\quad \times \left\|\Eb{A\tilde W^\top}\bracks{\prns{\hat K^\top \mathcal{W}_{m, N}\hat K + \lambda_N I}^{-1} - \prns{ K^\top \mathcal{W}_{m, N}K + \lambda_N I}^{-1}}\hat K^\top\mathcal{W}_{m, N}\right\| = \mathcal{O}_p\prns{\frac{\sqrt{\lambda_N}}{N}}, \\
&\norm{\Eb{A\tilde W^\top}\mathcal R_{3, b}}
    \le  \|\lambda_N \zeta_N\|\|\hat K - K\|\|\mathcal{W}_{m, N} K\prns{ K^\top \mathcal{W}_{m, N} K + \lambda_N I}^{-1}K^\top K\phi_W\| = \mathcal{O}_p\prns{\frac{\lambda_N}{\sqrt{N}}} \\
&\norm{\Eb{A\tilde W^\top}\mathcal R_{3, c}}
    \le \|\lambda_N \zeta_N\|\|\hat K - K\|^2 \|\prns{ K^\top K + \lambda_N I}^{-1}K^\top K\phi_W\|\|\mathcal{W}_{m, N}\| = \mathcal{O}_p\prns{\frac{\lambda_N}{N}},
\end{align*}
and 
\begin{align*}
&\norm{\Eb{A\tilde W^\top}\mathcal R_{4,a}} 
    \le \lambda_N\|\hat K - K\|\|K\zeta_N\|\|\mathcal{W}_{m, N}\|\\
    &\qquad\qquad\qquad\qquad \times \left\|\Eb{A\tilde W^\top}\bracks{\prns{\hat K^\top \mathcal{W}_{m, N}\hat K + \lambda_N I}^{-1} - \prns{ K^\top \mathcal{W}_{m, N} K + \lambda_N I}^{-1}}\right\| = \mathcal{O}_p\prns{\frac{1}{{N}}}, \\
&\norm{\Eb{A\tilde W^\top}\mathcal R_{4,b}}  \le \lambda_N\|K\zeta_N\|\|\hat K - K\|\|\mathcal{W}_{m, N}\|\|\prns{ K^\top \mathcal{W}_{m, N} K + \lambda_N I}^{-1}K^\top  K\phi_W\| = \mathcal{O}_p\prns{\frac{\lambda_N}{\sqrt{N}}}. 
\end{align*}
According to \cref{eq: thm1-eq3} in the proof of \cref{lemma: theta-asymp}, we therefore have 
\begin{align*}
&\sqrt{N}\Eb{A\tilde W^\top}\prns{\hat\theta - \theta_{\min}^*}\\
    =& \sqrt{N}\Eb{A\tilde W^\top}\mathcal R_{2, a} + \mathcal{O}_p\prns{\lambda_N \sqrt{N} + \frac{1}{\sqrt{\lambda_N N}}} \\
    =&  \sqrt{N}\Eb{A\tilde W^\top}\prns{ K^\top  \mathcal{W}_{m, N} K + \lambda_N I}^{-1}K^\top  \mathcal{W}_{m, N} \prns{\hat b - \hat K\theta_{\min}^*} + \mathcal{O}_p\prns{\lambda_N \sqrt{N} + \frac{1}{\sqrt{\lambda_N N}}} = \cref{eq: step1}.
\end{align*}
Finally, 
\begin{align}
&\Eb{A\tilde W^\top}\prns{ K^\top  \mathcal{W}_{m, N} K + \lambda_N I}^{-1}K^\top  \mathcal{W}_{m, N} - \Eb{A\tilde W^\top}\prns{ K^\top  \mathcal{W}_{m, \infty} K}^{+}K^\top  \mathcal{W}_{m, \infty}\label{eq: step2-1} \\
=& \phi_W^\top K^\top K \bracks{\prns{ K^\top  \mathcal{W}_{m, N} K + \lambda_N I}^{-1} - \prns{ K^\top  \mathcal{W}_{m, \infty} K + \lambda_N I}^{-1}}K^\top  \mathcal{W}_{m, N} \nonumber  \\
+& \phi_W^\top K^\top K \prns{ K^\top  \mathcal{W}_{m, \infty} K + \lambda_N I}^{-1}K^\top\prns{\mathcal{W}_{m, N} - \mathcal{W}_{m, \infty}} \nonumber \\
+& \Eb{A\tilde W^\top}\braces{\prns{ K^\top  \mathcal{W}_{m, \infty} K + \lambda_N I}^{-1} - \prns{ K^\top  \mathcal{W}_{m, \infty} K}^{+}}K^\top  \mathcal{W}_{m, \infty} \nonumber 
\end{align}
Here because $\|{\mathcal{W}_{m, N} - \mathcal{W}_{m, \infty}}\| = \smallO_p\prns{1}$, we have 
\begin{align*}
&\norm{\phi_W^\top K^\top K \bracks{\prns{ K^\top  \mathcal{W}_{m, N} K + \lambda_N I}^{-1} - \prns{ K^\top  \mathcal{W}_{m, \infty} K + \lambda_N I}^{-1}}K^\top  \mathcal{W}_{m, N}} \\
\le& \|\phi_W^\top K^\top K \prns{ K^\top  \mathcal{W}_{m, \infty} K + \lambda_N I}^{-1}\|\|K^\top\prns{\mathcal{W}_{m, N} - \mathcal{W}_{m, \infty}}K\|\|\prns{ K^\top  \mathcal{W}_{m, N} K + \lambda_N I}^{-1}K^\top  \mathcal{W}_{m, N}\| \\
=& \smallO_p\prns{1}, 
\end{align*}
and 
\begin{align*}
&\norm{\phi_W^\top K^\top K \prns{ K^\top  \mathcal{W}_{m, \infty} K + \lambda_N I}^{-1}K^\top\prns{\mathcal{W}_{m, N} - \mathcal{W}_{m, \infty}}}\\
\le & \norm{\phi_W^\top K^\top K \prns{ K^\top  \mathcal{W}_{m, \infty} K + \lambda_N I}^{-1}K^\top}\norm{\mathcal{W}_{m, N} - \mathcal{W}_{m, \infty}} = \smallO_p\prns{1}.
\end{align*}
According to \cref{lemma: pseudo-inverse}, we have 
\begin{align*}
\norm{\Eb{A\tilde W^\top}\braces{\prns{ K^\top  \mathcal{W}_{m, \infty} K + \lambda_N I}^{-1} - \prns{ K^\top  \mathcal{W}_{m, \infty} K}^{+}}K^\top  \mathcal{W}_{m, \infty}} = \mathcal{O}\prns{\lambda_N} = \smallO\prns{1}.
\end{align*}
Therefore, \cref{eq: step2-1} is $\smallO_p\prns{1}$, which in turn proves \cref{eq: step2}
\end{proof}

\subsubsection{Proof of \cref{thm: CI} and \cref{thm: GMM-optimality}}
\begin{proof}[Proof of \cref{thm: CI}]
According to Central Limit Theorem, we have that 
\begin{align*}
\frac{\sqrt{N}\prns{\hat\gamma - \gamma^*}}{\sigma\prns{\mathcal{W}_{m, \infty}}} \overset{\mathrm{d}}{\to} \mathcal{N}\prns{0, 1}.
\end{align*}
As long as we can prove that 
\begin{align}\label{eq: consis-var}
\hat\sigma^2\prns{\mathcal{W}_{m, N}} - \sigma^2\prns{\mathcal{W}_{m, \infty}} \overset{\mathrm{p}}{\to} 0,
\end{align}
we can use Slutsky's theorem to prove that  
\begin{align*}
\frac{\sqrt{N}\prns{\hat\gamma - \gamma^*}}{\hat\sigma\prns{\mathcal{W}_{m, N}}} \overset{\mathrm{d}}{\to} \mathcal{N}\prns{0, 1},
\end{align*}
which in turn implies the asymptotic validity of the confidence interval. In the rest of the proof, we show how to prove \cref{eq: consis-var}.

First, note that 
\begin{align*}
\hat\sigma^2\prns{\mathcal{W}_{m, N}} - \sigma^2\prns{\mathcal{W}_{m, \infty}} 
    =& \En{\hat\psi^2\prns{O; \hat\theta, \hat\gamma, \mathcal{W}_{m, N}}} - \Eb{\psi^2\prns{O; \theta_{\min}^*, \gamma^*, \mathcal{W}_{m, \infty}}} \\
    =& \En{\hat\psi^2\prns{O; \hat\theta, \hat\gamma, \mathcal{W}_{m, N}}} - \En{\psi^2\prns{O; \theta^*_{\min}, \gamma^*, \mathcal{W}_{m, \infty}}} \\
    +& \En{\psi^2\prns{O; \theta^*_{\min}, \gamma^*, \mathcal{W}_{m, \infty}}} - \Eb{\psi^2\prns{O; \theta^*_{\min}, \gamma^*, \mathcal{W}_{m, \infty}}}.
\end{align*}
We denote the three differences above as $\mathcal{R}_{5}, \mathcal{R}_{6}$ respectively, and denote
\begin{align*}
\Delta\prns{O} = \abs{\hat\psi \prns{O; \hat\theta, \hat\gamma, \mathcal{W}_{m, N}} - \psi \prns{O; \theta^*_{\min}, \gamma^*, \mathcal{W}_{m, \infty}}}.
\end{align*}
Then 
\begin{align}
\mathcal{R}_5
    &= \En{\hat\psi^2\prns{O; \hat\theta, \hat\gamma, \mathcal{W}_{m, N}}} - \Eb{\psi^2\prns{O; \theta_{\min}^*, \gamma^*, \mathcal{W}_{m, \infty}}} \nonumber \\
    &\le \En{\Delta\prns{O}\prns{\Delta\prns{O} + 2\psi \prns{O; \theta^*_{\min}, \gamma^*, \mathcal{W}_{m, \infty}}}} \nonumber  \\
    &\le \En{\Delta^2\prns{O}} + 2\prns{\En{\Delta^2\prns{O}}}^{1/2}\prns{\En{\psi^2\prns{O; \theta^*_{\min}, \gamma^*, \mathcal{W}_{m, \infty}}}}^{1/2},\label{eq: thm-CI-R5}
\end{align}
where 
\begin{align*}
{\En{\psi^2\prns{O; \theta^*_{\min}, \gamma^*, \mathcal{W}_{m, N}}}} = \mathcal{R}_6 + \sigma^2\prns{\mathcal{W}_{m, \infty}}.
\end{align*}
We will show that $\mathcal{R}_6 = \smallO_p\prns{1}$. So we only need to show $\prns{\En{\Delta^2\prns{O}}}^{1/2} = \smallO_p\prns{1}$ in order to show that $\mathcal{R}_5  = \smallO_p\prns{1}$.

Recall that 
\begin{align*}
&\psi\prns{O_i; \theta_{\min}^*, \gamma^*, \mathcal{W}_{m, \infty}} = -\frac{1}{\Eb{A}} 
\braces{{A_i\prns{\gamma^* -\tilde W^\top_i\theta_{\min}^*}}+ \Psi\prns{\mathcal{W}_{m, \infty}} \prns{1-A_i}\tilde Z_i\prns{Y_{i, 0} - \tilde W_i^\top \theta_{\min}^*}}, \\
&    \hat\psi\prns{O_i; \hat\theta, \hat\gamma, \mathcal{W}_{m, N}} = -\frac{1}{\En{A}} 
\braces{{A_i\prns{\hat\gamma -\tilde W^\top_i\hat\theta}}+ \hat \Psi\prns{\mathcal{W}_{m, N}} \prns{1-A_i}\tilde Z_i\prns{Y_{i, 0} - \tilde W_i^\top \hat\theta}}, 
\end{align*}
where 
\begin{align*}
&\Psi(\mathcal{W}_{m, \infty}) = \Eb{A \tilde W^\top}\prns{K^\top \mathcal{W}_{m, \infty}K}^{+}K^\top\mathcal{W}_{m, \infty}, \\
&\hat\Psi\prns{\mathcal{W}_{m, N}} = \En{A \tilde W^\top}\prns{\hat K^\top\mathcal{W}_{m, N}\hat K + \lambda_N I}^{-1}\hat K^\top\mathcal{W}_{m, N}. 
\end{align*}
Thus 
\begin{align}
&\prns{\En{\Delta^2\prns{O}}}^{1/2} \nonumber \\
=& \braces{\En{\prns{\hat\psi\prns{O; \hat\theta, \hat\gamma, \mathcal{W}_{m, N}}  - \psi\prns{O; \theta_{\min}^*, \gamma^*, \mathcal{W}_{m, \infty}}}^2}}^{1/2}\nonumber \\
=& \prns{\En{\Delta^2_1\prns{O}}}^{1/2} + \prns{\En{\Delta^2_2\prns{O}}}^{1/2}  + \prns{\En{\Delta^2_3\prns{O}}}^{1/2} + \prns{\En{\Delta^2_4\prns{O}}}^{1/2}, \label{eq: thm-CI-delta}
\end{align}
where 
\begin{align*}
&{{\Delta_1\prns{O}}} = -\frac{1}{{\En{A}}} {A\prns{\hat\gamma-\gamma^* - \tilde W^\top\prns{\hat\theta-\theta_{\min}^*}}}, \\
&\Delta_2\prns{O} = \prns{\frac{1}{\Eb{A}} - \frac{1}{\En{A}}}\prns{A\prns{\gamma^*-\tilde W^\top\theta_{\min}^*} + \Psi(\mathcal{W}_{m, \infty})\prns{1-A}\tilde Z\prns{ Y_0 - \tilde W^\top\theta_{\min}^*}}, \\
&\Delta_3\prns{O} = -\Psi(\mathcal{W}_{m, \infty})\prns{1-A}\tilde Z\tilde W^\top\prns{\hat\theta-\theta^*_{\min}},\\
&\Delta_4\prns{O} = \prns{\hat\Psi\prns{\mathcal{W}_{m, N}}-\Psi(\mathcal{W}_{m, \infty})}\prns{1-A}\tilde Z\prns{ Y_0 - \tilde W^\top\theta_{\min}^*}.
\end{align*}
We now bound each term respectively. It is easy to show that 
\begin{align*}
\prns{\En{\Delta_1\prns{O}^2}}^{1/2} = \mathcal{O}_p\prns{1} \times \prns{\abs{\hat\gamma-\gamma^*}  + \|\hat\theta-\theta_0\|} = \mathcal{O}_p\prns{\frac{1}{\sqrt{N}} + {\lambda_N +\frac{1}{N \lambda_N}}},
\end{align*}
and 
\begin{align*}
\prns{\En{\Delta_2\prns{O}^2}}^{1/2} = \mathcal{O}_p\prns{1}\abs{{\frac{1}{\Eb{A}} - \frac{1}{\En{A}}}} = \mathcal{O}_p\prns{\frac{1}{\sqrt{N}}},
\end{align*}
and 
\begin{align*}
\prns{\En{\Delta_3\prns{O}^2}}^{1/2} = \mathcal{O}_p\prns{1}\times\|\hat\theta-\theta^*_{\min}\| = \mathcal{O}_p\prns{{\lambda_N + { \frac{1}{N\sqrt{\lambda_N}}}}},
\end{align*}
and 
\begin{align*}
\prns{\En{\Delta_3\prns{O}^2}}^{1/2} = \mathcal{O}_p\prns{1}\times\norm{\hat\Psi\prns{\mathcal{W}_{m, N}}-\Psi(\mathcal{W}_{m, \infty})}.
\end{align*}
Here 
\begin{align}
&\hat\Psi\prns{\mathcal{W}_{m, N}} - \Psi(\mathcal{W}_{m, \infty})\nonumber \\
=& \En{A \tilde W^\top}\prns{\hat K^\top\mathcal{W}_{m, N}\hat K + \lambda_N I}^{-1}\hat K^\top\mathcal{W}_{m, N} - \Eb{A \tilde W^\top}\prns{K^\top \mathcal{W}_{m, \infty}K}^{+}K^\top\mathcal{W}_{m, \infty}\nonumber  \\
=& \prns{\En{A \tilde W^\top} - \Eb{A \tilde W^\top}}\prns{\hat K^\top\mathcal{W}_{m, N}\hat K + \lambda_N I}^{-1}\hat K^\top\mathcal{W}_{m, N} \label{eq: phi-1}\\
+&  \Eb{A \tilde W^\top}\braces{\prns{\hat K^\top\mathcal{W}_{m, N}\hat K + \lambda_N I}^{-1} - \prns{K^\top \mathcal{W}_{m, N}K + \lambda_N I}^{-1}}\hat K^\top\mathcal{W}_{m, N} \label{eq: phi-2} \\
+& \Eb{A \tilde W^\top}\prns{K^\top \mathcal{W}_{m, N}K + \lambda_N I}^{-1}\bracks{\prns{\hat K - K}^\top\mathcal{W}_{m, N} - K^\top\prns{\mathcal{W}_{m, N} - \mathcal{W}_{m, \infty}}} \label{eq: phi-3}\\
+& \Eb{A \tilde W^\top}\braces{\prns{K^\top \mathcal{W}_{m, N}K + \lambda_N I}^{-1} - \prns{K^\top \mathcal{W}_{m, \infty}K + \lambda_N I}^{-1}}K^\top \mathcal{W}_{m, \infty}.\label{eq: phi-4} \\
+& \Eb{A \tilde W^\top}\braces{\prns{K^\top \mathcal{W}_{m, \infty}K + \lambda_N I}^{-1} - \prns{K^\top \mathcal{W}_{m, \infty}K}^{+}}K^\top \mathcal{W}_{m, \infty}.\label{eq: phi-5}
\end{align}
By following the proof of \cref{lemma: nonunique,thm: mu-asymp}, we can show that  
\begin{align*}
\norm{\cref{eq: phi-1}} = \mathcal{O}_p\prns{\frac{1}{\sqrt{N\lambda_N}}},
\end{align*}
and 
\begin{align*}
\norm{\cref{eq: phi-2}} 
    &\le \norm{\phi_W^\top K^\top K\prns{K^\top \mathcal{W}_{m, N}K + \lambda_N I}^{-1}}\norm{\hat K^\top\mathcal{W}_{m, N}\hat K - K^\top \mathcal{W}_{m, N}K} \\
    & \times \norm{\prns{\hat K^\top\mathcal{W}_{m, N}\hat K + \lambda_N I}^{-1}\hat K^\top\mathcal{W}_{m, N}} = \mathcal{O}_p\prns{\frac{1}{\sqrt{N\lambda_N}}},
\end{align*}
and 
\begin{align*}
\norm{\cref{eq: phi-3}} 
    &= \norm{\phi_W^\top K^\top K\prns{K^\top \mathcal{W}_{m, N}K + \lambda_N I}^{-1}}\prns{\norm{\prns{\hat K - K}^\top\mathcal{W}_{m, N}} + \norm{K^\top\prns{\mathcal{W}_{m, N} - \mathcal{W}_{m, \infty}}}} \\
    &= \mathcal{O}_p\prns{\frac{1}{\sqrt{N}}} + \mathcal{O}_p\prns{1}\times \norm{\mathcal{W}_{m, N} - \mathcal{W}_{m, \infty}},
\end{align*}
and 
\begin{align*}
\norm{\cref{eq: phi-4}} 
    &\le  \norm{\phi_W^\top K^\top K\prns{K^\top \mathcal{W}_{m, N}K + \lambda_N I}^{-1}}\|K^\top\prns{\mathcal{W}_{m, N} - \mathcal{W}_{m, \infty}}K\|\\
    &\times \norm{\prns{K^\top \mathcal{W}_{m, N}K + \lambda_N I}^{-1}K\mathcal{W}_{m, \infty}} = \mathcal{O}_p\prns{1}\times \norm{\mathcal{W}_{m, N} - \mathcal{W}_{m, \infty}},
\end{align*}
and according to \cref{lemma: pseudo-inverse}, 
\begin{align*}
\norm{\cref{eq: phi-5}} = \mathcal{O}\prns{\lambda_N}.
\end{align*}
Therefore, 
\begin{align*}
\norm{\hat\Psi\prns{\mathcal{W}_{m, N}}-\Psi(\mathcal{W}_{m, \infty})} = \mathcal{O}_p\prns{\lambda_N + \frac{1}{\sqrt{N\lambda_N}} + \frac{1}{\sqrt{N}}} + \mathcal{O}_p\prns{1}\times \norm{\mathcal{W}_{m, N} - \mathcal{W}_{m, \infty}}.
\end{align*}
Then according to \cref{eq: thm-CI-delta}, we have 
\begin{align*}
\prns{\En{\Delta^2\prns{O}}}^{1/2} = \mathcal{O}_p\prns{\frac{1}{\sqrt{N}} + {\lambda_N +\frac{1}{\sqrt{N \lambda_N}}  }} + \mathcal{O}_p\prns{1}\times \norm{\mathcal{W}_{m, N} - \mathcal{W}_{m, \infty}},
\end{align*}
and according to \cref{eq: thm-CI-R5}, $\mathcal{R}_5$ converges to $0$ at the same rate. 

Finally, we obviously have 
\begin{align*}
\mathcal{R}_6 = \mathcal{O}_p\prns{\frac{1}{\sqrt{N}}}.
\end{align*}
Therefore, 
\begin{align*}
\hat\sigma^2\prns{\mathcal{W}_{m, N}} - \sigma^2\prns{\mathcal{W}_{m, \infty}} = \mathcal{O}_p\prns{\lambda_N + \frac{1}{\sqrt{N\lambda_N}} + \frac{1}{\sqrt{N}}} + \mathcal{O}_p\prns{1}\times \norm{\mathcal{W}_{m, N} - \mathcal{W}_{m, \infty}}.
\end{align*}
\end{proof}

\begin{proof}[Proof of \cref{thm: GMM-optimality}]
This conclusion directly follows from \cref{thm: opt-general} in \cref{sec: app-opt}.
\end{proof}

\subsection{Proofs in \cref{sec: discuss}}
\begin{proof}[Proof of \cref{lemma: time-varying-bridge-simple}]
The conclusion of this lemma follows from the proof of \cref{lemma: time-varying-bridge-general} by additionally notating that under condition 1, 
\begin{align*}
\Eb{\eta_t \mid U_0, X, A = 0} 
  &= \Eb{\eta_t \mid U_0, A = 0} \\ 
  &= \Eb{\eta_t U_0^\top \mid A = 0}\prns{\Eb{U_0 U_0^\top \mid A = 0}}^{-1}\prns{U_0 - \Eb{U_0 \mid A = 0}},
\end{align*}
and 
\begin{align*}
\Eb{U_{-T_0} \mid U_0, X, A = 0} 
  &= \Eb{U_{-T_0} \mid U_0, A = 0} \\ 
  &= \Eb{U_{-T_0} U_0^\top \mid A = 0}\prns{\Eb{U_0 U_0^\top \mid A = 0}}^{-1}\prns{U_0 - \Eb{U_0 \mid A = 0}}.
\end{align*}
\end{proof}

\begin{proof}[Proof of \cref{thm: bridge-varying-obs}]
Note that 
\begin{align*}
    Y_{\post} 
    &= 
    \begin{bmatrix}
    V_{1}^\top \Gamma_0 U_0 + V_{1}^\top \eta_0  \\
    V_2^\top \Gamma_{1:0} U_0 + V_{2}^\top\eta_1 + V_{2}^\top\Gamma_1\eta_0 \\
    \vdots \\
    V_{T_1}^\top \Gamma_{\prns{T_1-1}:0}U_0 + V_{T_1}^\top\sum_{k=1}^{T_1}\Gamma_{\prns{T_1-1}:\prns{T_1-k+1}}\eta_{T_1-k} 
    \end{bmatrix}
    + \mb{B}_{\post}X + \epsilon_{\post} \\
    &= 
    \begin{bmatrix}
    V_{1}^\top \Gamma_0 \\
    V_2^\top \Gamma_{1:0} \\
    \vdots \\
    V_{T_1}^\top \Gamma_{\prns{T_1-1}:0}
    \end{bmatrix}
    U_0 +
    \mb{B}_{\post}X +
     \epsilon_{\post} + 
     \begin{bmatrix}
     V_1^\top\eta_0 \\
     V_2^\top\eta_1 + V_2^\top \Gamma_1\eta_0 \\
     \vdots \\
     V_{T_1}^\top \sum_{k=1}^{T_1}\Gamma_{\prns{T_1-1}:\prns{T_1-k+1}}\eta_{T_1 - k}
     \end{bmatrix}.
\end{align*}
We denote the coefficient matrix of $U_0$ in the last display as $\tilde{\mb{V}}_{\post}$. 

Note that $\prns{Y_{\pre}, Y_0}$ only depend on $U_0, U_{\pre}, X, \epsilon_{\pre}, \epsilon_0$ and $Y_{\post}$ only depend on $U_{\post}, X, \epsilon_{\post}$. Given the asserted condition, we already have $\epsilon_{\post} \perp \prns{\epsilon_0, \epsilon_\pre}$, and because $\epsilon_t \perp \prns{U_s, X, A}$ for any $t, s$ according to \cref{assump: time-varying}, we have that 
\begin{align*}
\epsilon_{\post} \perp \prns{\epsilon_0, \epsilon_\pre, U_{\pre}, U_0, X} \mid U_0, A = 0, X.
\end{align*}
Moreover, because $\eta_t \perp \prns{U_{-T_0}, \dots, U_t} \mid X$ for any $t$, we have 
\begin{align*}
U_{\post} \perp  U_{\pre} \mid X, U_0.
\end{align*}
Due to $A \perp U_s \mid X, U_0$  and $\epsilon_t \perp \prns{U_s, X, A}$  for any $s$ and $t$, we have that 
\begin{align*}
U_{\post} \perp  \prns{\epsilon_0, \epsilon_\pre, U_{\pre}, U_0, X} \mid X, A =0, U_0.
\end{align*}
Therefore, 
\begin{align*}
Y_{\post} \perp \prns{Y_{\pre}, Y_0} \mid U_0, A = 0, X.
\end{align*}
Following the proof of \cref{lemma: post-treatment}, we can show that for any $\theta^* = \prns{\theta^*_1, \theta^*_2}$ that solves \cref{eq: bridge-2-simple}, 
\begin{align*}
\Eb{Y_{\post}\prns{Y_0 - \theta_1^{*\top}Y_{\pre} - \theta_2^{*\top}X} \mid X, U, A = 0},
\end{align*}
which in turn implies 
\begin{align*}
    &\Eb{\begin{bmatrix}
    Y_{\post} \\
    X 
    \end{bmatrix}
    \prns{1-A}\prns{Y_{0} - \theta^{*\top}_1 Y_{\pre} - \theta_2^{*\top} X}}  = \mb{0}_{\prns{T_1 + d} \times 1}.
\end{align*}
Namely, $\theta^* = \prns{\theta^*_1, \theta^*_2}$ satisfies \cref{eq: bridge-2-obs}. 

Recall that in the proof of \cref{lemma: time-varying-bridge-general}, we show that  $\theta^* = \prns{\theta^*_1, \theta^*_2}$ satisfies \cref{eq: bridge-2-simple} if and only if 
\begin{align*}
V_0 = \tilde{\mb{V}}^\top_{\pre}\prns{G_{11}, G_{21}}\theta_1^*,
\end{align*}
and 
\begin{align*}
   \theta_2^* 
      = \tilde B 
      \coloneqq& b_0 - \mb{B}_\pre^\top\theta_1^* + \tilde{\mb{V}}^\top_{\pre}\prns{G_{12}, G_{22}}\theta_1^* \\
      &- \big(\Eb{U_0^\top \mid A = 0}\tilde{\mb{V}}^\top_{\pre}\prns{G_{11}, G_{21}}\theta_1^* + \Eb{X^\top \mid A = 0}\tilde{\mb{V}}^\top_{\pre}\prns{G_{12}, G_{22}} \\
      &- \sum_{m=1}^{T_0} \theta^*_{1, m}V_{-m}^\top \Gamma_{\prns{-m-1}:\prns{-T_0}}\Eb{U_{-T_0} \mid A = 0}\big)\iota_d,
\end{align*}
Moreover, from the proof of \cref{lemma: time-varying-bridge-general}, we have that 
\begin{align*}
&Y_0\prns{0}-\theta_1^{*\top} Y_{\pre}\prns{0} - \theta_2^{*\top} X \\
  =& 
  \prns{V_0 - \tilde{\mb{V}}^\top_{\pre}\prns{G_{11}, G_{21}}\theta_1^*}^\top U_0 +  \prns{\tilde B - \theta_2^{*}}^\top X + \epsilon_0 -  \theta_1^{*\top}\epsilon_{\pre}.\\
  &- 
  \sum_{m = 1}^{T_0} \theta_{1, m}V_{-m}^\top \Gamma_{\prns{-m-1}:\prns{-T_0}}\prns{U_{-T_0} - \Eb{U_{-T_0} \mid U_0, A = 0, X} }\\
  &-  \sum_{m = 1}^{T_0} \theta_{1, m} V_{-m}^\top \sum_{k=1}^{-m + T_0}\Gamma_{\prns{-m-1}:\prns{-m-k+1}}\prns{\eta_{-m-k} - \Eb{\eta_{-m-k} \mid U_0, A = 0, X}}.
\end{align*}
Since $\prns{\eta_t: t\in\post} \perp U_{-T_0} \mid X$, $\prns{\eta_t: t\in\post} \perp \prns{\eta_s: s \in \pre} \mid X$ and $A \perp U_t \mid U_0, X$ for any $t$, we have that $\prns{\eta_t: t \in \post} \perp \prns{\eta_s: s \in \pre} \mid U_0, A = 0, X$ and $\prns{\eta_t: t \in \post} \perp U_{-T_0} \mid U_0, A = 0, X$. It follows that for any $t \in \pre, s \in \post$, we have 
\begin{align*}
&\Eb{\prns{1-A}\eta_t\prns{U_{-T_0} - \Eb{U_{-T_0} \mid U_0, A = 0, X} }^\top} \\
=& \Prb{A = 0}\Eb{\eta_t \mid U_0, A = 0, X}\Eb{\prns{U_{-T_0} - \Eb{U_{-T_0} \mid U_0, A = 0, X} }^\top\mid U_0, A= 0, X} = 0,
\end{align*}
and 
\begin{align*}
&\Eb{\prns{1-A}\eta_t\prns{\eta_{s} - \Eb{\eta_{s} \mid U_0, A = 0, X} }^\top} \\
=& \Prb{A = 0}\Eb{\eta_t \mid U_0, A = 0, X}\Eb{\prns{\eta_{s} - \Eb{\eta_{s} \mid U_0, A = 0, X} }^\top \mid U_0, A= 0, X} = 0.
\end{align*}
Moreover, we have $\epsilon_{\post} \perp \prns{\epsilon_0, \epsilon_{\pre}}$ and $\epsilon_t \perp \prns{U_s, \eta_s, X, A}$ for any $t, s$. It follows that 
\begin{align*}
\Eb{\prns{1-A}\epsilon_{\post}\prns{Y_0\prns{0}-\theta_1^{*\top} Y_{\pre}\prns{0} - \theta_2^{*\top} X }} = 0.
\end{align*}
Therefore, 
\begin{align*}
    &\Eb{\begin{bmatrix}
    Y_{\post} \\
    X 
    \end{bmatrix}
    \prns{1-A}\prns{Y_{0} - \theta^{*\top}_1 Y_{\pre} - \theta_2^{*\top} X}} \\
    =& 
    \Eb{\begin{bmatrix}
    Y_{\post} \\
    X 
    \end{bmatrix}
    \prns{1-A}\prns{Y_{0}\prns{0} - \theta^{*\top}_1 Y_{\pre}\prns{0} - \theta_2^{*\top} X}} \\
    =& 
    \begin{bmatrix}
    \tilde{\mb{V}}_{\post} & \mb{B}_{\post}  \\
    0 & I  
    \end{bmatrix}
    \Eb{
    \prns{1-A}
    \begin{bmatrix}
    U_0 \\
    X
    \end{bmatrix}
    \begin{bmatrix}
    U_0^\top & X^\top
    \end{bmatrix}
    }
    \begin{bmatrix}
    V_0 - \tilde{\mb{V}}^\top_{\pre}\prns{G_{11}, G_{21}}\theta_1^* \\
    \tilde B - \theta_2^*
    \end{bmatrix}.
\end{align*} 
Since the following matrix has full  column rank:
\begin{align*}
       \begin{bmatrix}
    \tilde{\mb{V}}_{\post} & \mb{B}_{\post}  \\
    0 & I  
    \end{bmatrix}
    \Eb{
  \begin{bmatrix}
    U_0 \\
    X 
    \end{bmatrix}
    \begin{bmatrix}
    U_0^\top & X^\top 
    \end{bmatrix} \mid A = 0.
    },
\end{align*}
we have that $\theta^* = \prns{\theta^*_1, \theta^*_2}$ solves \cref{eq: bridge-2-simple} if and only if it also satisfies \cref{eq: bridge-2-obs}.

It follows that for any $\theta^*$ that solves  \cref{eq: bridge-2-obs}, we have
\begin{align*}
\gamma^* = \Eb{\theta_1^{*\top}Y_{\pre} + \theta_2^{*\top}X \mid A = 1}.
\end{align*}
Thus $\gamma^*$ is identifiable. 
\end{proof}

\subsection{Proofs in Appendix}
\subsubsection{Proofs in \cref{sec: app-existing}}
\begin{proof}[Proof of \cref{lemma: horizontal-vertical-form}]
The first conclusion  is directly implied by \cref{lemma: bridge-prelim} so here we focus on proving the second conclusion. Note that 
\begin{align*}
 &Y_{\C, t}(0) =  \mb{U}_{\C}V_t + \epsilon_{\C, t}, \\
 &\frac{1}{N_1}\sum_{i \in \T} Y_{i, t}\prns{0} = \frac{1}{N_1}\sum_{i \in \T} U_i^\top V_t + \frac{1}{N_1}\sum_{i \in \T} \epsilon_{i, t}.
\end{align*}
Thus for any $w^* \in \R{|\C|}$ such that $\mb{U}^\top_\C w^* = \frac{1}{N_1}\sum_{i \in \T} U_i$,
\begin{align*}
\frac{1}{N_1}\sum_{i \in \T} Y_{i, t}\prns{0}
    &=  w^{*\top} Y_{\C, t}\prns{0} + \prns{\frac{1}{N_1}\sum_{i \in \T} U_i - w^{*\top} \mathbf U_{\C}}V_t + \prns{{\frac{1}{N_1}\sum_{i \in \T} \epsilon_{i, t}} - w^{*\top} \epsilon_{\C,t}} \\
    &= w^{*\top} Y_{\C, t} +  {\frac{1}{N_1}\sum_{i \in \T} \epsilon_{i, t}} - w^{*\top} \epsilon_{\C,t}.
\end{align*}
\end{proof}

\begin{proof}[Proof for \cref{lemma: horizontal-bias}]
First, note that as $N_0 \to \infty$,
\begin{align*}
 \frac{1}{N_0}\sum_{i\in\C}{Y}_{i, \pre}{Y}^\top_{i, \pre} 
    &\to \Eb{{Y}_{i, \pre}{Y}^\top_{i, \pre} \mid A_i = 0} \\
    &= \Eb{\prns{\mb{V}_{\pre} U_i + \epsilon_{i, \pre}}\prns{\mb{V}_{\pre} U_i + \epsilon_{i, \pre}}^\top \mid A_i = 0} = \mb{V}_{\pre}\Sigma_{U \mid 0} \mb{V}_{\pre}^\top + \sigma_\epsilon^2 I, \\
 \frac{1}{N_0}\sum_{i\in\C}{Y}_{i, \pre}\epsilon_{i, \pre}^\top    &\to \Eb{{Y}_{i, \pre}\epsilon_{i, \pre}^\top \mid A_i = 0} = \Eb{\prns{\mb{V}_{\pre} U_i + \epsilon_{i, \pre}}{\epsilon_{i, \pre}}^\top \mid A_i  = 0} = \sigma_\epsilon^2 I.
 \end{align*}
Next, we fix an arbitrary $\theta^*_1$ such that $\mb{V}^\top_{\pre}\theta^*_1 = V_{0}$. So
\[
\gamma^S = \frac{1}{N_1}\sum_{i \in \T}\prns{\theta^{*\top}_1 Y_{i, \pre} + \epsilon_{i, 0} - \theta^{*\top}_1 \epsilon_{i, \pre}}.
\]
It follows that 
\begin{align*} 
& Y_{i, 0}\prns{0} = Y_{i, \pre}^\top\theta^*_1 + \epsilon_{i, 0} - \theta_1^{*\top}\epsilon_{i, \pre}, ~~ \forall i, \\
&\hat \gamma_{\HR} - \gamma^* 
    =  \frac{1}{N_1}\sum_{i \in \T} Y_{i, \pre}^\top\prns{\hat\theta_{\HR} - \theta^*_1} - \overline \epsilon_{\T, 0} + \theta^{*\top}_1 \overline \epsilon_{\T, \pre}.
\end{align*}
Further note that 
\begin{align*}
 \hat \theta_{\HR} - \theta^*_1
    &= \prns{\frac{1}{N_0}\sum_{i\in\C}{Y}_{i, \pre}{Y}^\top_{i, \pre}}^{-1}\frac{1}{N_0}\sum_{i\in\C}{Y}_{i, \pre}Y_{i, 0} - \theta^*_1  \\
    &= \prns{\frac{1}{N_0}\sum_{i\in\C}{Y}_{i, \pre}{Y}^\top_{i, \pre}}^{-1}\frac{1}{N_0}\sum_{i\in\C}{Y}_{i, \pre}\prns{Y^\top_{i, \pre}\theta^*_1 + \epsilon_{i, 0} - \theta^{*\top}_1 \epsilon_{i, \pre}}  - \theta^*_1 \\
    &= \prns{\frac{1}{N_0}\sum_{i\in\C}{Y}_{i, \pre}{Y}^\top_{i, \pre}}^{-1}\frac{1}{N_0}\sum_{i\in\C}{Y}_{i, \pre}\prns{\epsilon_{i, 0} - \theta^{*\top}_1 \epsilon_{i, \pre}} \\
    &=-{\prns{\Eb{{Y}_{i, \pre}{Y}^\top_{i, \pre}\mid A_i = 0}}^{-1}\Eb{{Y}_{i, \pre}{\epsilon^\top_{i, \pre}}\mid A_i = 0} }\theta^*_1 + \mathcal{O}_p\prns{\frac{1}{\sqrt{N_0}}} \\
    &= \prns{\frac{1}{\sigma^2_\epsilon}\mb{V}_{\pre}\Sigma_U \mb{V}_{\pre}^\top + I }^{-1}\theta^*_1 + \mathcal{O}_p\prns{\frac{1}{\sqrt{N_0}}}.
 \end{align*}
 Thus as $N_0 \to \infty$, 
 \begin{align*}
 \hat \gamma_{\HR} - \overline{\gamma}^*  
    &=  \frac{1}{N_1}\sum_{i \in \T} Y_{i, \pre}^\top\prns{\hat\theta_{\HR} - \theta^*_1}  - \overline \epsilon_{\T, 0} + \theta^{*\top}_1 \overline \epsilon_{\T, \pre} \\
    &\to - \prns{\overline U_{\T}\mb{V}_{\pre}^\top + \overline\epsilon_{\T, \pre}^\top}\prns{\frac{1}{\sigma^2_\epsilon}\mb{V}_{\pre}\Sigma_{U\mid 0} \mb{V}_{\pre}^\top + I }^{-1}\theta^*_1 + \theta^{*\top}_1 \overline \epsilon_{\T, \pre}- \overline \epsilon_{\T, 0}  \\
    &= -\overline U_{\T}\prns{\frac{1}{\sigma^2_\epsilon}\mb{V}_{\pre}^\top \mb{V}_{\pre}\Sigma_{U\mid 0} + I }^{-1} \mb{V}_{\pre}^\top\theta^*_1 +  \overline\epsilon_{\T, \pre}^\top \bracks{I - \prns{\frac{1}{\sigma^2_\epsilon}\mb{V}_{\pre}\Sigma_{U\mid 0} \mb{V}_{\pre}^\top + I }^{-1}}\theta^*_1 - \overline \epsilon_{\T, 0}\\
    &= -\overline U_{\T}\prns{\frac{1}{\sigma^2_\epsilon}\mb{V}_{\pre}^\top \mb{V}_{\pre}\Sigma_{U\mid 0} + I }^{-1} \mb{V}_{\pre}^\top\theta^*_1 \\
    &\qquad +  \overline\epsilon_{\T, \pre}^\top \mb{V}_{\pre}\prns{\sigma_\epsilon^2 \Sigma_{U\mid 0}^{-1} + \mb{V}_{\pre}^\top\mb{V}_{\pre}}^{-1}\mb{V}_{\pre}^\top\theta_1^* - \overline \epsilon_{\T, 0} \tag{Woodbury Identity} \\
    &= -\overline U_{\T}\prns{\frac{1}{\sigma^2_\epsilon}\mb{V}_{\pre}^\top \mb{V}_{\pre}\Sigma_{U\mid 0} + I }^{-1} V_0 +  \overline\epsilon_{\T, \pre}^\top \mb{V}_{\pre}\prns{\sigma_\epsilon^2 \Sigma_{U\mid 0}^{-1} + \mb{V}_{\pre}^\top\mb{V}_{\pre}}^{-1}V_0  - \overline \epsilon_{\T, 0}\\
    &= \mathcal B_{\HR} + \mathcal V_{\HR}
 \end{align*}
Finally, note that 
\begin{align*}
\sigma_{\min}\prns{\frac{1}{\sigma^2_\epsilon}\mb{V}_{\pre}^\top \mb{V}_{\pre}\Sigma_{U\mid 0} + I} 
    &\ge \sigma_{\min}\prns{\frac{1}{\sigma^2_\epsilon}\mb{V}_{\pre}^\top \mb{V}_{\pre}\Sigma_{U\mid 0}} -\sigma_{\max}\prns{-I} \\
    &\ge \frac{1}{\sigma^2_\epsilon}\sigma_{\min}\prns{\Sigma_{U\mid 0}}\sigma_{\min}^2\prns{\mb{V}_{\pre}} - 1.
\end{align*}
Thus 
\begin{align*}
|\mathcal B_\HR| 
    &\le \|\overline U_{\T}\| \|V_0\|/\sigma_{\min}\prns{\frac{1}{\sigma^2_\epsilon}\mb{V}_{\pre}^\top \mb{V}_{\pre}\Sigma_{U\mid 0} + I} \\
    &\le \frac{\sigma^2_\epsilon}{\sigma_{\min}\prns{\Sigma_{U\mid 0}}\sigma_{\min}^2\prns{\mb{V}_{\pre}} - \sigma^2_\epsilon}\|\overline U_{\T}\|\|V_0\|,
\end{align*}
and 
\begin{align*}
\abs{\mathcal V_{\HR} + \overline \epsilon_{\T, 0}}     &\le  \left\|\mb{V}_{\pre}\prns{\sigma_\epsilon^2 \Sigma_{U\mid 0}^{-1} + \mb{V}_{\pre}^\top\mb{V}_{\pre}}^{-1}\right\|\|\overline\epsilon_{\T, \pre}^\top\|\|V_0\| \\
    &\le \frac{\sigma_{\max}\prns{\Sigma_{U\mid 0}}\sigma_{\min}\prns{\mb{V}_{\pre}}}{\sigma_{\max}\prns{\Sigma_{U\mid 0}}\sigma^2_{\min}\prns{\mb{V}_{\pre}} - \sigma_{\epsilon}^2}\|\overline\epsilon_{\T, \pre}^\top\|\|V_0\|, 
\end{align*}
where the last inequality follows from \cref{lemma: svd-bound}.
\end{proof}

\begin{proof}[Proof for \cref{lemma: vertical-bias}]
Note that 
 \begin{align*}
 \hat w_{\VR} = \prns{\frac{1}{T_0}\sum_{t = -T_0}^{-1}Y_{\C, t}Y_{\C, t}^\top}^{-1}\frac{1}{T_0}\sum_{t = -T_0}^{-1}{Y}_{\C, t}\prns{\frac{1}{N_1}\sum_{i\in \T} Y_{\T, t}}
 \end{align*}

 Before proving the conclusion, note that conditionally on $\mb{U}_\C$, when $T_0 \to \infty$, we have 
 \begin{align*}
 \frac{1}{T_0}\sum_{t = -T_0}^{-1}Y_{\C, t}Y_{\C, t}^\top
    &= \frac{1}{T_0}\sum_{t = -T_0}^{-1}\prns{\mb{U}_\C V_t + \epsilon_{\C, t}}\prns{\mb{U}_\C V_t + \epsilon_{\C, t}}^\top \to \mb{U}_\C \overline V^\otimes \mb{U}_\C^\top + \sigma^2_\epsilon I \\
 \frac{1}{T_0}\sum_{t = -T_0}^{-1}{Y}_{\C, t}\prns{\frac{1}{N_1}\sum_{i\in \T} Y_{\T, t}}
    &= \frac{1}{T_0}\sum_{t = -T_0}^{-1}\prns{\mb{U}_\C V_t + \epsilon_{\C, t}}\prns{V_t^\top\overline{{U}}_\T + \overline\epsilon_{\T, t}} \to \mb{U}_\C \overline V^\otimes\overline{{U}}_\T \\
 \frac{1}{T_0}\sum_{t = -T_0}^{-1}{Y}_{\C, t}\epsilon_{\C,t}^\top 
    &= \frac{1}{T_0}\sum_{t = -T_0}^{-1}\prns{\mb{U}_\C V_t + \epsilon_{\C, t}}\epsilon_{\C,t}^\top \to \sigma^2_\epsilon I . 
 \end{align*}
Now fix an arbitrary $w \in \R{|\C|}$ such that $\mb{U}^\top_\C w = \overline U_\T$. Thus 
\begin{align*}
&\frac{1}{N_1}\sum_{i \in \T}Y_{i, t}\prns{0} = Y_{\C, t}^\top w + \overline{\epsilon}_{\T, t} - \epsilon_{\C, t}^\top w,  \forall t.
\end{align*}
It follows that 
 \begin{align*}
 &\hat \gamma_{\VR} = Y_{\C, 0}^\top \hat w_{\VR} \\   
    =& Y_{\C, 0}^\top\prns{\frac{1}{T_0}\sum_{t = -T_0}^{-1}Y_{\C, t}Y_{\C, t}^\top}^{-1}\frac{1}{T_0}\sum_{t = -T_0}^{-1}{Y}_{\C, t}\prns{\frac{1}{N_1}\sum_{i \in \T}Y_{i, t}}  \\
    =& Y_{\C, 0}^\top\prns{\frac{1}{T_0}\sum_{t = -T_0}^{-1}Y_{\C, t}Y_{\C, t}^\top}^{-1}\frac{1}{T_0}\sum_{t = -T_0}^{-1}{Y}_{\C, t}\prns{Y_{\C, t}^\top w +  {\overline \epsilon_{\T, t} - w^\top \epsilon_{\C,t}}} \\
    =& \prns{\mb{U}_\C V_0 + \epsilon_{\C, 0}}^\top w + \prns{\mb{U}_\C V_0 + \epsilon_{\C, 0}}^\top\prns{\frac{1}{T_0}\sum_{t = -T_0}^{-1}Y_{\C, t}Y_{\C, t}^\top}^{-1}\frac{1}{T_0}\sum_{t = -T_0}^{-1}\prns{\mb{U}_\C V_t + \epsilon_{\C, t}}\prns{{\overline \epsilon_{\T, t} - w^\top \epsilon_{\C,t}}}
 \end{align*} 
 When $T_0 \to \infty$, we have 
 \begin{align*}
 \hat \gamma_{\VR}  &\to V_0^\top \mb{U}_\C^\top w - V_0^\top \mb{U}_\C^\top\prns{\mb{U}_\C \overline V^\otimes \mb{U}_\C^\top + \sigma^2_\epsilon I}^{-1}\sigma^2_\epsilon w+ \epsilon_{\C, 0}^\top\bracks{I - \sigma^2_\epsilon\prns{\mb{U}_\C \overline V^\otimes \mb{U}_\C^\top + \sigma^2_\epsilon I}^{-1}}w\\
    &= \gamma^S - \overline{\epsilon}_{\T, 0} - V_0^\top \mb{U}_\C^\top\prns{\mb{U}_\C \overline V^\otimes \mb{U}_\C^\top + \sigma^2_\epsilon I}^{-1}\sigma^2_\epsilon w+ \epsilon_{\C, 0}^\top\bracks{I - \sigma^2_\epsilon\prns{\mb{U}_\C \overline V^\otimes \mb{U}_\C^\top + \sigma^2_\epsilon I}^{-1}}w   \\
    &= \gamma^S + \mathcal B_{\VR} + \mathcal V_{\VR}.
 \end{align*}
Now we analyze $\mathcal B_{\VR}$ and $\mathcal V_{\VR}$. 
\begin{align*}
\mathcal B_{\VR} = - V_0^\top \mb{U}_\C^\top\prns{\mb{U}_\C \overline V^\otimes \mb{U}_\C^\top + \sigma^2_\epsilon I}^{-1}\sigma^2_\epsilon w
    &= -V_0^\top \prns{\frac{1}{\sigma^2_\epsilon}\mb{U}_\C^\top\mb{U}_\C \overline V^\otimes  + I}^{-1}\mb{U}_\C^\top w\\ 
    &= - V_0^\top \prns{\frac{1}{\sigma^2_\epsilon}\mb{U}_\C^\top\mb{U}_\C \overline V^\otimes  + I}^{-1}\overline{U}_\T,
\end{align*}
and 
\begin{align*}
\mathcal V_{\VR} 
    &= \epsilon_{\C, 0}^\top\bracks{I -\prns{\frac{1}{ \sigma^2_\epsilon}\mb{U}_\C \overline V^\otimes \mb{U}_\C^\top + I}^{-1}}w - \overline{\epsilon}_{\T, 0} \\
    &=  \epsilon_{\C, 0}^\top \mb{U}_\C\prns{\mb{U}_\C^\top \mb{U}_\C + \sigma^2_\epsilon {\overline V^{\otimes-1}}}^{-1}\mb{U}_\C^\top w - \overline{\epsilon}_{\T, 0} \\
    &= \epsilon_{\C, 0}^\top \mb{U}_\C\prns{\mb{U}_\C^\top \mb{U}_\C + \sigma^2_\epsilon {\overline V^{\otimes-1}}}^{-1}\overline{U}_\T  - \overline{\epsilon}_{\T, 0}. 
\end{align*}
Moreover, according to \cref{lemma: svd-bound},
\begin{align*}
\abs{\mathcal B_{\VR}} 
    &\le 
    \frac{1}{\abs{\C}}\|V_0\|\|\overline U_\T\|\left\|\prns{\frac{1}{\sigma_\epsilon^2 |\C|}\sum_{i \in \C}U_iU_i^\top\overline V^{\otimes} + I}^{-1}\right\|\\
    &\le \frac{1}{\abs{\C}}\frac{\sigma_\epsilon^2}{\sigma_{\min}\prns{\sum_{i \in \C}U_iU_i^\top/\abs{\C}}\sigma_{\min}\prns{\overline V^\otimes} - \sigma_{\epsilon}^2/\abs{\C}}\|V_0\|\|\overline U_\T\|
\end{align*}
and similarly 
\begin{align*}
\abs{\mathcal V_{\VR} + \overline\epsilon_{\T, 0}}  &\le \prns{\sum_{i \in \C}\epsilon_{i, 0}U_{i}/\abs{\C}}\frac{\sigma_\epsilon^2}{\sigma_{\min}\prns{\sum_{i \in \C}U_iU_i^\top/\abs{\C}}\sigma_{\min}\prns{\overline V^{\otimes}} - \sigma_{\epsilon}^2/\abs{\C}}\|\|\overline V^{\otimes}\|\|\overline U_\T\|.
\end{align*}
\end{proof}

\begin{proof}[Proof of \cref{lemma: factor}]
This conclusion follows from Theorem 2 in \cite{xiong2019large} and we only need to verify their assumptions S1-S3. 
It is easy to show that Assumption S1 in \cite{xiong2019large} is satisfied with $q_{ij}, q_{ij, kl} = 1$ in their notations and Assumption S2 in \cite{xiong2019large} is satisfied as well. Finally, Assumption S3 in \cite{xiong2019large} is satisfied with $\omega, \omega_j, \omega_{jj} = 1$. Then Theorem 2 statement 3 in  \cite{xiong2019large} implies \cref{eq: factor-1}, which in turn implies \Cref{eq: factor-2}.
\end{proof}

\subsubsection{Proofs in \cref{sec: app-opt}}
\begin{proof}[Proof of \cref{corollary: mu-asymp2}]
Note that $\tilde\gamma$ satisfies 
\begin{align*}
\tilde\gamma-\gamma^* 
    &= -\frac{1}{\En{A}}\braces{\En{A\prns{\tilde W^\top\tilde\theta - \gamma^*}} - \W_{21, N}\W_{11, N}^{-1}\En{\prns{1-A}\tilde Z\prns{Y_0 - \tilde W^\top\tilde\theta}}} \\
    &= -\frac{1}{\En{A}}\bigg\{\En{A\prns{\tilde W^\top\theta^*_{\min} - \gamma^*}} - \En{\W_{21, N}\W_{11, N}^{-1}\prns{1-A}\tilde Z\prns{Y_0 - \tilde W^\top\theta^*_{\min}}}\\
    &\qquad\qquad\qquad\qquad\qquad\qquad + \En{A\tilde W^\top + \W_{21, N}\W_{11, N}^{-1}\prns{1-A}\tilde Z\tilde W^\top}\prns{\tilde\theta-\theta_{\min}^*} \bigg\} \\
    &= -\frac{1}{\Eb{A}}\bigg\{\En{A\prns{\tilde W^\top\theta^*_{\min} - \gamma^*}} - \En{\W_{21, \infty}\W_{11, \infty}^{-1}\prns{1-A}\tilde Z\prns{Y_0 - \tilde W^\top\theta^*_{\min}}}\\
    &\qquad\qquad\qquad\qquad + \Eb{A\tilde W^\top + \W_{21, \infty}\W_{11, \infty}^{-1}\prns{1-A}\tilde Z\tilde W^\top}\prns{\tilde\theta-\theta_{\min}^*} \bigg\} + \smallO_p\prns{\frac{1}{\sqrt{N}}}.
\end{align*}
By following the proof of \cref{lemma: two-key-facts} statement 2, it is easy to show that there exists $\phi_{ZW}$ such that $\Eb{\prns{1-A}\tilde Z\tilde W^\top} = K^\top K \phi_{ZW}$. Then by following the proof of \cref{thm: mu-asymp}, we can show that 
\begin{align*}
&\Eb{A\tilde W^\top + \W_{21, \infty}\W_{11, \infty}^{-1}\prns{1-A}\tilde Z\tilde W^\top}\prns{\tilde\theta-\theta_{\min}^*} \\
&=  \prns{\Eb{A \tilde W^\top} + \W_{21, \infty}\W_{11, \infty}^{-1}\Eb{\prns{1-A}\tilde Z\tilde W^\top}} \\
&\qquad\times \braces{{\Eb{(1-A)\tilde W Z^\top}} \mathcal{W}^{-1}_{11, \infty}\Eb{(1-A)\tilde Z \tilde W^\top}}^{+}{\Eb{(1-A)\tilde W \tilde Z^\top}}\mathcal{W}^{-1}_{11, \infty}\\
&\qquad \times \En{\prns{1-A}\tilde Z\prns{Y_0 - \tilde W^\top\theta_{\min}^*}} + \mathcal{O}_p\prns{\lambda_N+ \frac{1}{\sqrt{\lambda_N}N}}.
\end{align*}
This proves the asserted conclusion of this theorem. 
\end{proof}

\begin{proof}[Proof of \cref{thm: opt-general}]
First, by following the proof of \cref{thm: CI}, we can show that for any positive definite matrix $\W_{\infty}$, when $N \to \infty$,
\begin{align*}
\tilde\sigma_N^2\prns{\W_{\infty}} \coloneqq \Eb{\tilde\psi_N\prns{O; \theta_{\min}^*, \gamma^*, \W_{\infty}}} \to \tilde\sigma^2\prns{\W_{\infty}}, 
\end{align*}
where 
\begin{align*}
&\tilde\psi_N\prns{O_i; \theta^*_{\min}, \gamma^*, \W_{\infty}}  = -\frac{1}{\Eb{A}}\braces{{A_i\prns{\gamma^* -\tilde W^\top_i\theta_{\min}^*}} + \tilde\Psi_N(\mathcal{W}_{\infty})\prns{1-A_i}\tilde Z_i\prns{Y_{i, 0}  - \tilde W_i^\top\theta_{\min}^{*}}},\\
&\tilde\Psi_N(\mathcal{W}_{\infty}) = \prns{\Eb{A \tilde W^\top} + \W_{21, \infty}\W_{11, \infty}^{-1}\Eb{\prns{1-A}\tilde Z\tilde W^\top}} \\
&~\times \braces{{\Eb{(1-A)\tilde W Z^\top}} \mathcal{W}^{-1}_{11, \infty}\Eb{(1-A)\tilde Z \tilde W^\top}+\lambda_N I}^{-1}{\Eb{(1-A)\tilde W \tilde Z^\top}}\mathcal{W}^{-1}_{11, \infty}- \W_{21, \infty}\W_{11, \infty}^{-1}.
\end{align*}
Next, we show that 
\begin{align}\label{eq: opt-eq-1}
\lim_{N \to \infty} {\tilde\sigma_N^2\prns{\W_{\infty}} - \tilde\sigma_N^2\prns{\Sigma}} \ge 0,
\end{align}
which in turn proves the asserted conclusion. 

Note that 
\begin{align*}
&-\Eb{A}\tilde\psi_N\prns{O_i; \theta^*_{\min}, \gamma^*, \W_{\infty}}\\
    & = -\Eb{A}\psi_N\prns{O_i; \theta^*_{\min}, \gamma^*, \Sigma} + \underbrace{\prns{\tilde\Psi_N(\mathcal{W}_{\infty}) - \tilde\Psi_N(\Omega)}\prns{1-A_i}\tilde Z_i\prns{Y_{i, 0}  - \tilde W_i^\top\theta_{\min}^{*}}}_{\text{Rem}_{N, i}}
\end{align*}
This means that 
\begin{align*}
\prns{\Eb{A}}^2{\tilde\sigma_N^2\prns{\W_{\infty}}} = \prns{\Eb{A}}^2 \tilde\sigma_N^2\prns{\Sigma} + \Eb{\text{Rem}^2_{N, i}} - 2\Eb{A}\Eb{\tilde\psi_N\prns{O_i; \theta^*_{\min}, \gamma^*, \Sigma} \text{Rem}_{N, i}}.
\end{align*}
We will prove \cref{eq: opt-eq-1} by showing that $\lim_{N\to\infty} \abs{\Eb{\tilde\psi_N\prns{O_i; \theta^*_{\min}, \gamma^*, \Sigma} \text{Rem}_{N, i}}} = 0$.

It is easy to verify that 
\begin{align*}
&\Eb{\tilde\psi_N\prns{O_i; \theta^*_{\min}, \gamma^*, \Sigma} \text{Rem}_{N, i}} = \Eb{\prns{\tilde\Phi_N\prns{\Sigma} + \Sigma_{gm}\Sigma_{m}^{-1}}\Sigma_m\prns{\tilde\Phi_N\prns{\W} - \tilde\Phi_N\prns{\Sigma}}^\top} \\
=& \prns{\Eb{A\tilde W^\top} + \Sigma_{gm}\Sigma_{m}^{-1}K}\prns{K^\top\Sigma_m^{-1}K + \lambda_N I}^{-1}\bigg\{K^\top\prns{\Sigma_m^{-1}\Sigma_{mg} - \W_{11, \infty}^{-1}\W_{12, \infty}} \\
&\qquad\qquad\qquad\qquad\qquad + K^\top\W_{11, \infty}^{-1}K\prns{K^\top\W_{11, \infty}^{-1}K + \lambda_N I}^{-1}\prns{K^\top\W_{11, \infty}^{-1}\W_{12, \infty} + \Eb{A\tilde W}} \\
&\qquad\qquad\qquad\qquad\qquad - K^\top\Sigma_{m}^{-1}K\prns{K^\top\Sigma_{m}^{-1}K + \lambda_N I}^{-1}\prns{K^\top\Sigma_m^{-1}\Sigma_{mg} + \Eb{A\tilde W}} \bigg\}.
\end{align*}
We note that 
\begin{align*}
&\prns{\Eb{A\tilde W^\top} + \Sigma_{gm}\Sigma_{m}^{-1}K}\prns{K^\top\Sigma_m^{-1}K + \lambda_N I}^{-1}K^\top\W_{11, \infty}^{-1}K\\
&\qquad\qquad\qquad\qquad \times \prns{K^\top\W_{11, \infty}^{-1}K + \lambda_N I}^{-1}\prns{K^\top\W_{11, \infty}^{-1}\W_{12, \infty} + \Eb{A\tilde W}}  \\
=& \prns{\Eb{A\tilde W^\top} + \Sigma_{gm}\Sigma_{m}^{-1}K}\prns{K^\top\Sigma_m^{-1}K + \lambda_N I}^{-1}\prns{K^\top\W_{11, \infty}^{-1}\W_{12, \infty} + \Eb{A\tilde W}} \\
-& \lambda_N  \prns{\Eb{A\tilde W^\top} + \Sigma_{gm}\Sigma_{m}^{-1}K}\prns{K^\top\Sigma_m^{-1}K + \lambda_N I}^{-2}\prns{K^\top\W_{11, \infty}^{-1}\W_{12, \infty} + \Eb{A\tilde W}},
\end{align*}
and similarly
\begin{align*}
&\prns{\Eb{A\tilde W^\top} + \Sigma_{gm}\Sigma_{m}^{-1}K}\prns{K^\top\Sigma_m^{-1}K + \lambda_N I}^{-1}K^\top\Sigma_{m}^{-1}K \\
&\qquad\qquad\qquad\qquad \times \prns{K^\top\Sigma_{m}^{-1}K + \lambda_N I}^{-1}\prns{K^\top\Sigma_m^{-1}\Sigma_{mg} + \Eb{A\tilde W}} \\
=& \prns{\Eb{A\tilde W^\top} + \Sigma_{gm}\Sigma_{m}^{-1}K}\prns{K^\top\Sigma_m^{-1}K + \lambda_N I}^{-1}\prns{K^\top\Sigma_m^{-1}\Sigma_{mg} + \Eb{A\tilde W}}\\
-&\lambda_N \prns{\Eb{A\tilde W^\top} + \Sigma_{gm}\Sigma_{m}^{-1}K}\prns{K^\top\Sigma_m^{-1}K + \lambda_N I}^{-2}\prns{K^\top\Sigma_m^{-1}\Sigma_{mg} + \Eb{A\tilde W}}.
 \end{align*} 
 If follows that 
 \begin{align*}
 &\Eb{\tilde\psi_N\prns{O_i; \theta^*_{\min}, \gamma^*, \Sigma} \text{Rem}_{N, i}}  \\
 =& \lambda_N  \prns{\Eb{A\tilde W^\top} + \Sigma_{gm}\Sigma_{m}^{-1}K}\prns{K^\top\Sigma_m^{-1}K + \lambda_N I}^{-2}K^\top\prns{\Sigma_m^{-1}\Sigma_{mg} - \W_{11, \infty}^{-1}\W_{12, \infty}}.
 \end{align*}
 By following the proof of \cref{lemma: support}, we know that 
 \begin{align*}
 \norm{\prns{\phi_W^\top K^\top + \Sigma_{gm}\Sigma_{m}^{-1}}K \prns{K^\top\Sigma_m^{-1}K + \lambda_N I}^{-2}K^\top} = \mathcal{O}(1).
 \end{align*}
 Therefore, 
 \begin{align*}
 \abs{\Eb{\tilde\psi_N\prns{O_i; \theta^*_{\min}, \gamma^*, \Sigma} \text{Rem}_{N, i}}} = \mathcal{O}\prns{\lambda_N} \to 0,
 \end{align*}
 which concludes the proof of \cref{eq: opt-eq-1}.
\end{proof}

\subsubsection{Proofs in \cref{sec: app-reg}}
\begin{proof}[Proof of \cref{lemma: unique-target}]
According to \cref{eq: bridge-theta}, 
\begin{align*}
\Theta^* = \braces{\mb{V}_\pre^\top \theta_1^* = V_0, \theta_2^* = {b_0 - \mb{B}_\pre^\top\theta^*_1}}.
\end{align*}
Partition matrix $M$ into the following form:
\begin{align*}
\begin{bmatrix}
M_{11} & M_{12} \\
M_{21} & M_{22}
\end{bmatrix},
\end{align*}
where $M_{11}, M_{12}, M_{21}, M_{22}$ are $T_0 \times T_0$, $T_0 \times d$, $d \times T_0$, $d \times d$ matrices respectively. Then for any $\theta \in \Theta^*$, we have 
\begin{align*}
\theta^\top M\theta 
  &= \theta_1^\top \begin{bmatrix}
      I_{T_0 \times T_0} & -\mb{B}_{\pre} 
      \end{bmatrix}
      M 
      \begin{bmatrix}
      I_{T_0 \times T_0} \\
       -\mb{B}_{\pre}^\top 
      \end{bmatrix}\theta_1 \\
  &+ \theta_1^\top M_{12}b_0 + b_0^\top M_{21}\theta_1 - \theta^\top_1 \mb{B}_\pre M_{22}b_0 - b_0^\top M_{22}\mb{B}_\pre^\top\theta_1.
\end{align*}
Under the asserted conclusion, the objective above is strictly convex in $\theta_1$, so $\theta^*_{M}$ is uniquely defined. 
\end{proof}

\begin{proof}[Proof of \cref{prop: asymp-variance}]
Recall that 
\begin{align*}
&\psi\prns{O_i; \theta_{\min}^*, \gamma^*, \mathcal{W}_{m, \infty}} = -\frac{1}{\Eb{A}} 
\braces{g\prns{O_i; \theta^*_{\min}, \gamma^*}+ \Psi\prns{\mathcal{W}_{m, \infty}} m\prns{O_i ; \theta^*_{\min}}}, \\
&\Psi(\mathcal{W}_{m, \infty}) = \Eb{A \tilde W^\top}\braces{{\Eb{(1-A)\tilde W Z^\top}} \mathcal{W}_{m, \infty}\Eb{(1-A)\tilde Z \tilde W^\top}}^{+}{\Eb{(1-A)\tilde W \tilde Z^\top}}\mathcal{W}_{m, \infty},
\end{align*}
where 
\begin{align*}
    &g\prns{O; \theta, \gamma} \coloneqq A\prns{\theta^{\top}_1 Y_{\pre} + \theta_2 X - \gamma},  \\
    &m\prns{O; \theta} \coloneqq \prns{1-A}\prns{Y_{0}  -  
    \prns{\theta^{\top}_1 Y_{\pre} + \theta_2 X}
    }
        \begin{bmatrix}
        Y_{\post} \\
        X
        \end{bmatrix}.
\end{align*}
It follows that 
\begin{align*}
\Eb{g^2\prns{O; \theta^*_M, \gamma^*}} 
  &= \Prb{A=1}\Eb{\prns{\theta^{*\top}_{M, 1} Y_{\pre} + \theta^*_{M 2} X - \gamma^*}^2 \mid A = 1} \\
  &= \Eb{\prns{V_0^\top\prns{U - \Eb{U \mid A = 1}} + b_0^\top\prns{X - \Eb{X \mid A = 1}}}^2 \mid A = 1} \\
  &= \begin{bmatrix}
        V_0 ^\top  & b_0^\top
        \end{bmatrix}
        \begin{bmatrix}
        \op{Cov}\prns{U, U \mid A = 1} & \op{Cov}\prns{U, X \mid A = 1} \\
        \op{Cov}\prns{X, U \mid A = 1} & \op{Cov}\prns{X, X \mid A = 1}
        \end{bmatrix}
    \begin{bmatrix}
        V_0  \\ b_0 
        \end{bmatrix},
\end{align*}
and 
\begin{align*}
&\Eb{m\prns{O; \theta_M^*}m^\top\prns{O; \theta_M^*}} \\
  =& \Eb{\prns{1-A}\prns{\epsilon_{0} - \theta_{M, 1}^{*\top}\epsilon_{\pre}}^2 \begin{bmatrix}
        Y_{\post} \\
        X
        \end{bmatrix}\begin{bmatrix}
        Y_{\post}^\top &
        X^\top
        \end{bmatrix}} \\
  =& \Prb{A = 0}\prns{\Eb{\epsilon_0^2} + \theta_{M, 1}^{*\top}\Eb{\epsilon_{\pre}\epsilon_{\pre}^\top}\theta_{M, 1}^{*}} \\
  \times & \begin{bmatrix}
  \mb{V}_{\post} & \mb{B}_{\post} \\
  \mb{0} & \mb{I}
  \end{bmatrix}
       \begin{bmatrix}
        \Eb{UU^\top \mid A = 0} & \Eb{UX^\top \mid A = 0} \\
        \Eb{XU^\top \mid A = 0} & \Eb{XX^\top \mid A = 0}
        \end{bmatrix}
        \begin{bmatrix}
  \mb{V}_{\post}^\top &  \mb{0} \\
  \mb{B}_{\post}^\top  & \mb{I}
  \end{bmatrix}
\end{align*}
The conclusion immediately follows from the fact that 
\begin{align*}
\Eb{g\prns{O; \theta^*_M, \gamma^*}m\prns{O; \theta^*_M}} = 0.
\end{align*}
\end{proof}

\subsubsection{Proofs in \cref{sec: app-dynamic}}
\begin{proof}[Proof of \cref{lemma: time-varying-bridge-general}]
By recursion, we have that for any positive integer $t$ and nonnegative integer $m$, 
\begin{align*}
&U_{t} =  \Gamma_{t-1}U_{t-1} + \eta_{t-1} \\
&U_{t} =  \Gamma_{t-1}\Gamma_{t-2}U_{t-2} + \Gamma_{t-1}\eta_{t-2}+\eta_{t-1} \\
&U_{t} =  \Gamma_{t-1}\Gamma_{t-2}\Gamma_{t-3}U_{t-3} + \Gamma_{t-1}\Gamma_{t-2}\eta_{t-3} + \Gamma_{t-1}\eta_{t-2}+\eta_{t-1} \\
&\dots \\
&U_{t} = \Gamma_{(t-1):\prns{t-m}}U_{t-m} + \sum_{k=1}^m\Gamma_{\prns{t-1}:\prns{t-k+1}}\eta_{t-k}.
\end{align*}
Therefore, we have 
\begin{align*}
&Y_{t}\prns{0} = V_t^\top \Gamma_{\prns{t-1}:\prns{-T_0}}U_{-T_0} + V_t^\top \sum_{k=1}^{t + T_0}\Gamma_{\prns{t-1}:\prns{t-k+1}}\eta_{t-k} + b_t^\top X + \epsilon_t, ~~ \forall t, \\
&Y_t\prns{0} = V_t^\top \Gamma_{\prns{t-1}:0} U_0 + V_t^\top \sum_{k=1}^{t}\Gamma_{\prns{t-1}:\prns{t-k+1}}\eta_{t-k} + b_t^\top X + \epsilon_t , ~~ \forall t > 0.
\end{align*}
Since we assume that $\Eb{U_{-T_0} \mid U_0, A = 0, X}$ is linear in $U_0$ and $X$, it must coincide with the projection of $U_{-T_0}$ onto the linear span of $U_0, X$ given $A = 0$:
\begin{align*}
&\Eb{U_{-T_0} \mid A = 0} + 
\begin{bmatrix}
\Eb{U_{-T_0}U_0^\top \mid A = 0} & \Eb{U_{-T_0}X^\top \mid A = 0}
\end{bmatrix}
\\
&\qquad\qquad\qquad\qquad \times \begin{bmatrix}
\Eb{U_{0}U_0^\top \mid A = 0} & \Eb{U_{0}X^\top \mid A = 0} \\
\Eb{XU_0^\top \mid A = 0} & \Eb{X X^\top \mid A = 0}
\end{bmatrix}^{-1}
\begin{bmatrix}
U_0 - \Eb{U_0 \mid A = 0} \\
X - \Eb{X \mid A = 0}
\end{bmatrix},
\end{align*}
which can be written as 
\begin{align*}
\Eb{U_{-T_0} \mid A = 0} 
  &+ \prns{\Eb{U_{-T_0}U_0^\top \mid A = 0}G_{11} + \Eb{U_{-T_0}X^\top \mid A = 0}G_{21}}\prns{U_0 - \Eb{U_0 \mid A = 0}} \\
  &+ \prns{\Eb{U_{-T_0}U_0^\top \mid A = 0}G_{12} + \Eb{U_{-T_0}X^\top \mid A = 0}G_{22}}\prns{X - \Eb{X \mid A = 0}},
\end{align*}
where 
\begin{align*}
\Eb{U_{-T_0}U_0^\top \mid A = 0} = \Eb{U_{-T_0}U_{-T_0}^\top \mid A = 0}\Gamma^\top_{\prns{-1}:\prns{-T_0}}.
\end{align*}
Similarly, $\Eb{\eta_t \mid U_0, A = 0, X}$ for any $t < 0$ is the following:
\begin{align*}
  &\prns{\Eb{\eta_{t}U_0^\top \mid A = 0}G_{11} + \Eb{\eta_t X^\top \mid A = 0}G_{21}}\prns{U_0 - \Eb{U_0 \mid A = 0}} \\
  +& \prns{\Eb{\eta_t U_0^\top \mid A = 0}G_{12} + \Eb{\eta_t X^\top \mid A = 0}G_{22}}\prns{X - \Eb{X \mid A = 0}},
\end{align*}
where 
\begin{align*}
\Eb{\eta_{t}U_0^\top \mid A = 0} = \Sigma_{\eta_t}\Gamma^\top_{\prns{-1}:\prns{t + 1}}.
\end{align*}
For any $\theta_1 \in \R{T_0}$, 
\begin{align*}
\sum_{m= 1}^{T_0}\theta_{1, m}Y_{-m}\prns{0} 
  &= \sum_{m = 1}^{T_0} \theta_{1, m}V_{-m}^\top \Gamma_{\prns{-m-1}:\prns{-T_0}}U_{-T_0} +  \sum_{m = 1}^{T_0} \theta_{1, m} V_{-m}^\top \sum_{k=1}^{-m + T_0}\Gamma_{\prns{-m-1}:\prns{-m-k+1}}\eta_{-m-k} \\
  &+  \sum_{m = 1}^{T_0} \theta_{1, m} b_{-m}^\top X +  \sum_{m = 1}^{T_0} \theta_{1, m} \epsilon_{-m}.
\end{align*}
Thus 
\begin{align*}
&\Eb{\sum_{m= 1}^{T_0}\theta_{1, m}Y_{-m}\prns{0}  \mid U_0 , A = 0, X} \\
=& \bigg[\sum_{m = 1}^{T_0} \theta_{1, m}V_{-m}^\top \Gamma_{\prns{-m-1}:\prns{-T_0}} \prns{\Sigma_{U_{-T_0}}\Gamma^\top_{\prns{-1}:\prns{-T_0}} G_{11} + \Sigma_{U_{-T_0}, X}G_{21}}\\
+& \sum_{m = 1}^{T_0} \theta_{1, m} V_{-m}^\top \sum_{k=1}^{-m + T_0}\Gamma_{\prns{-m-1}:\prns{-m-k+1}}\prns{\Sigma_{\eta_{-m-k}}\Gamma^\top_{\prns{-1}:\prns{-m - k + 1}}G_{11} + \Sigma_{\eta_{-m -k}, X}G_{21}}\bigg] 
\times  \prns{U_0 - \Eb{U_0 \mid A = 0}} \\
+ & \bigg[\sum_{m = 1}^{T_0} \theta_{1, m}V_{-m}^\top \Gamma_{\prns{-m-1}:\prns{-T_0}}\prns{\Sigma_{U_{-T_0}}\Gamma^\top_{\prns{-1}:\prns{-T_0}}G_{12} + \Sigma_{U_{-T_0}, X}G_{22}} \\
+& \sum_{m = 1}^{T_0} \theta_{1, m} V_{-m}^\top \sum_{k=1}^{-m + T_0}\Gamma_{\prns{-m-1}:\prns{-m-k+1}} \prns{\Sigma_{\eta_{-m-k}}\Gamma^\top_{\prns{-1}:\prns{-m - k + 1}}G_{12} + \Sigma_{\eta_{-m-k}, X}G_{22}} \bigg]
\times \prns{X - \Eb{X \mid A = 0}} \\
+& \theta_1^\top \mb{B}_{\pre} X + \sum_{m=1}^{T_0} \theta_{1, m}V_{-m}^\top \Gamma_{\prns{-m-1}:\prns{-T_0}}\Eb{U_{-T_0} \mid A = 0}.
\end{align*}
Therefore, 
\begin{align*}
0 = &\Eb{Y_0\prns{0} - \sum_{m= 1}^{T_0}\theta_{1, m}Y_{-m}\prns{0} - \theta_2^\top X \mid U_0, A = 0, X} \\
=& \Eb{Y_0\prns{0}\mid U_0, A = 0, X} - \Eb{\sum_{m= 1}^{T_0}\theta_{1, m}Y_{-m}\prns{0} + \theta_2^\top X\mid U_0, A = 0, X} \\
=& V_0^\top U_0 + b_0^\top X - \Eb{\sum_{m= 1}^{T_0}\theta_{1, m}Y_{-m}\prns{0} + \theta_2^\top X \mid U_0, A = 0, X} 
\end{align*}
if and only if $\tilde{\mb{V}}_{\pre}\prns{G_{11}, G_{21}} \in \R{T_0 \times r}$ has full column rank $r$, 
where the $t$th row of $\tilde{\mb{V}}_{\pre}\prns{G_{11}, G_{21}}$ is 
\begin{align*}
& V_{t}^\top \Gamma_{\prns{t-1}:\prns{-T_0}} \prns{\Sigma_{U_{-T_0}}\Gamma^\top_{\prns{-1}:\prns{-T_0}} G_{11} + \Sigma_{U_{-T_0}, X}G_{21}} \\
+& V_{t}^\top \sum_{k=1}^{t + T_0}\Gamma_{\prns{t-1}:\prns{t-k+1}}\prns{\Sigma_{\eta_{t-k}}\Gamma^\top_{\prns{-1}:\prns{t - k + 1}}G_{11} + \Sigma_{\eta_{t -k}, X}G_{21}}.
\end{align*}
When this is the case, for any solution $\theta_1^*$ to the linear equation $\tilde{\mb{V}}^\top_{\pre}\prns{G_{11}, G_{21}}\theta_1^* = V_0$, we can set 
\begin{align*}
   \theta_2^* 
      &= b_0 - \mb{B}_\pre^\top\theta_1^* + \tilde{\mb{V}}^\top_{\pre}\prns{G_{12}, G_{22}}\theta_1^* \\
      &- \big(\Eb{U_0^\top \mid A = 0}\tilde{\mb{V}}^\top_{\pre}\prns{G_{11}, G_{21}}\theta_1^* + \Eb{X^\top \mid A = 0}\tilde{\mb{V}}^\top_{\pre}\prns{G_{12}, G_{22}} \\
      &- \sum_{m=1}^{T_0} \theta^*_{1, m}V_{-m}^\top \Gamma_{\prns{-m-1}:\prns{-T_0}}\Eb{U_{-T_0} \mid A = 0}\big)\iota_d,
\end{align*}   
where $\iota_d$ is a $d \times 1$ vector with the first entry being $1$ and all other entries being $0$. It follows that $\theta^* = \prns{\theta_1^*, \theta_2^*}$ satisfies 
\begin{align*}
0 = \Eb{Y_0\prns{0} - \theta_{1}^{*\top}Y_{\pre} - \theta_2^{*\top} X \mid U_0, A = 0, X}.
\end{align*}
Note that given $U_t \perp A \mid X, U_0$ in \cref{assump: time-varying}, we have  $Y_{\pre} \perp A \mid X, U_0$.  So we can follow the proof of \cref{lemma: bridge} to show that $\theta^*$ satisfies \cref{eq: identification}. 
\end{proof}
\end{document}